\documentclass[11pt]{article}

\usepackage{amsfonts}
\usepackage{amssymb}
\usepackage{amstext}
\usepackage{amsmath}
\usepackage{enumerate}
\usepackage{algorithm}
\usepackage{algorithmic}

\usepackage{amsthm}

\usepackage{thmtools}

\usepackage[usenames,dvipsnames,svgnames,table]{xcolor}

\usepackage[pagebackref,letterpaper=true,colorlinks=true,pdfpagemode=UseNone,citecolor=OliveGreen,linkcolor=BrickRed,urlcolor=BrickRed,pdfstartview=FitH]{hyperref}


\newtheorem{theorem}{Theorem}[subsection]
\newtheorem{fact}[theorem]{Fact}
\newtheorem{lemma}[theorem]{Lemma}
\newtheorem{definition}[theorem]{Definition}

\newtheorem{proposition}[theorem]{Proposition}
\newtheorem{claim}[theorem]{Claim}
\newtheorem{remark}[theorem]{Remark}
\newtheorem{procedure}[theorem]{Procedure}

\newcommand{\opt}[0]{{\ensuremath{\sf{opt}}}}

\renewcommand{\d}{\delta}

\numberwithin{equation}{section}
\numberwithin{table}{section}

\renewcommand{\widetilde}{\tilde}

\newcommand{\R}{\ensuremath{\mathbb R}}

\newcommand{\F}{\ensuremath{\mathcal F}}

\newcommand{\E}[1]{{\mathbb{E}}\left[#1\right]}

\newcommand{\poly}{\operatorname{poly}}

\newcommand{\junk}[1]{}

\newcommand{\ol}{\overline}

\def\adj{\operatorname{adj}}

\providecommand{\norm}[1]{\left\lVert#1\right\rVert}

\def\b1{{\bf 1}}
\def\eps{{\epsilon}}

\def\R{\mathbb{R}}

\def\U{{\cal U}}
\def\V{{\cal V}}
\def\v{{\widetilde{v}}}
\def\AA{{A^{\circ 2}}}
\def\AAt{{(A^{(t)})^{\circ 2}}}
\def\BB{{B^{\circ 2}}}

\def\d{{\frac{d}{dt}}}

\def\diag{\operatorname{diag}}

\def\opt{\sf{opt}}

\def\tr{\operatorname{tr}}
\def\per{\operatorname{per}}

\global\long\def\defeq{\triangleq}

\global\long\def\capa{{\rm cap}}
\global\long\def\E{\mathbb{E}}
\global\long\def\P{\mathbb{P}}
\global\long\def\R{\mathbb{R}}

\newcommand{\inner}[2]{\langle #1, #2 \rangle}

\DeclareMathOperator{\rank}{rank}

\DeclareMathOperator{\dist}{{\rm dist}^2}
\DeclareMathOperator{\distance}{{\rm dist}}

\title{The Paulsen Problem, Continuous Operator Scaling,\\ and Smoothed Analysis}

\author{Tsz Chiu Kwok\footnote{University of Waterloo. Supported by NSERC Discovery Grant 2950-120715 and NSERC Accelerator Supplement 2950-120719. Email: \href{mailto:tckwok0@gmail.com}{tckwok0@gmail.com}},~~~~~
Lap Chi Lau\footnote{University of Waterloo. Supported by NSERC Discovery Grant 2950-120715 and NSERC Accelerator Supplement 2950-120719. Email: \href{mailto:lapchi@uwaterloo.ca}{lapchi@uwaterloo.ca}},~~~~~
Yin Tat Lee\footnote{University of Washington and Microsoft Research. Supported by NSF award CCF-1740551.  Part of this work was done while visiting University of Waterloo. Email: \href{mailto:yintat@uw.edu}{yintat@uw.edu}},~~~~~
Akshay Ramachandran\footnote{University of Waterloo. Supported by NSERC Discovery Grant 2950-120715 and NSERC Accelerator Supplement 2950-120719. Email: \href{mailto:a5ramachandran@uwaterloo.ca}{a5ramachandran@uwaterloo.ca}}}

\date{}

\begin{document}

\begin{titlepage}
\def\thepage{}
\thispagestyle{empty}

\maketitle

\begin{abstract}
The Paulsen problem is a basic open problem in operator theory:
Given vectors $u_1, \ldots, u_n \in \R^d$ that are $\eps$-nearly satisfying the Parseval's condition and the equal norm condition,
is it close to a set of vectors $v_1, \ldots, v_n \in \R^d$ that exactly satisfy the Parseval's condition and the equal norm condition?
Given $u_1, \ldots, u_n$, the squared distance (to the set of exact solutions) is defined as $\inf_{v} \sum_{i=1}^n \norm{u_i - v_i}_2^2$ where the infimum is over the set of exact solutions.
Previous results show that the squared distance of any $\eps$-nearly solution is at most $O(\poly(d,n,\eps))$ and there are $\eps$-nearly solutions with squared distance at least $\Omega(d\eps)$.
The fundamental open question is whether the squared distance can be independent of the number of vectors $n$.

We answer this question affirmatively by proving that the squared distance of any $\eps$-nearly solution is $O(d^{13/2} \eps)$.
Our approach is based on a continuous version of the operator scaling algorithm and consists of two parts.
First, we define a dynamical system based on operator scaling and use it to prove that the squared distance of any $\eps$-nearly solution is $O(d^2 n \eps)$.
Then, we show that by randomly perturbing the input vectors, the dynamical system will converge faster and the squared distance of an $\eps$-nearly solution is $O(d^{5/2} \eps)$ when $n$ is large enough and $\eps$ is small enough.
To analyze the convergence of the dynamical system, we develop some new techniques in lower bounding the operator capacity, a concept introduced by Gurvits to analyze the operator scaling algorithm.

\end{abstract}

\end{titlepage}

\thispagestyle{empty}

\tableofcontents

\newpage

\section{Introduction}

A set of $n$ vectors $v_1, \ldots, v_n \in \R^d$ is called an equal norm Parseval frame if it satisfies the Parseval's condition and the equal norm condition:
\begin{equation} \label{e:Parseval}
\sum_{i=1}^n v_i v_i^T = I_d \quad {\rm and} \quad
\norm{v_i}_2^2 = \frac{d}{n} {\rm~for~} 1 \leq i \leq n,
\end{equation} 
where $I_d$ is the $d \times d$ identity matrix.
A set of $n$ vectors $u_1, \ldots, u_n \in \R^d$ is an $\eps$-nearly equal norm Parseval frame if
\begin{equation} \label{e:epsParseval}
(1-\eps)I_d \preceq \sum_{i=1}^n u_i u_i^T \preceq (1+\eps)I_d
\quad {\rm and} \quad
(1-\eps) \frac{d}{n} \leq \norm{u_i}_2^2 \leq (1+\eps) \frac{d}{n}.
\end{equation}
Given two sets of vectors $U=\{u_i\}_{i=1}^n$ and $V=\{v_i\}_{i=1}^n$, the squared distance between them is defined as 
\begin{equation} \label{e:sqdist-vector}
\dist(U,V) = \sum_{i=1}^n \norm{u_i-v_i}_2^2.
\end{equation}
Let $\cal F$ be the set of equal norm Parseval frames.
Given a set of vectors $U=\{u_i\}_{i=1}^n$,
the squared distance to the set of equal norm Parseval frame is defined as
\begin{equation} \label{e:sqdist-frame}
\dist(U,{\cal F}) = \inf_{V \in {\cal F}} \dist(U,V).
\end{equation}
The Paulsen problem asks how close is an $\eps$-nearly equal norm Parseval frame to an equal norm Parseval frame.

\begin{definition}[the Paulsen problem] \label{d:Paulsen}
The Paulsen problem asks what is the best function $f(d,n,\eps)$ so that
\[
\dist(U,{\cal F}) \leq f(d,n,\eps)
\]
for any set of $d$-dimensional vectors $U=\{u_1, \ldots, u_n\}$ that forms an $\eps$-nearly equal norm Parseval frame.
\end{definition}

The fundamental open question of the Paulsen problem is whether $f(d,n,\eps)$ can be independent of the number of vectors $n$ and only dependent on the dimension $d$ and the error $\eps$~\cite{survey,projection}.

\subsection{History and Motivations}

The Paulsen problem has been open for over fifteen years despite receiving quite a bit of attention~\cite{survey,road,auto,projection}.
It has been listed as a major open problem in frame theory in the literature (see e.g.~\cite{survey,Mixon,Hilbert}).

An equal norm Parseval frame (also known as an unit-norm tight frame) is a natural generalization of an orthonormal basis.
It is used as an overcomplete basis (see the introductions in the books~\cite{Christensen,edited-book,overcompleteness}) and has various applications in signal processing and communication theory,
including noise and erasure reduction~\cite{HolmesPaulsen,erasures,graphs,Vershynin}, quantization robustness~\cite{GKK,BPY,quantization}, and digital fingerprinting~\cite{fingerprinting}.
For some applications in signal processing~\cite{Grassmannian,HolmesPaulsen} and quantum information theory~\cite{symmetric-quantum,existence-ETF}, equal norm Parseval frames with additional properties such as Grassmannian frames (which minimize the maximal inner product) and equiangular frames (in which the inner products are the same) are needed to provide optimal performance.

These applications motivate the ``frame design'' questions of constructing equal norm Parseval frames.
It is known that equal norm Parseval frames exist for any $d \leq n$.
However, it is difficult to construct equal norm Parseval frames, with only a few algebraic constructions known (e.g. truncation of Discrete Fourier transform matrices, vertices of the Platonic solids, constructions from groups; see the survey~\cite{Waldron}). 
On the other hand, it is known that the set of equal norm Parseval frames contains manifold of nontrivial dimensions~\cite{manifold}, and so the algebraic methods only produce a few examples from the continuum of the set of all equal norm Parseval frames~\cite{auto}.

Besides algebraic constructions, researchers have also used numerical methods to construct equal norm Parseval frames.
It is much easier to construct ``nearly'' equal norm Parseval frames as a set of random equal-norm vectors is nearly Parseval with high probability.
Tropp et al.~\cite{Tropp} proposed alternating projection algorithms to construct equal norm Parseval frames and equiangular frames from these nearly equal norm Parseval frames.
They show positive experimental results and some partial convergence analysis.
Holmes and Paulsen~\cite{HolmesPaulsen} studied the optimal parameters for Grassmannian frames which are even harder to construct.
They construct some nearly equal norm Parseval frames with small maximal inner product, and ask the question whether these are good estimates of the optimal parameters for Grassmannian frames.
This work led Paulsen to ask a number of people whether a nearly equal norm Parseval frame is always close to an equal norm Parseval frame (if so then their estimates are accurate), and eventually it is known as the Paulsen problem first formally stated in~\cite{road}.

Proving a good upper bound for the Paulsen problem would give us a firm foundation to work with nearly equal norm Parseval frames, both in theory and in applications.
Indeed, our method can be seen as a continuous version of the alternating projection algorithm of Tropp et al.~\cite{Tropp}, and our results can be viewed as a rigorous justification of the numerical approach for constructing equal norm Parseval frames.
We hope that our techniques will be useful to the difficult open question of constructing equiangular equal norm Parseval frames, as Tropp et al.~\cite{Tropp} also proposed an alternating projection algorithm for constructing these frames.

\subsection{Previous Work on the Paulsen Problem}

A compactness argument shows that the function $f$ in the Paulsen problem must exist~\cite{survey}.
There are simple examples showing that $f(d,n,\eps) \geq d \eps$~\cite{projection}.

Bodmann and Casazza~\cite{road} proved that $f(d,n,\eps) \leq O(d^{18} n^{4} \eps^2)$ when $d$ and $n$ are relatively prime.
Their approach is to analyze a dynamical system that improves the closeness to the equal norm condition while ensuring that the Parseval's condition is satisfied.
Casazza, Fickus, and Mixon~\cite{auto} proved that $f(d,n,\eps) \leq O(d^{42} n^{14} \eps^2)$ when $d$ and $n$ are relatively prime,
and they extended this result to the general case and proved that $f(d,n,\eps) \leq O(d^{13/7} n \eps^{2/7})$.
Their approach is to analyze a gradient descent algorithm that improves the closeness to the Parseval's condition while ensuring that the equal norm condition is satisfied.

Cahill and Casazza~\cite{projection} showed that the Paulsen problem is equivalent to another fundamental and deep problem in operator theory called the projection problem:
Find the best function $g(d,n,\eps)$ such that the following holds.
Given an $n$-dimensional orthonormal basis $e_1, \ldots, e_n \in \R^n$ and
a projection $P$ of rank $d$ that satisfies
\[
(1-\eps) \frac{d}{n} \leq \norm{Pe_i}^2_2 \leq (1+\eps) \frac{d}{n} {\rm~~for~} 1 \leq i \leq n,
\]
there is a projection $Q$ with $\norm{Qe_i}_2^2 = \frac{d}{n}$ for $1 \leq i \leq n$ and 
\[
\sum_{i=1}^n \norm{Pe_i - Qe_i}_2^2 \leq g(d,n,\eps).
\]
Cahill and Casazza~\cite{projection} proved that $f(d,n,\eps)$ and $g(d,n,\eps)$ are within a factor $2$ of each other.

\subsection{Results and Techniques}

Our main result is a proof that the function is independent of the number of vectors.

\begin{theorem} \label{t:main}
For any set $U$ of $d$-dimensional vectors that is an $\eps$-nearly equal norm Parseval frame,
\[\dist(U,{\cal F}) \leq O(d^{13/2} \eps).\]
\end{theorem}

There is a very natural approach towards solving the Paulsen problem.
Given an $\eps$-nearly equal norm Parseval frame $u_1, \ldots, u_n \in \R^d$,
we alternately fix the Parseval condition (by setting $u_i \gets (\sum_{i=1}^n u_i u_i^T)^{-\frac{1}{2}} u_i$) and the equal norm condition (by scaling $u_i$ so that $\norm{u_i}_2^2 = \frac{d}{n}$),
until both conditions are satisfied exactly and we keep track of the sum of the movement of these operations.
We observe that this natural alternating algorithm is a special case of the operator scaling algorithm studied in~\cite{gurvits,operator}; see Section~\ref{s:background}.
So, this alternating algorithm will converge under some mild condition~\cite{gurvits,operator}, and also there are closed-form formulas for the movement of each operation~\cite{survey}.
But the problem of this approach is that the sum of the movement could be very large, as the path to an exact solution could zig-zag between the alternating operations; see Subsection~\ref{ss:discrete}.

Our approach is based on a continuous version of the operator scaling algorithm~\cite{gurvits,operator}.
There are two main parts.
To avoid the zig-zag movement, we define a dynamical system based on the (discrete) operator scaling algorithm, so that the two alternating operations are combined into one and the movement is continuous.
The dynamical system satisfies some very nice identities.
Using these identities and the concept of operator capacity defined by Gurvits~\cite{gurvits} in analyzing the convergence of the operator scaling algorithm, we can bound the total movement of our dynamical system given a nearly equal norm Parseval frame.

\begin{theorem}[informal] \label{t:dynamical}
Given a set of $d$-dimensional vectors $U=\{u_1,\ldots,u_n\}$ that forms an $\eps$-nearly equal norm Parseval frame,
there is a dynamical system that transforms $U$ into a set of $d$-dimensional vectors $V=\{v_1,\ldots,v_n\}$ that forms an equal-norm Parseval frame and $\dist(U,V) \leq O(d^2 n \eps)$.
\end{theorem}

We prove Theorem~\ref{t:dynamical} in the more general operator setting instead of just the frame setting as in the Paulsen problem;
see the introduction in Section~\ref{s:dynamical}.
We believe that the operator setting is of independent interest,
e.g. it is closely related to the Brascamp-Lieb constants that will be discussed in Subsection~\ref{ss:related}.

Using the dynamical system for an arbitrary $\eps$-nearly equal norm Parseval frame, the analysis in Theorem~\ref{t:dynamical} is tight and the dependency on $n$ is unavoidable.
Our intuition is that the set of $\eps$-nearly equal norm Parseval frames with large total movement in our dynamical system is small.
The second part in our approach is a smoothed analysis~\cite{smoothed} of continuous operator scaling.
We prove that if we randomly perturb an arbitrary $\eps$-nearly equal norm Parseval frame appropriately (in a dependent manner), then the perturbed frame is simultaneously close to the original frame and an equal norm Parseval frame with high probability, by showing that its total movement in our dynamical system is independent of $n$.

\begin{theorem}[informal] \label{t:perturb}
Given a set of $d$-dimensional vectors $U=\{u_1,\ldots,u_n\}$ that forms an $\eps$-nearly equal norm Parseval frame with $n \gg d^4$ and $\eps \ll 1/d^{11/2}$,
we can perturb $U$ to obtain $\tilde U$ such that $\dist(U,\tilde{U}) \leq O(d^{5/2} \eps)$ and furthermore $\dist(\tilde U,{\cal F}) \leq O(\sqrt{d} \eps)$ by using the dynamical system in Theorem~\ref{t:dynamical}.
\end{theorem} 

This solves the problem when $n$ is large enough.
Together with Theorem~\ref{t:dynamical},
we obtain Threom~\ref{t:main} by using Theorem~\ref{t:dynamical} when $n$ is small; see Subsection~\ref{ss:together}.

To prove Theorem~\ref{t:perturb}, we develop some new techniques 
to analyze the convergence of the dynamical system in the perturbed frame.
Gurvits~\cite{gurvits} defined a notion called the operator capacity to analyze the operator scaling algorithm.
Recently, the operator scaling algorithm is used to design a polynomial time algorithm to solve the non-commutative rank problem~\cite{non-commutative}, while the key is a new lower bound on the operator capacity.
We find an interesting connection between our dynamical system and the operator capacity.
We use it to prove a better lower bound on the operator capacity in the perturbed instance, which implies a faster convergence rate in the perturbed instance.
We discuss some implications of our results to related work on operator scaling in the following subsection, including bounds on the optimal constants in Brascamp-Lieb inequalities and on the running time of fast algorithms for matrix scaling.

\subsection{Related Work on Frame Scaling, Operator Scaling, and Matrix Scaling} \label{ss:related}

Scaling a frame into an equal norm Parseval frame, and more generally, scaling an operator into a doubly stochastic operator (see Subsection~\ref{ss:operator}) has various applications in theoretical computer science.
Sometimes they go under different names such as radial isotropic positions in machine learning~\cite{HardtMoitra}, and geometric conditions in Brascamp-Lieb inequalities~\cite{Ball,Barthe,Brascamp-Lieb}.

An early application of frame scaling is discovered by Forster~\cite{Forster}, who showed that a set of $n$ vectors $v_1, \ldots, v_n \in \R^d$ can always be scaled (see Definition~\ref{d:scaling}) to an equal norm Parseval frame if every subset of $d$ vectors is linearly independent, and he used this result to derive a lower bound on the sign rank of the Hadamard matrix with applications in proving communication complexity lower bounds.
We note that Forster's scaling result was proved earlier in a more general setting by Gurvits and Samorodnitsky~\cite{GS} in their work of approximating mixed discriminants, and is also implicit in the work of Barthe~\cite{Barthe} in proving Brascamp-Lieb inequalities.
A recent application of frame scaling is found by Hardt and Moitra~\cite{HardtMoitra} in robust subspace discovery.

Operator scaling was introduced by Gurvits~\cite{gurvits} in an attempt to design a deterministic polynomial time algorithm for polynomial identity testing, and he used it to solve the special case when the commutative rank of a symbolic matrix is equal to its non-commutative rank (e.g. this includes the linear matroid intersection problem over reals).
Recently, Garg, Gurvits, Oliveira, and Wigderson~\cite{operator} improved Gurvits' analysis to prove that the alternating algorithm for operator scaling can be used to compute the non-commutative rank of a symbolic matrix in polynomial time.
Subsequently, the alternating algorithm for operator scaling is used by the same group~\cite{Brascamp-Lieb} to obtain a polynomial time algorithm to compute the optimal constants in Brascamp-Lieb inequalities, which we will elaborate more below as it is related to our work.

The Brascamp-Lieb inequalities~\cite{BL} and their reversed form established by Barthe~\cite{Barthe} are general classes of inequalities with important applications in functional analysis and convex geometry (e.g. including Nelson's hypercontractivity inequality and the Brunn-Minkowski inequality as special cases).
The optimal constants for thses inequalities are determined by Ball~\cite{Ball} assuming the geometric condition (which is a condition similar to that in John's ellipsoid theorem).
Garg, Gurvits, Oliveira and Wigderson~\cite{Brascamp-Lieb} show that the Brascamp-Lieb constants are equivalent to the capacity of an operator by a simple transformation, in which the geometric condition corresponds exactly to the doubly stochastic condition.
Therefore, the algorithm in~\cite{operator} can be employed to scale the input to satisfying the geometric condition so as to compute the optimal constant.
For our smoothed analysis in Section~\ref{s:smoothed}, we develop a new technique to proving a lower bound on the operator capacity and thus an upper bound on the Brascamp-Lieb constant.  In particular, this implies improved bounds on the Brascamp-Lieb constants for perturbed instances in the rank-one case (which is the case that Brascamp and Lieb proved in~\cite{BL}). 
See~\cite{Brascamp-Lieb} and the references therein for applications of these bounds to non-linear Brascamp-Lieb inequalities.

Matrix scaling~\cite{sinkhorn} is a well-studied special case of operator scaling.
It has applications in numerical analysis, in approximating permanents~\cite{LSW} and in combinatorial geometry~\cite{DGOS}.
Very recently, much faster algorithms are developed for matrix scaling by two independent research groups~\cite{Cohen,ALOW}.
Cohen, Madry, Tsipras and Vladu~\cite{Cohen} obtain an algorithm for matrix scaling with running time $\widetilde{O}(m \log \kappa \log^2(1/\eps))$, where $m$ is the number of nonzeros in the input matrix, $\kappa$ is the ratio between the largest and the smallest entries in the optimal scaling solution, and $\eps$ is the error parameter of the output.
Note that the algorithm is near linear time when $\kappa$ is bounded by a polynomial in $m$, but in general it could be exponentially large.
Not much is known about upper bounding $\kappa$ for specific instances, except when the input matrix is strictly positive~\cite{KK}.
Our techniques for smoothed analysis in Section~\ref{s:smoothed} provides a new way to bound $\kappa$; see Remark~\ref{r:condition}.
In particular, this implies that the algorithm in~\cite{Cohen} is near linear time in a pseudorandom instance as defined in Definition~\ref{d:pseudorandom} (not necessarily strictly positive).

To summarize, our techniques developed in solving the Paulsen problem provides new tools in bounding the mathematical quantities involved in scaling problems such as the operator capacity and $\kappa$ about optimal scaling solutions.
See the second half of Subsection~\ref{ss:overview} for an overview of these techniques.
These provide a new perspective to look at those quantities using the parameters in our dynamical system.
Currently, our smoothed analysis is tailored for the Paulsen problem, but we believe that it can be extended to the more general operator setting and also to more natural conditions (rather than just the psuedorandom condition in Definition~\ref{d:pseudorandom}) to prove useful results about other problems solved by scaling techniques.

\subsection{Organization and Overview}

We first review the background of operator scaling and see that the Paulsen problem is a special case in the operator framework in Section~\ref{s:background}.
We will also introduce the matrix scaling problem in Section~\ref{s:background} as this is a key intermediate problem in our proof of the second part.
We would like to mention that many results in this paper are first proved in the simpler matrix setting and then generalized to the operator setting.

We divide the proof of Theorem~\ref{t:main} into two sections.
In Section~\ref{s:dynamical}, we define our dynamical system based on operator scaling and prove Theorem~\ref{t:dynamical}.
The results in this section works in the more general operator setting.
We discover some nice formulas for the dynamical system to analyze its convergence.
And we establish a close connection between the operator capacity lower bound and the squared distance bound for the Paulsen problem.

In Section~\ref{s:smoothed}, we analyze the perturbation step to prove Theorem~\ref{t:perturb}.
Using a known reduction that we will discuss in Section~\ref{s:background} and see the proof in Subsection~\ref{ss:capLB},
we reduce the operator capacity lower bound to a matrix capacity lower bound.
Using some probabilistic arguments, we will show that by perturbing the vectors, the corresponding matrix will have some pseudorandom property.
Then we use a combinatorial argument to show that the pseudorandom property will imply a fast convergence of our dynamical system for matrix scaling.
Interestingly, we show that the fast convergence of our dynamical system will imply a stronger matrix capacity lower bound, and this leads to a bound on Paulsen problem without any dependency on the number of vectors.
We remark that the perturbation step only applies in the frame setting (rather than the general operator setting).

The proof ideas described so far are of high level.
We will give a more concrete technical overview in each section after the appropriate background is covered; see Subsection~\ref{ss:outline} and Subsection~\ref{ss:overview}.

\section{The Paulsen Problem, Operator Scaling, and Matrix Scaling} \label{s:background}

In this section, we first describe a natural alternating algorithm to solve the Paulsen problem in Subsection~\ref{ss:alternating}.
Then, we describe the operator scaling problem and the operator scaling algorithm in Subsection~\ref{ss:operator}, and then see that it captures the natural alternating algorithm for the Paulsen problem as a special case in Subsection~\ref{ss:reduction}.
The operator scaling algorithm is analyzed in~\cite{gurvits,operator} and we will use their tools in solving the Paulsen problem.
In Subsection~\ref{ss:capacity}, we describe the notion of operator capacity introduced by Gurvits~\cite{gurvits} and explain how it is used in analyzing the convergence of the operator scaling algorithm.
The operator capacity will be important in analyzing our continuous operator scaling algorithm.
Finally, in Subsection~\ref{ss:matrix}, we describe the matrix scaling problem and the notion of matrix capacity and explain how it is used as an intermediate step in proving lower bounds for the operator capacity.
We will also use the matrix capacity as an intermediate step to prove a stronger lower bound for the operator capacity for a perturbed instance.

We remark that the definitions and the results in this section will not be directly used in later proofs.
We will define formally what we need in Section~\ref{s:dynamical} and in Section~\ref{s:smoothed}, which will be slightly different from previous work.

\subsection{Alternating Algorithm for the Paulsen Problem} \label{ss:alternating}

There is a natural algorithm towards solving the Paulsen problem.
Let $U^{(0)} = \{u_1^{(0)},\ldots,u_n^{(0)}\}$ be the initial $\eps$-nearly equal norm Parseval frame, with
\[
(1-\eps)I_d \preceq \sum_{i=1}^n u_i^{(0)} {u_i^{(0)}}^T \preceq (1+\eps)I_d
\quad {\rm and} \quad
(1-\eps) \frac{d}{n} \leq \norm{u_i^{(0)}}_2^2 \leq (1+\eps) \frac{d}{n}.
\]
Given $U^{(t)}=\{u_1^{(t)},\ldots,u_n^{(t)}\}$ for some non-negative integer $t$,
we define $S^{(t)} = \sum_{i=1}^n u_i^{(t)} {u_i^{(t)}}^T$.
If $S^{(0)}$ is singular, then there is a zero eigenvalue and thus $\eps=1$.
Similarly, if $\norm{u_i^{(0)}}=0$ for some $i$, then we also have $\eps=1$.
In these cases, the Paulsen problem is trivial as we can just output an arbitrary equal norm Parseval frame $V = \{v_1,\ldots,v_n\}$ and 
\[\dist(U,{\cal F}) \leq \dist(U,V) = \sum_{i=1}^n \norm{u_i^{(0)} - v_i}_2^2 
\leq \sum_{i=1}^n (2\norm{u_i^{(0)}}_2^2 + 2\norm{v_i}_2^2)
= \sum_{i=1}^n O(\frac{d}{n}) = O(d) = O(d\eps).\]
Henceforth, we assume that $S^{(0)}$ is non-singular and $\norm{u_i^{(0)}} \neq 0$ for $1 \leq i \leq n$.
Then, we define 
\[
u_i^{(t+1)} = (S^{(t)})^{-\frac{1}{2}} u_i^{(t)} \quad {\rm and} \quad
u_i^{(t+2)} = u_i^{(t+1)}/\norm{u_i^{(t+1)}}.
\]
Note that $S^{(t)}$ remains to be non-singular and $\norm{u_i^{(t)}} \neq 0$ for $1 \leq i \leq n$, and so these vectors are well-defined.
By construction, it is easy to check that $U^{(t+1)}$ satisfies the Parseval condition and $U^{(t+2)}$ satisfies the equal norm condition (although $\norm{u_i^{(t+2)}}_2^2 = 1$ instead of $\norm{u_i^{(t+2)}}_2^2 = \frac{d}{n}$) for every even number $t$.
We would like to show that $U^{(t+1)}$ will converge to an equal norm Parseval frame for some even number $t$.
We observe that this alternating algorithm is a special case of the operator scaling algorithm in the next subsection.

\subsection{Operator Scaling} \label{ss:operator}

The operator scaling problem is defined by Gurvits~\cite{gurvits}.
Given $m \times n$ matrices $U_1, \ldots, U_k$,
the operator scaling problem is to find an $m \times m$ matrix $L$ and an $n \times n$ matrix $R$ such that if we set $V_i = L U_i R$ for $1 \leq i \leq k$ then
\begin{equation} \label{e:conditions}
\sum_{i=1}^k V_i {V_i}^T = I_{m} \quad {\rm and} \quad
\sum_{i=1}^k {V_i}^T V_i = \frac{m}{n} I_{n}.
\end{equation} 
The operator scaling algorithm studied in~\cite{gurvits,operator} is the natural alternating algorithm.
Let $\U^{(0)} = \{U^{(0)}_1, \ldots, U^{(0)}_k\}$ be the initial $m \times n$ matrices.
Given $\U^{(t)} = \{U^{(t)}_1, \ldots, U^{(t)}_k\}$ for an even number $t$,
we define 
\[
L^{(t)} = \sum_{i=1}^k U^{(t)}_i (U^{(t)}_i)^T \quad {\rm and} \quad
U_i^{(t+1)} = \big(L^{(t)}\big)^{-\frac{1}{2}} U_i^{(t)},
\]
and 
\[
R^{(t+1)} = \sum_{i=1}^k (U^{(t+1)}_i)^T U^{(t+1)}_i \quad {\rm and} \quad
U_i^{(t+2)} = U_i^{(t+1)} \big(R^{(t+1)}\big)^{-\frac{1}{2}}.
\]
We assume that $L^{(0)}$ is non-singular and $R^{(1)}$ is non-singular.
Then $L^{(t)}$ and $R^{(t+1)}$ remain to be non-singular and so these matrices are well-defined.
By construction, it is easy to check that $\U^{(t+1)}$ satisfies the first condition in~(\ref{e:conditions}) that 
\[
\sum_{i=1}^k U^{(t+1)}_i (U^{(t+1)}_i)^T 
= \sum_{i=1}^k \big(L^{(t)}\big)^{-\frac{1}{2}} U_i^{(t)} (U^{(t)}_i)^T \big(L^{(t)}\big)^{-\frac{1}{2}}
= \big(L^{(t)}\big)^{-\frac{1}{2}} L^{(t)} \big(L^{(t)}\big)^{-\frac{1}{2}}
= I_{m},
\] 
and $\U^{(t+2)}$ satisfies the (scaled) second condition in~(\ref{e:conditions}) that 
\[
\sum_{i=1}^k {U^{(t+2)}_i}^T U^{(t+2)}_i 
= \sum_{i=1}^k \big(R^{(t+1)}\big)^{-\frac{1}{2}}  (U_i^{(t+1)})^T  U_i^{(t+1)} \big(R^{(t+1)}\big)^{-\frac{1}{2}}
= \big(R^{(t+1)}\big)^{-\frac{1}{2}} R^{(t+1)} \big(R^{(t+1)}\big)^{-\frac{1}{2}}
= I_{n},
\]
for every even number $t$.
When the two matrices  
\[L := \prod_{t {\rm~even}} \big(L^{(t)}\big)^{-\frac12} \quad {\rm and} \quad 
R := \prod_{t {\rm~odd}} \big(R^{(t)}\big)^{-\frac12}\] 
converge, they are the required scaling matrices.

\subsection{Reducing The Paulsen Problem to Operator Scaling} \label{ss:reduction}

We observe that the natural alternating algorithm for the Paulsen problem is a special case of the operator scaling algorithm.
For each vector $u^{(0)}_i \in \R^d$, 
we associate an $d \times n$ matrix $U^{(0)}_i$ in which the $i$-th column is $u^{(0)}_i$ and all other columns are zero.
Given $U^{(0)} = \{u^{(0)}_1, \ldots, u^{(0)}_n\}$,
we apply the operator scaling algorithm to these matrices $\U^{(0)} = (U^{(0)}_1, \ldots, U^{(0)}_n)$.
In this reduction to operator scaling, we have $k:=n$ and $m:=d$.

By the definition of $U^{(0)}_i$, it is easy to check that
\[ 
L^{(0)}=\sum_{i=1}^n U^{(0)}_i (U^{(0)}_i)^T
=\sum_{i=1}^n u^{(0)}_i (u^{(0)}_i)^T=S^{(0)},
\]
and
\[R^{(1)} {\rm~is~the~} n \times n {\rm~diagonal~matrix~with~} R^{(1)}_{i,i} = \norm{u^{(1)}_i}_2^2,\]
and thus the operator scaling algorithm applies to $\U^{(0)}$ corresponds exactly to the alternating algorithm applies to $U^{(0)}$.
Inductively, we see that $\U^{(t)}$ corresponds exactly to $U^{(t)}$, and so the alternating algorithm for the Paulsen problem is a special case of the operator scaling algorithm.

Suppose the operator scaling algorithm converges to a solution $\V = \{V_1, \ldots, V_n\}$ that satisfies the two conditions in~(\ref{e:conditions}).
Each $V_i$ has nonzero entries only in the $i$-th column,
and we let $v_i$ be the $i$-th column of $V_i$.
The first condition implies the Parseval condition
\[\sum_{i=1}^k V_i {V_i}^T = I_{m}
\implies 
\sum_{i=1}^n v_i {v_i}^T = I_{d},\]
and the second condition implies the equal norm condition
\[\sum_{i=1}^k {V_i}^T V_i = \frac{m}{n} I_{n}
\implies 
\norm{v_i}_2^2 = \frac{d}{n} {\rm~for~} 1 \leq i \leq n.
\]
Therefore, the solution that the operator scaling algorithm converges to corresponds to a solution to the Paulsen problem.
This is the general approach that we will take to study the Paulsen problem. 

We remark that this reduction can be used to derive Forster's theorem~\cite{Forster}, which has an interesting application in proving communication complexity lower bound, from the results in operator scaling~\cite{GS} that we will describe in the next subsection.

\subsection{Operator Capacity and Convergence of Operator Scaling} \label{ss:capacity}

To analyze the convergence of the operator scaling algorithm, Gurvits~\cite{gurvits} defined the following important notion of operator capacity.
Note that this definition is for square matrices in~\cite{gurvits} and is extended to rectangular matrices in~\cite{Brascamp-Lieb}.
We will define a rectangular version for our purpose in Section~\ref{s:dynamical}.
Given $\U = \{U_1, \ldots, U_k\}$ where each $U_i \in \R^{n \times n}$, the operator $T_{\U}: \R^{n \times n} \to \R^{n \times n}$ is defined as $T_{\U}(X) = \sum_{i=1}^k U_i X U_i^T$.
The capacity of the operator $T_{\U}$ is defined as
\[
\capa(\U) := \inf_X\{ \det(T_{\U}(X))~|~X \succeq 0 {\rm~and~} \det(X)=1 \}.
\]
The capacity is used as a potential function to analyze the convergence of the operator scaling algorithm.
The following arguments in this subsection are from~\cite{gurvits,operator}.

Firstly, we establish an upper bound on the capacity.
When $\tr(\sum_{i=1}^k U_i U_i^T) = n$ (for instance when one of the two conditions in~(\ref{e:conditions}) is satisfied), 
then $\capa(\U)$ is always upper bounded by one because
(from Proposition 2.8 of~\cite{operator})
\begin{equation} \label{e:capUB}
\capa(\U) \leq \det(T_{\U}(I)) \leq (\tr(T_{\U}(I))/n)^n \leq 1,
\end{equation}
where the second inequality follows from the AM-GM inequality.

Secondly, we can keep track of the change of the capacity during the operator scaling algorithm.
It is clear from the definition that capacity is multiplicative.
Given $\U^{(t)}$, Proposition 2.7 from~\cite{operator} shows that
\begin{equation} \label{e:change}
\capa (\U^{(t+2)}) = \frac{\capa (\U^{(t)})}{\det(L^{(t)}) \cdot \det(R^{(t+1)})}.\end{equation}
To measure how close $\U = \{U_1, \ldots, U_k\}$ is to satisfying the two conditions in~(\ref{e:conditions}), it is defined
\[
\Delta(\U) = \tr[(\sum_{i=1}^k U_i U_i^T-I_n)^2] + \tr[(\sum_{i=1}^k U_i^T U_i - I_n)^2],
\]
so that $\Delta(\U) = 0$ if and only if $\U$ satisfies the two conditions in~(\ref{e:conditions}).
Using the formula of the change of the capacity,
it is shown in Lemma 5.2 of~\cite{operator} that the capacity is increased by a constant factor when $\Delta(\U) \geq 1$ and
\[
\capa(\U^{(t+2)}) \geq \capa(\U^{(t)}) \cdot \exp(\Delta(\U^{(t)}))
\quad {\rm~when~} \Delta(\U) \leq 1. 
\]

Finally, to show that the operator scaling algorithm converges in polynomial time, we need to establish a good lower bound on the initial operator capacity, and this is the key step in analyzing the operator scaling algorithm.
It is shown in Lemma 3.4 of~\cite{operator} that if the second condition $\sum_{i=1}^k U_i^T U_i = I_n$ is satisfied then 
\begin{equation} \label{e:operator-cap}
\capa(\U) \geq (1-\sqrt{n \cdot \Delta(\U)})^n.
\end{equation}

Now, with the upper bound, the lower bound and the change of the capacity, we can analyze the number of iterations for $\Delta(\U^{(t)}) \leq \delta$ for any $0 < \delta \leq 1$.
Suppose $\Delta(\U^{(t)}) > \delta$ for $0 \leq t \leq 2T$,
then 
\[
\capa(\U^{(2T)}) \geq \capa(\U^{(0)}) \cdot \exp(\delta \cdot T) 
\geq (1-\sqrt{n \cdot \Delta(\U^{(0)})})^n \cdot \exp(\delta \cdot T),
\]
and this would imply that 
\[\capa(\U^{(2T)}) > 1
\quad {\rm for} \quad
T = \Omega(n^{3/2} \sqrt{\Delta(\U^{(0)})} / \delta) 
\quad {\rm and} \quad
\Delta(\U^{(0)}) \leq \frac{1}{2n},
\]
contradicting that $\capa(\U^{(2T)}) \leq 1$.
To summarize, the operator capacity provides an indirect way to argue that $\Delta(\U^{(t)})$ would converge to zero as long as the initial capacity is positive.

It is then an important question to characterize when the operator capacity is positive.
Gurvits~\cite{gurvits} proved that $\capa(\U) > 0$ if and only if the operator $T_{\U}$ is rank non-decreasing, i.e. $\rank(T_{\U}(P)) \geq \rank(P)$ for any $P \succeq 0$.
Recently, Garg, Gurvits, Oliveria and Wigderson~\cite{operator} proved that when $T_{\U}$ is rank non-decreasing, then $\capa(\U) \geq \exp(-\poly(n,k))$ assuming the bit complexity of the input is bounded.
This implies that the operator scaling algorithm gives a polynomial time algorithm to determine whether $T_{\U}$ is rank non-decreasing, and this implies the first polynomial time algorithm for computing the non-commutative rank of a symbolic matrix.
For the Paulsen problem, we can assume that $\Delta(\U^{(0)})$ is small enough and thus the initial capacity is positive in Section~\ref{s:dynamical},  
and we will perturb the input so that the intial capacity is positive in Section~\ref{s:smoothed}. 

We will discuss the proof of the operator capacity lower bound (\ref{e:operator-cap}) in the next subsection,
which is based on a connection to matrix scaling and matrix capacity that we will define.
We will need to extend and improve this lower bound for the proof of Theorem~\ref{t:perturb}.

\subsection{Matrix Scaling and Matrix Capacity} \label{ss:matrix}

The operator scaling algorithm and the operator capacity are motivated by the corresponding definitions in a simpler matrix setting.

We call a non-negative matrix $A \in \R^{m \times n}$ doubly balanced if all the row sums are equal and all the column sums are equal.
Given a non-negative matrix $A$, the matrix scaling problem is to find an $m \times m$ diagonal matrix $L$ and an $n \times n$ diagonal matrix $R$ such that $LAR$ is doubly balanced.
The matrix scaling algorithm by Sinkhorn~\cite{sinkhorn} is the natural algorithm that alternatively scales the rows to have the same sum and then scales the columns to have the same sum until the resulting matrix becomes (close enough to) doubly balanced.
In the square case, it is known that (see e.g.~\cite{LSW}) the permanent of $A$ is positive if and only if there exist two sequences of positive diagonal matrices $L_i$ and $R_i$ such that $\lim_{i \to \infty} L_i A R_i$ is doubly stochastic.

There are many different analyses of the convergence of the matrix scaling algorithm.
There is an analysis~\cite{LSW} which is the same as the analysis of the operator scaling algorithm outlined in the previous subsection, using the permanent (which is the same as the capacity up to a simple transformation) as a potential function.
The capacity of a matrix is first explicitly stated in~\cite{operator},
although Gurvits and Yianilos~\cite{GY} considered an equivalent notion in measuring progress. 
Note that the definition is for square matrices, and we will define a rectangular version for our purpose in Section~\ref{s:dynamical}.
Given a non-negative matrix $A \in \R^{n \times n}$,
its capacity is defined as 
\[
\capa(A) = \inf_x \{ \prod_{i=1}^n (Ax)_i~|~\prod_{i=1}^n x_i = 1 {\rm~and~} x > 0\}.
\]
A matrix is called row-balanced if its row sums are the same.
Given a row-balanced matrix $A$ with average column sum one, 
we let $c_j$ be the sum of the $j$-th column and define
\[
\Delta(A) = \sum_{j=1}^n (c_j - 1)^2.
\]
It is proved in Lemma 3.2 of~\cite{operator} that for a row-balanced matrix $A$ with average column sum one, 
\begin{equation} \label{e:matrix-cap}
\capa(A) \geq (1 - \sqrt{n \cdot \Delta(A)})^n.
\end{equation}
The proof of the operator capacity lower bound (\ref{e:operator-cap}) in Lemma 3.4 of~\cite{operator} is through a reduction to this matrix capacity lower bound, which we will extend to our setting in Proposition~\ref{p:reduction-capacity}.

In the (discrete) operator scaling algorithm described in Subsection~\ref{ss:operator}, we can assume that one of the two conditions in~(\ref{e:conditions}) is satisfied and this simplifies the proofs.
In the continuous operator scaling algorithm that we will define in Section~\ref{s:dynamical}, typically both conditions are not satisfied and so we will need to slightly generalize the proof of (\ref{e:matrix-cap}) and the reduction from the operator capacity to the matrix capacity to prove an analogous statement of (\ref{e:operator-cap}).
This will be done in the Subsection~\ref{ss:capLB}.
We will also define a matrix version of the Paulsen problem in Section~\ref{ss:matrix-Paulsen}.
This will provide a new way to prove stronger matrix capacity lower bound in a perturbed instance.

\section{The Operator Paulsen Problem and Continuous Operator Scaling} \label{s:dynamical}

We consider the following generalization of the Paulsen problem to the operator setting.
We refer to a set of of matrices $\U = \{U_1, \ldots, U_k\}$ as an operator.

\begin{definition}[doubly balanced and doubly stochastic operator] \label{d:DS}
An operator $\V = \{V_1, \ldots, V_k\}$ where $V_i \in \R^{m \times n}$ for $1 \leq i \leq k$ is called doubly balanced if 
\[
\sum_{i=1}^k V_i V_i^T = cnI_m \quad {\rm~and~} \quad 
\sum_{i=1}^k V_i^T V_i = cmI_n
\]
for some scalar $c \geq 0$, and it is called doubly stochastic in the case when $c=1/n$, i.e.
\[
\sum_{i=1}^k V_i V_i^T = I_m \quad {\rm~and~} \quad 
\sum_{i=1}^k V_i^T V_i = \frac{m}{n} I_n.
\]
\end{definition}

\begin{definition}[$\eps$-nearly doubly stochastic operator] \label{d:epsDS}
An operator $\U = \{U_1, \ldots, U_k\}$ where $U_i \in \R^{m \times n}$ for $1 \leq i \leq k$ is called $\eps$-nearly doubly stochastic if 
\[
(1-\eps)I_m \preceq \sum_{i=1}^k U_i U_i^T \preceq (1+\eps) I_m 
\quad {\rm~and~} \quad 
(1-\eps)\frac{m}{n} I_n \preceq \sum_{i=1}^k U_i^T U_i \preceq (1+\eps) \frac{m}{n} I_n.
\]
\end{definition}

\begin{definition}[distance and squared distance] \label{d:distance}
Given $\U = \{U_1, \ldots, U_k\}$ and $\V = \{V_1, \ldots, V_k\}$ where $U_i,V_i \in \R^{m \times n}$ for $1 \leq i \leq k$, 
the squared distance between $\U$ and $\V$ is defined as
\[
\dist(\U,\V) := \sum_{i=1}^k \norm{U_i - V_i}_F^2,
\]
where $\norm{.}_F$ is the Frobenius norm of the matrix. 
The distance between $\U$ and $\V$ is defined as
\[
\distance(\U,\V) := \sqrt{\dist(\U,\V)} = \sqrt{\sum_{i=1}^k \norm{U_i - V_i}_F^2}.
\]
\end{definition}

\begin{definition}[the operator Paulsen problem] \label{d:operator-Paulsen}
Given $\U = \{U_1, \ldots, U_k\}$ where $U_i \in \R^{m \times n}$ for $1 \leq i \leq k$ that is $\eps$-nearly doubly stochastic,
the operator Paulsen problem asks what is the best function $h(k,n,m,\eps)$ so that 
\[
\inf_{\V} \dist(\U,\V) 
\leq h(k,n,m,\eps),
\]
where the infimum $\V$ is over the sets of matrices which are doubly stochastic.
\end{definition}

The main theorem in this section is that $h(k,n,m,\eps) \leq O(m^2 n \eps)$ which we will prove in Theorem~\ref{t:operator-Paulsen} in Subsection~\ref{ss:operator-Paulsen}.
Our approach is to define a dynamical system and use it to find a scaling solution defined as follows.

\begin{definition}[operator scaling] \label{d:scaling}
Given an operator $\U = \{U_1, \ldots, U_k\}$ where $U_i \in \R^{m \times n}$ for $1 \leq i \leq k$, we say $\V = \{V_1, \ldots, V_k\}$ is a scaling of $\U$ if there exist 
\[
L \in \R^{m \times m} \quad {\rm and} \quad 
R \in \R^{n \times n} \quad {\rm such~that~} V_i = L \cdot U_i \cdot R {\rm~for~} 1 \leq i \leq k.
\]
\end{definition}

This will imply Theorem~\ref{t:dynamical}, using a similar reduction as described in Subsection~\ref{ss:reduction}.

\subsubsection*{Organization and Overview}

We first see in Subsection~\ref{ss:discrete} that the natural attempt to use the (discrete) operator scaling algorithm described in Subsection~\ref{ss:operator} to solve the Paulsen problem would not work directly.
This motivates us to define a dynamical system based on the operator scaling algorithm in Subsection~\ref{ss:continuous}.
In Subsection~\ref{ss:outline},
we prove our main technical result that given $\U^{(0)}$ that is $\eps$-nearly doubly stochastic, the dynamical system will produce $\U^{(\infty)}$ that is doubly balanced with $\dist(\U^{(0)},\U^{(\infty)}) \leq O(m^2n \eps)$,
assuming some formulas for the dynamical system and a lower bound on the operator capacity.
Then, we derive the formulas in Subsection~\ref{ss:formulas},
and we prove the operator capacity lower bound in Subsection~\ref{ss:capLB}.
Finally, we do some preprocessing and postprocessing to obtain a bound on the operator Paulsen problem in Subsection~\ref{ss:operator-Paulsen}, and use a reduction similar to that in Subsection~\ref{ss:reduction} to obtain a bound on the Paulsen problem and prove Theorem~\ref{t:dynamical} in Subsection~\ref{ss:Paulsen}.

\subsection{Discrete Operator Scaling} \label{ss:discrete}

The natural first attempt to the Paulsen problem is to use the analysis in Subsection~\ref{ss:capacity} to bound the distance by the ``total movement'' in the operator scaling algorithm.
In this subsection, we restrict to the original frame setting of the Paulsen problem, and argue that this natural attempt of using discrete operator scaling would not work directly.

Given a set of $d$-dimensional vectors $U^{(0)} = \{u^{(0)}_1, \ldots, u^{(0)}_n\}$ that forms an $\eps$-nearly equal norm Parseval frame,
we apply the operator scaling algorithm (which is the natural alternating algorithm in Subsection~\ref{ss:alternating}) to obtain $U^{(2T)} = \{u^{(2T)}_1, \ldots, u^{(2T)}_n\}$.
As proved in Lemma~\ref{l:triangle} in the Subsection~\ref{ss:outline}, we can bound 
\[
\distance(U^{(0)},U^{(2T)}) = 
\sqrt{\sum_{i=1}^n \norm{u^{(2T)}-u^{(0)}}_2^2} 
\leq \sum_{t=1}^{2T} \sqrt{\sum_{i=1}^n \norm{u^{(t)}-u^{(t-1)}}_2^2}
= \sum_{t=1}^{2T} \distance(U^{(t)},U^{(t-1)}).
\]
It is shown in Proposition 5 of~\cite{survey} that $\dist(U^{(t)},U^{(t-1)})) \leq d\eps^2$ if $U^{(t-1)}$ is an $\eps$-nearly equal norm Parseval frame, and so we can bound $\distance(U^{(0)},U^{(2T)})$ by $O(\eps T\sqrt{d})$.
It remains to bound $T$ for $\Delta(U^{(2T)}) \leq \Delta(U^{(0)})/2$.
Using the analysis in Subsection~\ref{ss:capacity},
we need to set $T = \Theta(\poly(n,d)/\sqrt{\Delta(U^{(0)})})$.
It follows from the definition that $\Delta(U^{(0)}) = \Theta(\poly(n,d) \cdot \eps^2)$, and thus we can bound 
\[\distance(U^{(0)},U^{(2T)}) 
= O(\eps T \sqrt{d}) 
= O(\eps \poly(n,d)/\sqrt{\Delta(U^{(0)})})
= O(\poly(n,d)),\]
but the $\eps$ got cancelled.
Note that this bound is worse than the trivial bound that $\dist(U,\F) = O(d)$ as shown in Subsection~\ref{ss:alternating}.
We remark that the analysis as stated in Subsection~\ref{ss:capacity} only holds when $m=n$, but it can be extended to the case when $m \neq n$ as we assumed here.

\subsection{Continuous Operator Scaling} \label{ss:continuous}

The problem of the analysis in the previous subsection is that there are examples in which the alternating steps zigzag (i.e. $\distance(U^{(t+2)},U^{(t)})$ is small but $\distance(U^{(t+2)},U^{(t+1)})$ and $\distance(U^{(t+1)},U^{(t)})$ are large), and so bounding $\distance(U^{(0)},U^{(2T)})$ by $\sum_{t=1}^{2T} \distance(U^{(t)},U^{(t-1)})$ gives a poor bound,
but we do not know how to bound $\distance(U^{(t+2)},U^{(t)})$ directly.

Our idea is to define a continuous version of the operator scaling algorithm so that the two alternating steps are combined into one step and the movement is continuous, so that we can still bound the distance $\distance(U^{(0)},U^{(T)})$ by the total movement $\int_0^T {\sqrt{\sum_{i=1}^n \norm{\frac{d}{dt} u_i^{(t)}}_2^2} dt}$.
We note that the idea of combining two steps into one is also used in previous work (e.g. in Forster's work~\cite{Forster}), but we are not aware of previous work that considers a dynamical system for operator scaling.

We define our dynamical system in the more general operator setting.
There is a time $t$ in the evolution of the matrices, but we will drop the superscript to ease our notation whenever it is clear from the context.

\begin{definition}[size of an operator] \label{d:size}
Given $\U = (U_1, \ldots, U_n)$, let 
\[
s(\U) := \sum_{i=1}^k \norm{U_i}_F^2 = \tr(\sum_{i=1}^k U_i U_i^T) = \tr(\sum_{i=1}^k U_i^T U_i)
\]
be the size of the operator.
We use the shorthand $s$ when the system $\U$ is clear from the context.
\end{definition}

Our dynamical system is defined by the following differential equation.

\begin{definition}[dynamical system from operator scaling] \label{d:dynamical}
The following differential equation describes how $\U^{(t)} = \{U^{(t)}_1, \ldots, U^{(t)}_k\}$ changes over time:
\[
\d U_i := (sI_m - m\sum_{j=1}^k U_j {U_j}^T) U_i + U_i(sI_n - n\sum_{j=1}^k {U_j}^T U_j) \quad {\rm for~} 1 \leq i \leq k.
\]
\end{definition}

Let us informally see that the dynamical system is a continuous version of the operator scaling algorithm.
Consider one step $U_i \gets (\sum_{j=1}^k U_j U_j^T)^{-\frac12} U_i$ in the operator scaling algorithm.
If we move continuously, since $\frac{m}{s} \sum_{j=1}^k U_j U_j^T \approx I$, the dynamical system should update 
\[
U_i \gets (\frac{m}{s} \sum_{j=1}^k U_j U_j^T)^{-\frac12 dt} U_i 
\approx (I_m + \frac12 (I_m - \frac{m}{s} \sum_{j=1}^k U_j U_j^T) dt)U_i
\implies
\d U_i \approx \frac{1}{2} (I_m - \frac{m}{s} \sum_{j=1}^k U_j U_j^T) U_i,
\]
and similarly for another step 
\[
U_i \gets U_i(\frac{n}{s} \sum_{j=1}^k U_j^T U_j)^{-\frac12 dt} 
\approx U_i(I_n + \frac12 (I_n - \frac{n}{s} \sum_{j=1}^k U_j U_j^T))
\implies 
\d U_i \approx \frac{1}{2} U_i (I_n - \frac{n}{s} \sum_{j=1}^k U_j^T U_j).
\]
We arrive at Definition~\ref{d:dynamical} by combining the two steps and scaling appropriately. 
An important property in the definition is that
\[\tr(sI_m - m\sum_{j=1}^k U_j U_j^T)=\tr(sI_n - n\sum_{j=1}^k U_j^T U_j)=0,\] 
which leads to some nice formulas for the dynamical system as described in Subsection~\ref{ss:outline} and Subsection~\ref{ss:formulas} for the analysis to go through.

The following definition is the key parameter in our analysis.
We can think of $\eps$ as an $\ell_{\infty}$-error bound of the input,
and the following $\Delta$ as an $\ell_2$-error bound.
Indeed, we will work with $\Delta$ as the error measure in all our proofs, and only use the relation that $\Delta \leq 2m^2\eps^2$ as shown in Lemma~\ref{l:Delta-eps} to draw the conclusion.
We note that previous work in operator scaling~\cite{gurvits,operator} also uses a very similar quantity as the error measure, but in our definition the size is involved and the normalization is slightly different.

\begin{definition}[$\Delta$ of an operator] \label{d:Delta}
We measure the progress of our dynamical system by the following quantity:
\[
\Delta(\U) = \frac{1}{m} \tr[(sI_m - m\sum_{i=1}^k U_i U_i^T)^2] + \frac{1}{n} \tr[(sI_n - n\sum_{i=1}^k U_i^T U_i)^2],
\]
which is zero if and only if~~$\U$ is doubly balanced. 
The two conditions in~(\ref{e:conditions}) are scaled appropriately so that $m$ and $n$ are symmetric.
We use the shorthand $\Delta^{(t)}$ for $\Delta(\U^{(t)})$ when $\U$ is clear from the context.
\end{definition}

Another motivation for our dynamical system is that it moves in the direction that minimizes $\Delta(\U)$.
We do not include a proof of this claim as it will not be used elsewhere.

\subsection{Total Movement of Dynamical System} \label{ss:outline}

We prove our main technical result of this section in this subsection.
Assuming some formulas for the dynamical system in Subsection~\ref{ss:formulas} and a lower bound on the operator capacity in Subsection~\ref{ss:capLB},
we will prove that given $\U^{(0)}$ that is $\eps$-nearly stochastic, the dynamical system will produce $\U^{(\infty)}$ that is doubly balanced with $\dist(\U^{(0)},\U^{(\infty)}) \leq O(m^2n \eps)$.

First, as stated in Subsection~\ref{ss:discrete}, we will bound the squared distance of $\U$ to the set of doubly balanced operators by the total movement in our dynamical system.

\begin{lemma} \label{l:triangle}
Let $\U^{(0)}$ be the input operator to the dynamical system and $\U^{(T)}$ be the operator in the dynamical system at time $T$, we have
\[
\distance(\U^{(T)},\U^{(0)}) \leq \int_{0}^T \sqrt{\sum_{i=1}^k \norm{\d U_i^{(t)}}_F^2} dt.
\]
\end{lemma}
\begin{proof}
The proof follows from the triangle inequality of the Frobenius norm.
Given $\U = \{U_1, \ldots, U_k\}$, let $U = [U_1 | U_2 | \ldots | U_k]$ be the $m \times nk$ matrix which is the concatenation of the $k$ matrices.
Then,
\[
\distance(\U^{(T)},\U^{(0)})
= \sqrt{\sum_{i=1}^k \norm{U_i^{(T)}-U_i^{(0)}}_F^2}
= \norm{U^{(T)} - U^{(0)}}_F
= \norm{\int_{0}^T \d U^{(t)} dt}_F
\]
\[
\leq \int_0^T \norm{\d U^{(t)}}_F dt
= \int_0^T \sqrt{ \sum_{i=1}^k \norm{\d U_i^{(t)}}_F^2 } dt.
\]
\end{proof}

To analyze the convergence of the dynamical system, the operator capacity will play an important role.
The capacity of a square operator is defined in~\cite{gurvits} and the capacity of a rectangular operator is defined in~\cite{Brascamp-Lieb}.
Our definition of operator capacity is basically the same.
It is normalized in such a way that the capacity is upper bounded by one when the size of the operator is equal to one.

\begin{definition}[operator capacity] \label{d:capacity}
Given an operator $\U = \{U_1, \ldots, U_k\}$ where each $U_i \in \R^{m \times n}$, we define the capacity of $\U$ as
\[
\capa(\U) = \inf_{X \succeq 0} \frac{m \det(\sum_{i=1}^k U_i X U_i^T)^{1/m}}{\det(X)^{1/n}}.
\]
\end{definition}

As in Subsection~\ref{ss:capacity}, we will show an upper bound and a lower bound on operator capacity, and also keep track of the change of capacity in order to argue about the decrease of $\Delta$ over time.
First, we see that the capacity is always upper bounded by the size,
whose proof is basically the same to that of~(\ref{e:capUB}).

\begin{lemma} \label{l:capUB}
Given an operator $\U=\{U_1,\ldots,U_k\}$ where each $U_i \in \R^{m \times n}$, we have
\[\capa(\U) \leq s({\U}).\]
\end{lemma}
\begin{proof}
The proof follows from the AM-GM inequality:
\begin{eqnarray*}
\capa(\U) = \inf_{X \succeq 0} \frac{m \det(\sum_{i=1}^k U_i X U_i^T)^{1/m}}{\det(X)^{1/n}}
\leq m \det(\sum_{i=1}^k U_i U_i^T)^{\frac{1}{m}}
\leq m \Big(\tr(\sum_{i=1}^k U_i U_i^T)/m\Big)^{m \times \frac{1}{m}} = s({\U}).
\end{eqnarray*}
\end{proof}

One can show that $\capa(\U) = s(\U)$ when $\U$ is doubly balanced, but we do not need this fact in our proof.
In the following, we state the facts that we need for our proof.

{\bf Proof Steps:}
It turns out that there are very nice formulas of our dynamical system which can be used to bound the right hand side of Lemma~\ref{l:triangle}.
We will prove 
\begin{enumerate}[(i)]
\item in Lemma~\ref{l:s'} that 
\[\d s^{(t)} = -2\Delta^{(t)},\]
which in particular implies that the size of the operator is decreasing over time;

\item in Lemma~\ref{l:Delta'} that 
\[\d \Delta^{(t)} = -4(\sum_{i=1}^k \norm{\d U_i}_F^2),\]
which in particular implies that $\Delta^{(t)}$ is decreasing over time;

\item in Lemma~\ref{l:cap-unchanged} that $\capa(\U^{(t)})$ is unchanged over time;
\item and in Theorem~\ref{t:cap_lower_bound} that $\capa(\U) \geq s(\U) - mn\sqrt{\Delta(\U)}$, which implies that
\[
s^{(T)} \geq \capa(\U^{(T)}) = \capa(\U^{(t)}) \geq s^{(t)} - mn\sqrt{\Delta^{(t)}} \quad {\rm for~any~} T \geq t \geq 0,
\]
where the first inequality is by Lemma~\ref{l:capUB}.
\end{enumerate}

Using these lemmas, we can analyze the total movement of the dynamical system before $\Delta^{(T)} \leq \Delta^{(0)}/2$.
As in Subsection~\ref{ss:capacity},
we will show that the operator capacity provides an indirect way to argue that $\Delta^{(t)}$ converges to zero in the following proposition.
One can view the following proposition as reducing the total movement bound to the capacity lower bound, and we will make this more explicit in Section~\ref{s:smoothed} (see Proposition~\ref{p:movement-size}).

\begin{proposition} \label{p:half}
For $t \geq 0$ with $\Delta^{(t)} > 0$,
let $T$ be the first time that $\Delta^{(T)} = \Delta^{(t)}/2$.
Then
\[
T \leq t+\frac{mn}{\sqrt{\Delta^{(t)}}} \quad {\rm and} \quad
\dist(\U^{(T)}, \U^{(t)}) \leq 2mn\sqrt{\Delta^{(t)}}.
\]
\end{proposition}
\begin{proof}
The assumption implies that $\Delta^{(\tau)} > \Delta^{(t)}/2$ for $t \leq \tau < T$,
and thus it follows from point~(i) that
\[\frac{d}{d\tau} s^{(\tau)} = -2\Delta^{(\tau)} < -\Delta^{(t)} \quad {\rm for~} t \leq \tau < T.\]
The capacity lower bound thus allows us to conclude that 
\[T \leq t + \frac{mn}{\sqrt{\Delta^{(t)}}},\]
as otherwise 
\[s^{(T)} = s^{(t)} + \int_{t}^T \frac{d}{d\tau} s^{(\tau)} d\tau < s^{(t)} - (T-t)\Delta^{(t)} < s^{(t)} - mn\sqrt{\Delta^{(t)}}, 
\]
contradicting point~(iv).
Therefore,
\begin{eqnarray*}
\distance(\U^{(T)},\U^{(t)}) 
& \leq & \int_{t}^T \sqrt{\sum_{i=1}^k \norm{\frac{d}{d\tau} U_i^{(\tau)}}_F^2} d\tau
\\
& = & 2\int_t^T \sqrt{-\frac{d}{d\tau} \Delta^{(\tau)}} d\tau
\\
& \leq & 2 \sqrt{ \int_t^T (-\frac{d}{d\tau} \Delta^{(\tau)}) d\tau \cdot \int_t^T 1 d\tau}\\
& = & 2 \sqrt{-(\Delta^{(T)}-\Delta^{(t)})(T-t)}\\ 
& \leq & \sqrt{2mn \sqrt{\Delta^{(t)}}},
\end{eqnarray*}
where the first inequality is by Lemma~\ref{l:triangle},
the first equality is by point~(ii),
the second inequality is by Cauchy-Schwarz,
and the last inequality is by the bound on $T$ above and the assumption that $\Delta^{(T)} = \Delta^{(t)}/2$.
Squaring both sides proves the lemma.
\end{proof}

Using this argument repeatedly will give us a decreasing geometric sequence and we can prove the main technical result in this section.

\begin{theorem} \label{t:total-movement}
Given any operator $\U^{(0)} = \{U_1^{(0)}, \ldots, U_k^{(0)}\}$ where $U_i \in \R^{m \times n}$ for $1 \leq i \leq k$,
the dynamical system in Definition~\ref{d:dynamical} will move $\U^{(0)}$ to $\U^{(\infty)}$ such that
\[
\Delta(\U^{(\infty)}) = 0 \quad {\rm and} \quad
\dist(\U^{(\infty)}, \U^{(0)}) \leq O(mn \sqrt{\Delta^{(0)}}).
\]
\end{theorem}
\begin{proof}
If $\Delta^{(0)}=0$, then $\U^{(0)}$ is already doubly balanced and we are done.
Otherwise, for $j \geq 0$,
let $T_j$ be the first time when $\Delta(\U^{(T_j)}) = 2^{-j} \Delta(\U^{(0)})$.
By Proposition~\ref{p:half},
\[
T_j 
\leq \sum_{l=1}^j \frac{mn}{\sqrt{\Delta^{(T_{l-1})}}} 
= \sum_{l=1}^j \frac{mn}{\sqrt{2^{-l} \Delta^{(0)}}} 
= O(\frac{2^{j/2} mn}{\sqrt{\Delta^{(0)}}}),
\]
and so $\Delta(\U^{(t)}) \to 0$ as $t \to \infty$ as $\Delta^{(t)}$ is decreasing over time. 
Furthermore, by Proposition~\ref{p:half},
\[
\distance(\U^{(T_j)},\U^{(0)})
\leq \sum_{l=1}^j \distance(\U^{(T_l)},\U^{(T_{l-1})})
\leq \sum_{l=1}^j \sqrt{2mn \sqrt{\Delta^{(T_{l-1})}}}
\leq O(\sqrt{mn\sqrt{\Delta^{(0)}}}),
\]
where the last inequality follows as it is a sum of a decreasing geometric sequence.
Squaring both sides gives that
\[
\dist(\U^{(\infty)},\U^{(0)}) \leq O(mn \sqrt{\Delta^{(0)}}).
\]
\end{proof}

In Subsection~\ref{ss:formulas}, we will prove the first three items in the proof steps.
In Subsection~\ref{ss:capLB}, we will prove the operator capacity lower bound, which is similar to the proof of~(\ref{e:operator-cap}) in~\cite{operator}, by reducing to the corresponding matrix capacity lower bound.

\subsection{Formulas for the Dynamical System} \label{ss:formulas}

In this subsection, we will prove the first three items in the proof steps in Subsection~\ref{ss:outline}.
Recall from Definition~\ref{d:dynamical} that the dynamical system is
\[
\d U_i := (sI_m - m\sum_{j=1}^k U_j {U_j}^T) U_i + U_i(sI_n - n\sum_{j=1}^k {U_j}^T U_j) \quad {\rm for~} 1 \leq i \leq k,
\]
from Definition~\ref{d:size} that the size is
\[s := \tr(\sum_{i=1}^k U_i U_i^T) = \tr(\sum_{i=1}^k U_i^T U_i),\]
and from Definition~\ref{d:Delta} that 
\[
\Delta = \frac{1}{m} \tr[(sI_m - m\sum_{i=1}^k U_i U_i^T)^2] + \frac{1}{n} \tr[(sI_n - n\sum_{i=1}^k U_i U_i^T)^2].
\]
All the quantities change over time $t$, but we drop the superscript for ease of notation.
We will define some shorthands for the following proofs.
\begin{definition}[operator shorthand] \label{d:shorthand}
Define 
\[B_m = \sum_{i=1}^k U_i U_i^T, \quad C_m = sI_m - mB_m, 
\quad {\rm and} \quad
B_n = \sum_{i=1}^k U_i^T U_i, \quad C_n = sI_n - nB_n.
.\]
Note that $\tr(B_m) = \tr(B_n) = s$ and $\tr(C_m) = \tr(C_n) = 0$,
and $B_m,C_m \in \R^{m \times m}$ and $B_n,C_n \in \R^{n \times n}$ are symmetric.
Also,
\[
\d U_i = C_m U_i + U_i C_n \quad {\rm and} \quad
\Delta = \frac{1}{m} \tr(C_m^2) + \frac{1}{n} \tr(C_n^2).
\]
\end{definition}

\subsubsection*{Formula for the change of $s$}

We are ready to prove point (i) in Subsection~\ref{ss:outline}.

\begin{lemma} \label{l:s'}
\[\d s = -2\Delta.\]
\end{lemma}
\begin{proof}
From the definition of the size and the shorthands,
\[
\d s = \d \tr(B_m) = \tr(\d B_m).
\]
We compute
\begin{eqnarray*}
\d B_m
& = &
\sum_{i=1}^k \big( U_i (\d U_i^T) + (\d U_i) U_i^T \big)
\\
& = &
\sum_{i=1}^k \Big( U_i(C_m U_i + U_i C_n)^T  + (C_m U_i + U_i C_n)U_i^T \Big)
\\
& = &
\sum_{i=1}^k \Big( U_i U_i^T C_m + U_i C_n U_i^T + C_m U_i U_i^T + U_i C_n U_i^T \Big)\\
& = & 
B_m C_m + C_m B_m + 2\sum_{i=1}^k U_i C_n U_i^T.
\end{eqnarray*}
Therefore,
\begin{eqnarray*}
\d s & = & \tr (\d B_m)
= \tr ( B_m C_m + C_m B_m + 2\sum_{i=1}^k U_i C_n U_i^T )
\\
& = & 2\tr (C_m B_m) + 2\sum_{i=1}^k \tr(C_n U_i^T U_i)
\\
& = & 2\tr(C_m B_m) + 2\tr(C_n B_n).
\end{eqnarray*}
On the other hand,
\begin{eqnarray*}
\Delta & = & \frac{1}{m} \tr(C_m^2) + \frac{1}{n} \tr(C_n^2)\\
& = & \frac{1}{m} \tr(C_m (sI_m - mB_m)) + \frac{1}{n} \tr(C_n (sI_n - nB_n))\\
& = & - \tr(C_m B_m) - \tr(C_n B_n),
\end{eqnarray*}
where the last equality holds because $\tr(C_m) = \tr(C_n) = 0$.
\end{proof}

\subsubsection*{Formula for the change of $\Delta$}

We proceed to point~(ii) in Subsection~\ref{ss:outline}.

\begin{lemma} \label{l:Delta'}
\[\d \Delta = -4(\sum_{i=1}^k \norm{\d U_i}_F^2).\]
\end{lemma}
\begin{proof}
From the definition of $\Delta$ and the shorthands,
\begin{eqnarray*}
\d \Delta & = &
\frac{1}{m} \tr[C_m (\d C_m) + (\d C_m) C_m] + \frac{1}{n} \tr[C_n (\d C_n) + (\d C_n) C_n]
\\
& = & \frac{2}{m} \tr[(\d C_m) C_m] + \frac{2}{n} \tr[(\d C_n) C_n]
\\
& = & \frac{2}{m} \tr[\d(sI_m - mB_m) \cdot C_m] + \frac{2}{n} \tr[\d(sI_n - nB_n) \cdot C_n]
\\
& = & -2\tr[(\d B_m) C_m] - 2\tr[(\d B_n) C_n],
\end{eqnarray*}
where the last equality uses that $\tr(C_m) = \tr(C_n) = 0$.
Using the calculation of $\d B_m$ (and similarly $\d B_n$) in Lemma~\ref{l:s'},
we continue and use the cyclic property of trace to get
\begin{eqnarray*}
\d \Delta & = &
-2 \tr[ C_m (B_m C_m + C_m B_m + 2 \sum_{i=1}^k U_i C_n U_i^T) ]
-2 \tr[ C_n (B_n C_n + C_n B_n + 2 \sum_{i=1}^k U_i^T C_m U_i) ]
\\
& = &
-4 \big (\tr [C_m^2 B_m] + \tr[C_n^2 B_n] + 2 \tr[\sum_{i=1}^k C_m U_i C_n U_i^T] \big).
\end{eqnarray*}
On the other hand,
\begin{eqnarray*}
\sum_{i=1}^k \norm{\d U_i}_F^2
& = & \sum_{i=1}^k \tr[(\d U_i^T) (\d U_i)]
= \sum_{i=1}^k \tr[ (C_m U_i + U_i C_n)^T (C_m U_i + U_i C_n)]
\\
& = & 
\sum_{i=1}^k \Big( 
\tr[U_i^T C_m^2 U_i] + \tr[C_n U_i^T C_m U_i] + \tr[U_i^T C_m U_i C_n] + \tr[C_n U_i^T U_i C_n]
\Big)\\
& = & \sum_{i=1}^k \tr[C_m^2 U_i U_i^T] + 2\sum_{i=1}^k \tr[C_m U_i C_n U_i^T]
+ \sum_{i=1}^k \tr[C_n^2 U_i^T U_i]
\\
& = &
\tr[C_m^2 B_m] + \tr[C_n^2 B_n] + 2\tr[\sum_{i=1}^k C_m U_i C_n U_i^T].
\end{eqnarray*}
\end{proof}

\subsubsection*{Capacity Unchanged}

We prove point~(iii) in Subsection~\ref{ss:outline} that the capacity is unchanged over time.
We will show that $\U^{(t)}$ is a scaling of $\U^{(0)}$ at each time $t \geq 0$ as defined in Definition~\ref{d:scaling}, 
and we will argue that the scaling matrices have determinant one and thus the operator capacity is unchanged.
There are different ways to prove this.
One nice way is to use product integration~\cite{Slavik}.
We will prove it directly using elementary calculus starting with the following lemma.

\begin{lemma} \label{l:jacobi}
Suppose $C^{(t)} \in \R^{m \times m}$ is Lipschitz continuous over $t$.
Then the unique solution to the differential equation $\d X^{(t)} = C^{(t)} X^{(t)}$ with initial condition $X^{(0)} = I_m$ satisfies for any $t \ge 0$,
\[
\det(X^{(t)})
= \exp\big(\int_0^t \tr(C^{(\tau)}) d\tau\big).
\]
Similarly, suppose $C^{(t)} \in \R^{n \times n}$ is Lipschitz continuous over $t$.
Then the unique solution to the differential equation $\d Y^{(t)} = Y^{(t)} C^{(t)}$ with initial condition $Y^{(0)} = I_n$ satisfies for any $t \geq 0$,
\[
\det(Y^{(t)})
= \exp\big(\int_0^t \tr(C^{(\tau)}) d\tau\big).
\]
\end{lemma}
\begin{proof}
By the Jacobi's formula, we have
\[
\d \det(X^{(t)})
= \tr[(\d X^{(t)}) \adj(X^{(t)})]
= \tr[C^{(t)} X^{(t)} \adj(X^{(t)})]
= \tr[C^{(t)}] \det(X^{(t)}).
\]
Hence $z^{(t)} = \det(X^{(t)})$ satisfies $\d z^{(t)} = \tr(C^{(t)}) z^{(t)}$ and $z^{(0)} = 1$.
Since $\tr(C^{(t)})$ is Lipschitz continuous over $t$, there is a unique solution to the differential equation $\d z^{(t)} = \tr(C^{(t)}) z^{(t)}$ and $z^{(0)} = 1$ by standard theory (e.g. see Theorem~2.1 of~\cite{Bjo}).
On the other hand, $z^{(t)} = \exp(\int_0^t \tr(C^{(\tau)}) d\tau)$ also satisfies the differential equation as
\[
\d\exp(\int_0^t \tr(C^{(\tau)}) d\tau) 
=\Big( \d \int_0^t \tr(C^{(\tau)}) d\tau \Big) \Big( \exp(\int_0^t \tr(C^{(\tau)}) d\tau) \Big)
= \tr(C^{(t)}) \cdot \exp(\int_0^t \tr(C^{(\tau)}) d\tau).
\]
Therefore,
\[
\det(X^{(t)})
= z^{(t)}
= \exp(\int_0^t \tr(C^{(\tau)}) d\tau).
\]
The proof of the statement for $Y$ is analogous.
\end{proof}

We use the above lemma to prove that $U_i^{(t)}$ in our dynamical system in Definition~\ref{d:dynamical} is always an operator scaling as defined in Definition~\ref{d:scaling}.

\begin{lemma} \label{l:unitary-scaling}
Let $U_i^{(t)} \in \R^{m \times n}$ be the solution to our dynamical system in Definition~\ref{d:dynamical}:
\[\d U_i^{(t)} = C_m^{(t)} U_i^{(t)} + U_i^{(t)} C_n^{(t)}.\]
Then we can write $U_i^{(t)} = X^{(t)} U_i^{(0)} Y^{(t)}$ for $1 \leq i \leq k$ where $X^{(t)} \in \R^{m \times m}$ and $Y^{(t)} \in \R^{n \times n}$ are independent of $i$ and furthermore $\det(X^{(t)}) = \det(Y^{(t)}) = 1$ for all $t \ge 0$.
\end{lemma}
\begin{proof}
Since $\U^{(0)} = \{U_1^{(0)}, \ldots, U_k^{(0)}\}$ is $\eps$-nearly doubly balanced as stated in Definition~\ref{d:epsDS},
each entry of $U_i^{(0)}$ is bounded.
Since $\d U_i^{(t)}$ is a polynomial in the entries of $\U^{(t)}$,
each entry of $\d U_i^{(t)}$ is also bounded.
By standard theory of differential equations (e.g. see Theorem~2.1 of~\cite{Bjo}),
our dynamical system has a unique solution given the initial value $\U^{(0)}$.

We now construct a solution that satisfies the dynamical system.
As $C_m^{(t)}$ and $C_n^{(t)}$ are Lipschitz continuous over $t$,
we can apply Lemma~\ref{l:jacobi} to get
the unique solution $X^{(t)}$ to the differential equation $\d X^{(t)} = C_m^{(t)} X^{(t)}$ with initial value $X^{(0)} = I_m$ satisfying 
\[\det(X^{(t)}) = \exp(\int_0^t \tr(C_m^{(\tau)}) d\tau),\]
and the unique solution $Y^{(t)}$ to the differential equation $\d Y^{(t)} = Y^{(t)} C_n^{(t)}$ with initial value $Y^{(0)} = I_n$ satisfying
\[\det(Y^{(t)}) = \exp(\int_0^t \tr(C_n^{(\tau)}) d\tau).\]
Recall from Definition~\ref{d:shorthand} that our dynamical system is defined in such a way that $\tr(C_m^{(\tau)})=\tr(C_n^{(\tau)})=0$ for every $\tau \geq 0$,
we have the important property that $\det(X^{(t)}) = \det(Y^{(t)}) = 1$ for all $t \geq 0$.
Consider $U_i^{(t)} = X^{(t)} U_i^{(0)} Y^{(t)}$.
It is clear that the differential equation 
\[\d U_i^{(t)} = (\d X^{(t)}) U_i^{(0)} Y^{(t)} + X^{(t)} U_i^{(0)} (\d Y^{(t)})
= C_m^{(t)} X^{(t)} U_i^{(0)} Y^{(t)} + X^{(t)} U_i^{(0)} Y^{(t)} C_n^{(t)} =
C_m^{(t)} U_i^{(t)} + U_i^{(t)} C_n^{(t)}\]
is satisfied with initial value $U_i^{(0)}$.
The lemma thus follows from the uniqueness of the solution of our dynamical system.
\end{proof}

We can conclude that the operator capacity is unchanged over time.
First, we prove an identity of the change of the capacity analogous to~(\ref{e:change}).
The proof of the following lemma is basically the same as the proof of Proposition 2.7 in~\cite{operator,gurvits}, adapted to our definition of capacity.

\begin{lemma} \label{l:cap-change}
Let $\U = \{U_1, \ldots, U_k\}$ and $\V = \{V_1, \ldots, V_k\}$ where $U_i,V_i \in \R^{m \times n}$ for $1 \leq i \leq k$ and $V_i = XU_iY$ for $1 \leq i \leq k$ for some $X \in \R^{m \times m}$ and $Y \in \R^{n \times n}$.
Then
\[
\capa(\V) = \det(X)^{2/m} \det(Y)^{2/n} \capa(\U).
\]
\end{lemma}
\begin{proof}
By the definition of operator capacity in Definition~\ref{d:capacity},
\begin{eqnarray*}
\capa(\V) & = &
\inf_{Z \succeq 0} \frac{\det \Big( \sum_{i=1}^k (X U_i Y) Z (X U_i Y)^T \Big)^{1/m}}{\det(Z)^{1/n}}
\\
& = &
\inf_{Z \succeq 0} \frac{\det \Big( \sum_{i=1}^k X U_i Y Z 
{Y}^T {U_i}^T {X}^T  \Big)^{1/m}}{\det(Z)^{1/n}}\\
& = &
\det(X)^{2/m} \inf_{Z \succeq 0}   \frac{\det \Big( \sum_{i=1}^k U_i Y Z {Y}^T {U_i}^T  \Big)^{1/m}}{\det(Z)^{1/n}}
\\
& = &
\det(X)^{2/m} \inf_{Z \succeq 0}  \frac{\det \Big( \sum_{i=1}^k U_i Z {U_i}^T  \Big)^{1/m}}{\det((Y)^{-1} Z ({Y}^{T})^{-1})^{1/n}}
\\ 
& = &
\det(X)^{2/m} \det(Y)^{2/n} 
\inf_{Z \succeq 0}  \frac{\det \Big( \sum_{i=1}^k U_i Z {U_i}^T  \Big)^{1/m}}{\det(Z)^{1/n}}
\\
& = & \det(X)^{2/m} \det(Y)^{2/n} \capa(\U).
\end{eqnarray*}
\end{proof}

It follows easily from Lemma~\ref{l:cap-change} that the capacity is unchanged over time as $\det(X)=\det(Y)=1$ from Lemma~\ref{l:unitary-scaling}.

\begin{lemma} \label{l:cap-unchanged}
For any $t \ge 0$, we have
\[
\capa(\U^{(t)}) = \capa(\U^{(0)}).
\]
\end{lemma}
\begin{proof}
By Lemma~\ref{l:unitary-scaling}, there exist $X^{(t)} \in \R^{m \times m}$ and $Y^{(t)} \in \R^{n \times n}$ such that $U_i^{(t)} = X^{(t)} U_i^{(0)} Y^{(t)}$.
Therefore, by Lemma~\ref{l:cap-change},
we have $\capa(\U^{(t)}) = \det(X^{(t)})^{2/m} \det(Y^{(t)})^{2/n} \capa(\U^{(0)})$.
The lemma follows from Lemma~\ref{l:unitary-scaling} that
$\det(X^{(t)}) = \det(Y^{(t)}) = 1$.
\end{proof}

\subsection{Lower Bound on Operator Capacity} \label{ss:capLB}

In this subsection, we prove a lower bound on $\capa(\U)$ based on $\Delta(\U)$ in a similar form to~(\ref{e:operator-cap}).
As in~\cite{operator}, we do this by reducing the problem to proving a lower bound on matrix capacity as discussed in Subsection~\ref{ss:capacity} and Subsection~\ref{ss:matrix}.

\begin{remark}[comparsion with previous work]
We remark that all the results in this subsection can be seen as simple variants of the results about capacity in~\cite{gurvits,operator,Brascamp-Lieb}.
The following are some technical differences that do not allow us to use their results directly.
In the (discrete) operator scaling algorithm described in Subsection~\ref{ss:operator}, we can assume that one of the two conditions in Definition~\ref{d:DS} is satisfied and this simplifies the proofs.
In the continuous operator scaling algorithm that we defined in Definition~\ref{d:dynamical}, typically both conditions are not satisfied and so we will need to slightly generalize the proof of (\ref{e:matrix-cap}) and the reduction from the operator capacity to the matrix capacity to prove an analogous statement of (\ref{e:operator-cap}).
Also our normalizations are slightly different and so their proofs needed to be adapted. 
And some definitions (such as size of a matrix, rectangular matrix capacity) are not defined in previous work.
We will explain before each statement what are the differences with previous work.

The new ideas about lower bounding operator capacity are in Section~\ref{s:smoothed}, in which we have developed a technique to analyze capacity using our dynamical system.
\end{remark}

Since the capacity of an operator will be reduced to the capacity of a matrix, let us begin with the corresponding definitions of a matrix.

\begin{definition}[row sum and column sum] \label{d:matrix-sum}
Given a non-negative matrix $A \in \R^{m \times n}$, we define
\[
r_i(A) := \sum_{j=1}^n A_{ij}
\]
to be the $i$-th row sum of $A$ for $1 \leq i \leq m$ and
\[
c_j(A) := \sum_{i=1}^m A_{ij}
\]
to be the $j$-th column sum of $A$ for $1 \leq j \leq n$.
We use the shorthands $r_i$ and $c_j$ when $A$ is clear from the context.
\end{definition}

\begin{definition}[doubly balanced and doubly stochastic matrix] \label{d:DS-matrix}
A matrix $A \in \R^{m \times n}$ is doubly balanced if 
the row sums are the same and the column sums are the same, i.e.
$r_{i_1}(A) = r_{i_2}(A)$ for all $1 \leq i_1, i_2 \leq m$
and $c_{j_1}(A) = c_{j_2}(A)$ for all $1 \leq j_1, j_2 \leq n$.

A matrix $A \in \R^{m \times n}$ is doubly stochastic if it is doubly balanced and all the row sums are one.
\end{definition}

\begin{definition}[size of a matrix] \label{d:matrix-size}
Given a non-negative matrix $A \in \R^{m \times n}$, we define
\[
s(A) := \sum_{i=1}^m \sum_{j=1}^n A_{ij}
\]
to be the size of the matrix.
We use the shorthand $s$ when $A$ is clear from the context.
\end{definition}

The following is our definition of the capacity of a rectangular matrix.
It is normalized in such a way that the capacity is at most one when the size is equal to one.

\begin{definition}[matrix capacity] \label{d:matrix-capacity}
Given a non-negative matrix $A \in \R^{m \times n}$,
we define the capacity of $A$ as
\[
\capa(A) = \inf_{x > 0} \frac{m (\prod_{i=1}^m (Ax)_i)^{1/m}}{(\prod_{i=1}^n x_i)^{1/n}}.
\]
\end{definition}

We also define a measure of how close a matrix is to doubly balanced, which is similar to that in~\cite{operator} but with the size involved and a different normalization.

\begin{definition}[$\Delta$ of a matrix] \label{d:matrix-Delta}
Given a non-negative matrix $A \in \R^{m \times n}$, we define
\[
\Delta(A) = \frac{1}{m} \sum_{i=1}^m (s-mr_i)^2+ \frac{1}{n} \sum_{j=1}^n (s-nc_j)^2,
\]
where $s$ is the size of $A$ and $r_i, c_j$ are row and column sums of $A$.
Note that $\Delta(A)=0$ if and only if $A$ is doubly balanced.
We use the shorthand $\Delta$ when $A$ is clear from the context.
\end{definition}

The following lemma shows an upper bound on $\capa(A)$ using $s(A)$ and characterizes when it is tight.
In the square case, the lower bound $\capa(A) \ge s(A)$ when $A$ is doubly balanced is known in Lemma~3.2 in \cite{operator}.
We use the same idea to prove the result for general rectangular matrices and the proof is basically the same.

\begin{lemma} \label{l:matrix-capUB}
For any matrix $A \in \R^{m \times n}$, we have
\[
\capa(A) \leq s(A),
\]
and when $A$ is doubly balanced then $\capa(A)=s(A)$.
\end{lemma}
\begin{proof}
We first prove the upper bound by using the all-one vector $1_n$ as a test vector so that
\[
\capa(A) 
\leq m \Big(\prod_{i=1}^m (A 1_n)_i\Big)^{1/m}
\leq \sum_{i=1}^m (A 1_n)_i
= \sum_{i=1}^m \sum_{j=1}^n A_{ij} 
= s(A),
\]
where the second inequality is by the AM-GM inequality.

Next we consider a doubly balanced matrix $A$.
By scaling, we assume that $s(A)=1$, so that $\sum_{j=1}^n A_{ij} = 1/m$
and $\sum_{i=1}^m A_{ij} = 1/n$.
We consider the logarithm of the capacity
\begin{eqnarray*}
\log \capa(A) 
& = & \inf_{x>0, x \in \R^n} \Big( \log m + \frac{1}{m} \sum_{i=1}^m \log\big(\sum_{j=1}^n A_{ij}x_j \big) - \frac{1}{n} \sum_{j=1}^n \log x_j \Big)
\\
& = & \inf_{x>0, x \in \R^n} \Big( \frac{1}{m} \sum_{i=1}^m \log\big( \sum_{j=1}^n mA_{ij} x_j \big) - \frac{1}{n} \sum_{j=1}^n \log x_j \Big)
\\
& \geq & \inf_{x>0, x \in \R^n} \Big( \frac{1}{m} \sum_{i=1}^m \sum_{j=1}^n mA_{ij} \log x_j - \frac{1}{n} \sum_{j=1}^n \log x_j \Big)
\\
& = & \inf_{x>0, x \in \R^n} \Big(  \sum_{j=1}^n \sum_{i=1}^m A_{ij} \log x_j - \frac{1}{n} \sum_{j=1}^n \log x_j \Big)
\\
& = & 0,
\end{eqnarray*}
where the inequality is by the concavity of $\log$ and the assumption that $\sum_{j=1}^n m A_{ij} = 1$, and the last equality is by the assumption that $\sum_{i=1}^m A_{ij} = 1/n$.
This implies that $\capa(A)=1$ when $A$ is doubly balanced then $s(A)=1$,
and the lemma follows by reducing to the case when $s(A)=1$ by scaling.
\end{proof}

To prove a capacity lower bound on an operator,
we will reduce it to proving a lower bound on a rectangular matrix,
which will be further reduced to proving a lower bound on a square matrix,
which can be done by modifying previous techniques in~\cite{LSW,gurvits,operator}.

\subsubsection*{Reducing Operator Capacity to Matrix Capacity}

We show a reduction from the operator capacity to the matrix capacity.
This is similar to the proof of Lemma 3.4 in~\cite{operator} but there are some differences that do not allow us to use their result directly.
One difference is that it is assumed in~\cite{operator} that the second condition in Definition~\ref{d:DS} is satisfied, which holds during the execution of the discrete operator scaling algorithm.
In our dynamical system, however, typically both conditions in Definition~\ref{d:DS} are not satisfied and so we need to prove the reduction without this assumption.
Also, we consider rectangular matrices with slightly different definitions for $\Delta$ and capacity.

\begin{proposition} \label{p:reduction-capacity}
Given an operator $\U=\{U_1, \ldots, U_k\}$ with $U_i \in \R^{m \times n}$,
there is a non-negative matrix $A \in \R^{m \times n}$ with
\[\capa(A) \leq \capa(\U), \quad \Delta(A) \leq \Delta(\U), \quad {\rm and~} s(A) = s(\U).\]
Furthermore, if $\sum_{l=1}^k U_l U_l^T = pI_m$ for some $p \geq 0$, then $r_i(A)=p$ for $1 \leq i \leq m$, and similarly if $\sum_{l=1}^k U_l^T U_l = qI_n$ for some $q \geq 0$, then $c_j(A) = q$ for $1 \leq j \leq n$.
\end{proposition}
\begin{proof}
Recall that
\[
\capa(\U) = \inf_{X \succeq 0} \frac{m \det(\sum_{l=1}^k U_l X U_l^T)^{1/m}}{\det(X)^{1/n}}.
\]
Let $X \in \R^{n \times n}$ be an approximate minimizer to this optimization problem such that 
\[
X \succ 0 \quad {\rm and} \quad \frac{m \det(\sum_{l=1}^k U_l X U_l^T)^{\frac{1}{m}}}{\det(X)^{1/n}} \leq \capa(\U) + \delta {\rm~for~some~} \delta>0.
\]
We consider the eigen-decomposition of the positive semidefinite matrices:
\[
X = \sum_{j=1}^n \lambda_j f_j f_j^T \quad {\rm and} \quad 
\sum_{l=1}^k U_l X U_l^T = \sum_{i=1}^m \sigma_i g_i g_i^T,
\]
where $\lambda_1, \ldots, \lambda_n > 0$ are the eigenvalues of $X$ with $f_1, \ldots, f_n \in \R^n$ an orthonormal set of eigenvectors,
and similarly $\sigma_1, \ldots, \sigma_m \geq 0$ are the eigenvalues of $X$ with $g_1, \ldots, g_m \in \R^m$ an orthonormal set of eigenvectors.
Since the determinant of a matrix is equal to the product of its eigenvalues, we have
\[
\frac{m (\prod_{i=1}^m \sigma_i)^{1/m}}{(\prod_{j=1}^n \lambda_j)^{1/n}}
= \frac{m \det(\sum_{l=1}^k U_l X U_l^T)^{\frac{1}{m}}}{\det(X)^{1/n}} \leq \capa(\U) + \delta.
\]
To reduce to matrix capacity, we write down the linear transformation $A$ from $\{\lambda_j\}_{j=1}^n$ to $\{\sigma_i\}_{i=1}^m$.
It follows from the eigen-decompositions that
\[
\sum_{i=1}^m \sigma_i g_i g_i^T =
\sum_{l=1}^k U_l (\sum_{j=1}^n \lambda_j f_j f_j^T) U_l^T =
\sum_{j=1}^n \lambda_j \sum_{l=1}^k U_l f_j f_j^T U_l^T.
\]
As $\{g_i\}_{i=1}^m$ is an orthonormal basis,
by multiplying both sides by $g_i^T$ on the left and $g_i$ on the right,
we see that
\[
\sigma_i = \sum_{j=1}^n \lambda_j \Big( g_i^T (\sum_{l=1}^k U_l f_j f_j^T U_l^T) g_i \Big) \quad {\rm for~} 1 \leq i \leq m. 
\]
Let $A \in \R^{m \times n}$ be the matrix with
\begin{equation} \label{e:A}
A_{ij} := g_i^T (\sum_{l=1}^k U_l f_j f_j^T U_l^T) g_i \quad {\rm for~} 1 \leq i \leq m, 1 \leq j \leq n.
\end{equation}
Note that $A$ is non-negative as $\sum_{l=1}^k U_l f_j f_j^T U_l^T \succeq 0$.
Let $\sigma \in \R^m$ be a vector with the $i$-th entry being $\sigma_i$ and $\lambda \in \R^n$ be a vector with the $j$-th entry being $\lambda_j$.
It follows from the definition of $A$ that
\[
\sigma = A \lambda.
\]
By the definition of matrix capacity in Definition~\ref{d:matrix-capacity} and
using $\lambda$ as a test vector, we get
\[
\capa(A) 
= \inf_{x > 0} \frac{m (\prod_{i=1}^m (Ax)_i)^{1/m}}{(\prod_{i=1}^n x_i)^{1/n}} 
\leq \frac{m (\prod_{i=1}^m (A \lambda)_i)^{1/m}}{(\prod_{j=1}^n \lambda_j)^{1/n}}
= \frac{m (\prod_{i=1}^m \sigma_i)^{1/m}}{(\prod_{j=1}^n \lambda_j)^{1/n}}
\leq \capa(\U) + \delta.
\]

Next we check the second claim that $\Delta(A) \leq \Delta(\U)$.
The $i$-th row sum of $A$ is
\[
r_i = \sum_{j=1}^n A_{ij} 
= \sum_{j=1}^n g_i^T \Big(\sum_{l=1}^k U_l f_j f_j^T U_l^T \Big) g_i
= g_i^T \Big( \sum_{l=1}^k U_l \big( \sum_{j=1}^n f_j f_j^T \big) U_l^T \Big) g_i
= g_i^T \Big( \sum_{l=1}^k U_l U_l^T \Big) g_i,
\]
where the last equality is because $\{f_j\}_{j=1}^n$ is an orthonormal basis.
In particular, if $\sum_{l=1}^k U_l U_l^T = pI_m$, then $r_i = p\norm{g_i}^2 = p$.
The $j$-th column sum of $A$ is
\begin{eqnarray*}
c_j 
& = & \sum_{i=1}^m A_{ij}
= \sum_{i=1}^m g_i^T \Big(\sum_{l=1}^k U_l f_j f_j^T U_l^T \Big) g_i
= \sum_{i=1}^m \tr[ \Big(\sum_{l=1}^k U_l f_j f_j^T U_l^T \Big) g_i g_i^T ]
\\
& = & \tr[ \Big(\sum_{l=1}^k U_l f_j f_j^T U_l^T \Big) \Big( \sum_{i=1}^m g_i g_i^T \Big)]
= \tr[ \Big(\sum_{l=1}^k U_l f_j f_j^T U_l^T \Big)] 
= f_j^T \Big( \sum_{l=1}^k U_l^T U_l \Big) f_j,
\end{eqnarray*}
where the second last equality is because $\{g_1\}_{i=1}^m$ is an orthonormal basis and the last equality is by the cyclic property of trace.
In particular, if $\sum_{l=1}^k U_l^T U_l = qI_n$, then $c_j = q\norm{f_j}^2 = q$.
The size of the matrix $A$ is
\[
s(A) = \sum_{i=1}^m r_i = \sum_{i=1}^m g_i^T \Big(\sum_{l=1}^k U_l U_l^T\Big) g_i
= \sum_{i=1}^m \tr[\Big( \sum_{l=1}^k U_l U_l^T \Big) g_i g_i^T]
= \tr[\Big( \sum_{l=1}^k U_l U_l^T \Big)] = s(\U),
\]
equal to the size of the operator $\U$ in Definition~\ref{d:size}.
Since $s(A) = s(\U)$, we will just use $s := s(A) = s(\U)$ in the following.
By Definition~\ref{d:matrix-Delta} and the above identities,
\begin{eqnarray*}
\Delta(A) 
& = & \frac{1}{m} \sum_{i=1}^m (s-mr_i)^2 + \frac{1}{n} \sum_{j=1}^n (s-nc_j)^2
\\
& = &
\frac{1}{m} \sum_{i=1}^m \Big(s - m g_i^T \Big( \sum_{l=1}^k U_l U_l \Big) g_i \Big)^2
+ \frac{1}{n} \sum_{j=1}^n \Big(s - n f_j^T \Big( \sum_{l=1}^k U_l U_l \Big) f_i \Big)^2
\\
& = &
\frac{1}{m} \sum_{i=1}^m \Big(g_i^T \Big( sI_m - m\sum_{l=1}^k U_l U_l \Big) g_i \Big)^2
+ \frac{1}{n} \sum_{j=1}^n \Big(f_j^T \Big( sI_n - n\sum_{l=1}^k U_l U_l \Big) f_i \Big)^2
\\
& \leq &
\frac{1}{m} \sum_{i=1}^m \Big(g_i^T \Big( sI_m - m\sum_{l=1}^k U_l U_l \Big)^2 g_i \Big)
+ \frac{1}{n} \sum_{j=1}^n \Big(f_j^T \Big( sI_n - n\sum_{l=1}^k U_l U_l \Big)^2 f_i \Big)
\\
& = &
\frac{1}{m} \tr[\Big( sI_m - m\sum_{l=1}^k U_l U_l \Big)^2]
+ \frac{1}{n} \tr[\Big( sI_n - n\sum_{l=1}^k U_l U_l \Big)^2]
\\
& = &
\Delta(\U),
\end{eqnarray*}
where the inequality follows from $g_i g_i^T \preceq I_m$ as $\norm{g_i}_2^2 \leq 1$, the second last equality follows from $\sum_{i=1}^m g_i g_i^T = I_m$ as $\{g_i\}_{i=1}^m$ is an orthonormal basis and the cyclic property of trace,
and the last equality is by the definition of $\Delta(\U)$ in Definition~\ref{d:Delta}.

For each $l \in \mathbb N$, let $A_l$ be the above constructed matrix when $\delta = 1 / l$, so that $\capa(A_l) \le \capa(\U) + 1 / l$, $\Delta(A_l) \le \Delta(\U)$, and $s(A_l) = s(\U)$.
Since $s(A_l) = s(\U)$ for all $i$, $A_l$ is uniformly bounded, and hence by the Bolzano-Weierstrass Theorem we can obtain a convergent subsequence $A_{l_i}$.
Let $A = \lim_{i \to \infty} A_{l_i}$.
As $\capa(B)$, $\Delta(B)$ and $s(B)$ are all continuous functions on $B$, we have 
\[\capa(A) = \lim_{i \to \infty} \capa(A_{l_i}) \le \capa(\U), 
\quad {\rm and} \quad \Delta(A) \leq \Delta(\U) 
\quad {\rm and} \quad s(A) = s(\U).
\]
\end{proof}

Proposition~\ref{p:reduction-capacity} allows us to establish operator capacity lower bound by proving matrix capacity lower bound, which is usually a relatively simpler problem.
For example, suppose we have a matrix capacity lower bound
$\capa(A) \geq s - mn\sqrt{\Delta(A)/2}$ as we will prove in Proposition~\ref{p:matrix-cap},
then from Proposition~\ref{p:reduction-capacity} we have
\[
\capa(\U) \geq \capa(A) \geq s - mn\sqrt{\Delta(A)/2} \geq s - mn\sqrt{\Delta(\U)/2}.\]
Hence, we will focus on proving a lower bound on matrix capacity.

\subsubsection*{Reducing Rectangular Matrices to Square Matrices}

To prove a lower bound on the capacity of a rectangular matrix,
we will further reduce it to proving a lower bound on the capacity of a square matrix.
Given a rectangular matrix $A \in \R^{m \times n}$,
we will construct a square matrix $B \in \R^{mn \times mn}$ such that $\capa(A)=\capa(B)$ and $\Delta(A)=\Delta(B)$ and $s(A)=s(B)$.
Our construction uses the tensor product which is defined as follows.

\begin{definition}[tensor product]
Let $X \in \mathbb R^{m \times n}$ and $Y \in \mathbb R^{p \times q}$.
The tensor product $X \otimes Y$ is the $mp \times nq$ block matrix:
\[
X \otimes Y
=
\begin{pmatrix}
X_{11} Y & \cdots & X_{1n} Y \\
\vdots & \ddots & \vdots \\
X_{m1} Y & \cdots & X_{mn} Y
\end{pmatrix}
.
\]
We index the rows by pairs of integers $(i, j) \in [m] \times [p]$ and the columns by pairs of integers $(k, l) \in [n] \times [q]$, so that we have $(X \otimes Y)_{(i, k), (j, l)} = X_{ij} Y_{kl}$.
We usually omit the brackets and write $(X \otimes Y)_{ik, jl}$ instead.
\end{definition}

Let $J_{n \times m}$ be the all-one $n \times m$ matrix.
We will consider the $mn \times mn$ dimensional matrix $B := A \otimes \frac1{mn} J_{n \times m}$.
The following lemma implies that $\capa(A)=\capa(B)$ and $\Delta(A) = \Delta(B)$ and $s(A)=s(B)$.

We note that a more general result is proved in~\cite{Brascamp-Lieb} for rectangular operators.
It is possible to reduce the matrix case to the operator case and use the result in~\cite{Brascamp-Lieb} to prove the following lemma.
Instead of presenting the reduction, 
we present a simpler direct proof for this special case.

\begin{lemma} \label{l:reduction-square}
Let $A \in \mathbb R^{m \times n}$ be a rectangular matrix and $J_{p \times q} \in \mathbb R^{p \times q}$ be the all-one matrix.
Then 
\[\capa(A) = \capa(A \otimes \frac1{pq} J_{p \times q})
\quad {\rm and} \quad
\Delta(A) = \Delta(A \otimes \frac{1}{pq} J_{p \times q})
\quad {\rm and} \quad
s(A) = s(A \otimes \frac{1}{pq} J_{p \times q}).
\]
\end{lemma}
\begin{proof}
Let $B := A \otimes \frac1{pq} J_{p \times q}$.
We will prove that $\capa(B) \leq \capa(A)$ and $\capa(A) \leq \capa(B)$.
The first part holds for any tensor product, while the second part uses that $B$ is a tensor product of $A$ and a scaled version of $J$.

To prove the first part, we will prove that for any two rectangular matrices $A \in \R^{m \times n}$ and $C \in \R^{p \times q}$, we have that $B := A \otimes C$ satisfies $\capa(B) \leq \capa(A) \times \capa(C)$.
Fix $\delta > 0$.
Let $x \in \R^n$ and $y \in \R^q$ be approximate minimizers in the capacity of $A$ and $C$, that is
\[
\capa(A) = \inf_{z>0} \frac{m\prod_{i = 1}^m (Az)_i^{1/m}}{\prod_{j = 1}^n z_j^{1/n}} \leq \frac{m\prod_{i = 1}^m (Ax)_i^{1/m}}{\prod_{j = 1}^n x_j^{1/n}}
\leq \capa(A) + \delta,
\]
and similarly
\[
\capa(C) = 
\inf_{z>0} \frac{p\prod_{k = 1}^p (Cz)_k^{1/p}}{\prod_{l = 1}^q z_l^{1/q}}
\leq \frac{p \prod_{k = 1}^p (Cy)_k^{1/p}}{\prod_{l = 1}^q y_l^{1/q}}
\le \capa(B) + \delta.
\]
By considering the vector $x \otimes y \in \R^{nq}$, we have
\begin{align*}
\capa(A \otimes C)
& = \inf_{z \in \R^{nq}} \frac{mp \prod_{i = 1}^{m} \prod_{k=1}^p ((A \otimes C)z)_{ik}^{1/{mp}}}{\prod_{j = 1}^{n} \prod_{l=1}^q z_{jl}^{1/{nq}}} 
\le \frac{mp \prod_{i = 1}^m \prod_{k = 1}^p ((A \otimes C)(x \otimes y))_{ik}^{1/{mp}}}{\prod_{j = 1}^n \prod_{l = 1}^q (x \otimes y)_{jl}^{1/{nq}}} 
\\
& = \frac{mp \prod_{i = 1}^m \prod_{k = 1}^p (Ax \otimes Cy)_{ik}^{1/{mp}}}{\prod_{j = 1}^n \prod_{l = 1}^q (x \otimes y)_{jl}^{1/{nq}}} 
= \Big(\frac{m \prod_{i = 1}^m (Ax)_i^{1/m}}{\prod_{j = 1}^n x_j^{1/n}}\Big) \times \Big(\frac{p \prod_{k = 1}^p (Cy)_k^{1/p}}{\prod_{l = 1}^q y_l^{1/q}}\Big) 
\\
& \le (\capa(A) + \delta) \times (\capa(C) + \delta).
\end{align*}
Taking $\delta \to 0$ proves that $\capa(B) \leq \capa(A) \times \capa(C)$.
Since $\capa(\frac{1}{pq} J_{p \times q}) \leq s(\frac{1}{pq} J_{p \times q}) = 1$ by Lemma~\ref{l:matrix-capUB},
we have $\capa(B) \leq \capa(A)$, proving the first part.

To prove the second part, 
we show that $\capa(B) + \delta \ge \capa(A)$ for any $\delta > 0$.
Let $x \in \mathbb R^{n \times q}$ be an approximate minimizer to capacity of $B$ such that
\[
\capa(B) = 
\inf_{z>0} \frac{mp\prod_{i = 1}^m \prod_{k = 1}^p (Bz)_{ik}^{1/{mp}}}{\prod_{j = 1}^n \prod_{l = 1}^q z_{jl}^{1/{nq}}}
\leq \frac{mp\prod_{i = 1}^m \prod_{k = 1}^p (Bx)_{ik}^{1/{mp}}}{\prod_{j = 1}^n \prod_{l = 1}^q x_{jl}^{1/{nq}}}
\le \capa(B) + \delta.
\]
We define $\bar x \in \mathbb R^n$ be such that for $1 \leq j \leq n$,
\[
\bar x_j = \frac1q \sum_{l = 1}^q x_{jl}.
\]
We use the special property of $B$ to show that
\begin{align*}
(Bx)_{ik}
& = ((A \otimes \frac1{pq} J_{p \times q})x)_{ik} 
= \sum_{j = 1}^n \sum_{l = 1}^q (A \otimes \frac1{pq} J_{p \times q})_{ik, jl} \cdot x_{jl} \\
& = \sum_{j = 1}^n \sum_{l = 1}^q (A_{ij}) \frac1{pq} x_{jl} 
= \frac1p \sum_{j = 1}^n A_{ij} \big(\frac1q \sum_{l = 1}^q x_{jl}\big)
= \frac1p \sum_{j = 1}^n A_{ij} \bar x_j 
= \frac1p (A \bar x)_i.
\end{align*}
Therefore, using $\bar x$ as a test vector for capacity of $A$, we have
\begin{align*}
\capa(B) + \delta & 
\geq mp \frac{\prod_{i = 1}^m \prod_{k = 1}^p (Bx)_{ik}^{1/{mp}}}{\prod_{j = 1}^n \prod_{l = 1}^q x_{jl}^{1/{nq}}}
 = m \frac{\prod_{i = 1}^m (A \bar x)_{i}^{1/m}}{\prod_{j = 1}^n \prod_{l = 1}^q x_{jl}^{1/{nq}}}
\ge m \frac{\prod_{i = 1}^m (A \bar x)_{i}^{1/m}}{\prod_{j = 1}^n \bar x_j^{1/n}}
\ge \capa(A),
\end{align*}
where the second last inequality follows from the AM-GM inequality
that $\bar x_j = \frac{1}{q} \sum_{l=1}^q x_{jl} \ge \prod_{l = 1}^q x_{jl}^{1/q}$ for $1 \leq j \leq n$.
Taking $\delta \to 0$ proves that $\capa(A) \leq \capa(B)$, and thus $\capa(A)=\capa(B)$.

Next we prove that $\Delta(A) = \Delta(B)$, where we only need the property that $\frac{1}{pq} J_{p \times q}$ is a doubly balanced matrix with $s(\frac{1}{pq} J_{p \times q})=1$.
Let $B:= A \otimes C$ where $C$ is a doubly balanced matrix with $s(C)=1$.
Let $1_d$ be the $d$-dimensional all-one vector.
Note that $C 1_q = \frac{1}{p} 1_p$ and $1_p^T C = \frac{1}{q} 1_q^T$.
For any $1 \leq i \leq m$ and $1 \leq k \leq p$, the row sum
\[
r_{ik}(B)
= (B (1_n \otimes 1_q))_{ik}
= ((A 1_n) \otimes (C 1_q))_{ik}
= \frac1p (A 1_n)_i
= \frac1p r_i(A).
\]
Similarly, for any $1 \leq j \leq n$ and $1 \leq l \leq q$, the column sum
\[
c_{jl}(B)
= ((1_m \otimes 1_p)^T B)_{jl}
= ( (1_m^T A) \otimes (1_p^T C))_{jl}
= \frac1q (1_m^T A)_j
= \frac1q c_j(A).
\]
Also $s(B) = \sum_{i = 1}^m \sum_{k = 1}^p r_{ik}(B) = \sum_{i = 1}^m r_i(A) = s(A)$.
Hence,
\begin{align*}
\Delta(B)
& = \frac1{mp} \sum_{i = 1}^m \sum_{k = 1}^p (s(B) - mp \cdot r_{ik}(B))^2 + \frac1{nq} \sum_{j = 1}^n \sum_{l = 1}^q (s(B) - nq \cdot c_{jl}(B))^2 \\
& = \frac1{mp} \sum_{i = 1}^m \sum_{k = 1}^p (s(B) - m \cdot r_i(A))^2 + \frac1{nq} \sum_{j = 1}^n \sum_{l = 1}^q (s(B) - n \cdot c_j(A))^2 \\
& = \frac1m \sum_{i = 1}^m (s(A) - m r_i(A))^2 + \frac1n \sum_{j = 1}^n (s(A) - n c_j(A))^2 \\
& = \Delta(A)
\end{align*}
\end{proof}

Using Lemma~\ref{l:reduction-square}, we can reduce the problem of proving a lower bound on $\capa(A)$ using $\Delta(A)$ for a rectangular matrix $A$ to proving a lower bound on $\capa(B)$ using $\Delta(B)$ for a square matrix $B$, where $B := A \otimes \frac{1}{mn} J_{n \times m}$.
Hence we will focus on proving a lower bound on matrix capacity for square matrices.

\begin{remark}
\label{r:matrix_size_depends_on_gcd}
We can also take $B$ to be a smaller matrix when $m$ and $n$ are not relatively prime.
Suppose $g = \gcd(m, n)$ is the greatest common divisor of $m$ and $n$.
Then we would set $B$ to be $A \otimes \frac{g^2}{mn} J_{n / g \times m / g}$, which is an $(mn / g) \times (mn / g)$ square matrix.
This will be a more efficient reduction which implies stronger bounds in our theorems.
\end{remark}

\subsubsection*{Lower Bound on Matrix Capacity for Square Matrices}

The following proofs follow the same approach in~\cite{operator}.
Again we could not directly apply their proofs, as they assumed that the row sums are equal, which holds in discrete operator scaling but not in continuous operator scaling.

We will use the following well-known facts from~\cite{gurvits,LSW}.
\begin{fact}[\cite{gurvits,LSW}] \label{f:permanent}
Given an $n \times n$ non-negative matrix $B$,
the permanent of $B$ is defined as
\[
\per(B) := \sum_{\pi \in S_n} \prod_{i=1}^n B_{i,\pi(i)},
\]
where $\pi$ is over all permutations of $n$ elements.
It is known that $\per(B) > 0$ if and only if $\capa(B) > 0$.

One can consider the bipartite graph $G=(U,V;E)$ associated to $B$,
where $X=\{1,\ldots,n\}$ and $Y=\{1,\ldots,n\}$ and $ij \in E$ if and only $B_{ij}>0$.
Since $B$ is non-negative,
it follows from the definition that $\per(B)>0$ if and only if $G$ has a perfect matching.
By Hall's theorem for bipartite matching, $\per(B)=0$ if and only if there exist subsets $X' \subseteq X$ and $Y' \subseteq Y$ such that $|X'|+|Y'|>n$ and $B_{ij}=0$ for all $i \in X'$ and $j \in Y'$.
\end{fact}

We first prove a lower bound on $\Delta$ when the capacity is zero,
without the assumption that $A$ is row-balanced.
The constant is slightly better than that in~\cite{LSW} and is tight,
and the proof is somewhat different.

\begin{lemma} \label{l:dist_lower_bound}
Let $A \in \R^{n \times n}$ be a non-negative matrix with $s = \sum_{ij} A_{ij} = n^2$.
If $\capa(A) = 0$, then $\Delta(A) \ge 2n^2$.
\end{lemma}
\begin{proof}
By the first part of Fact~\ref{f:permanent},
we have $\capa(A) = 0$ if and only if $\per(A)=0$.
By the second part of Fact~\ref{f:permanent}, 
$\per(A) = 0$ if and only if there exist 
\[U \subseteq [n] \quad {\rm and} \quad V \subseteq [n] \quad {\rm such~that} \quad
|U| + |V| > n \quad {\rm and} \quad A_{U, V}=0.
\]
Let $T := |U| + |V| > n$.
The claim in this lemma is clearly true when $|U| = n$ or $|V| = n$, as this implies that some $r_i$ or some $c_j$ is equal to zero.
So in the following we assume $|U| < n$ and $|V| < n$.
From Definition~\ref{d:matrix-Delta} and using our assumption that $s=n^2$,
\begin{align*}
\Delta(A)
& = \frac1n \sum_{i=1}^n (s - n r_i)^2 + \frac1n \sum_{j=1}^n (s - n c_j)^2 
\\
& = n \sum_{i=1}^n (n - r_i)^2 + n \sum_{j=1}^n (n - c_j)^2 
\\
& = n \sum_{i \in U} (n - r_i)^2 + n \sum_{i \not \in U} (n - r_i)^2 
+ n \sum_{j \in V} (n - c_j)^2 + n \sum_{j \not \in V} (n - c_j)^2 
\\
& \ge \frac n{|U|} \Big(\sum_{i \in U} (n - r_i)\Big)^2 + \frac n{n - |U|} \Big(\sum_{i \not \in U} (n - r_i)\Big)^2 + \frac n{|V|} \Big(\sum_{j \in V} (n - c_j)\Big)^2 + \frac n{n - |V|} \Big(\sum_{j \not \in V} (n - c_j)\Big)^2,
\end{align*}
where the inequality follows from $n \sum_{i = 1}^n x_i^2 \ge (\sum_{i = 1}^n x_i)^2$ by the Cauchy-Schwarz inequality.
To bound the right hand side of the above inequality,
we divide the non-zero entries of $A$ into three groups.
Let 
\[\alpha = \sum_{i \in U, j \notin V} A_{ij}, \quad \beta = \sum_{i \notin U, j \in V} A_{ij}, \quad {\rm and~} \gamma = \sum_{i \not \in U, j \not \in V} A_{ij}
.\]
Then $\alpha + \beta + \gamma = s = n^2$ and the right hand side of the above inequality can be written as
\begin{align*}
&~\frac n{|U|} (n |U| - \alpha)^2 + \frac n{n - |U|} (n (n - |U|) - (\beta + \gamma))^2 + \frac n{|V|} (n |V| - \beta)^2 + \frac n{n - |V|} (n (n - |V|) - (\alpha + \gamma))^2 
\\
= &~\frac n{|U|} (n |U| - \alpha)^2 + \frac n{n - |U|} (n |U| - \alpha)^2 + \frac n{|V|} (n |V| - \beta)^2 + \frac n{n - |V|} (n |V| - \beta)^2.
\end{align*}
Since $n(\frac1x + \frac1{n-x}) = (\frac1x + \frac1{n - x}) (x + (n - x)) \ge 4$ by Cauchy-Schwarz,
the above line is at least
\[
4 (n |U| - \alpha)^2 + 4 (n |V| - \beta)^2
\geq 2 (n |U| - \alpha + n |V| - \beta)^2
\geq 2 (nT - n^2)^2 
\geq 2n^2,
\]
where the first inequality is by $2a^2+2b^2 \geq (a+b)^2$,
the second inequality is by the definition that $T=|U|+|V|$ and $n^2=\alpha+\beta+\gamma \geq \alpha+\beta$,
and the final inequality is by $T > n$ which follows from the assumption that $\capa(A)=0$ as we argued in the beginning of this proof.
\end{proof}

We are ready to derive a lower bound for the matrix capacity.
The following proof is similar to the proof of Lemma~3.2 in~\cite{operator}, which is based on an argument in Claim~3.3 of~\cite{LSW}.

\begin{proposition} \label{p:matrix-cap}
If $A \in \mathbb R^{n \times n}$ with $s=\sum_{ij} A_{ij}$, then
\[
\capa(A) \geq s - n\sqrt{\frac{\Delta(A)}2}.
\]
\end{proposition}
\begin{proof}
We first assume that $s=n^2$ and prove that $\capa(A) \geq n^2 - n\sqrt{\Delta(A)/2}$, and then we derive the general case by scaling.
We write $A$ as a sum of a doubly balanced matrix and a non-balanced part
such that
\[A = \lambda B + (1 - \lambda) C {\rm~for~} \lambda \in [0, 1], 
\quad \sum_{ij} B_{ij} = \sum_{ij} C_{ij} = n^2, 
\quad B {\rm~is~doubly~balanced~and~} \capa(C)=0.
\]
Such a decomposition exists by repeatedly removing permutations from $A$ as in~\cite{LSW,operator}.

Since $B$ is doubly balanced and $\sum_{ij} B_{ij} = n^2$,
we have $r_i(B) = n$ for $1 \leq i \leq n$ and $c_j(B) = n$ for $1 \leq j \leq n$.
This implies that
\begin{eqnarray*}
\Delta(A) & = & 
\frac{1}{n} \sum_{i=1}^n \Big(n^2 -  n\big(\lambda r_i(B) + (1-\lambda) r_i(C) \big) \Big)^2 + \frac{1}{n} \sum_{j=1}^n \Big( n^2 -  n\big( \lambda c_j(B) + (1-\lambda) c_j(C) \big) \Big)^2
\\
& = &
\frac{1}{n} \sum_{i=1}^n \Big((1-\lambda)(n^2 - nr_i(C)) \Big)^2 + \frac{1}{n} \sum_{j=1}^n \Big( (1-\lambda)(n^2 - nc_j(C)) \Big)^2
\\
& = &
(1-\lambda)^2 \Delta(C) 
\\
& \geq & 
2(1-\lambda)^2 n^2,
\end{eqnarray*}
where the inequality is from Lemma~\ref{l:dist_lower_bound} as $\sum_{ij} C_{ij}=n^2$ and $\capa(C)=0$.
This implies that 
\[\lambda \ge 1 - \sqrt{\frac{\Delta(A)}{2n^2}}.\] 
It follows from Definition~\ref{d:matrix-capacity} that
\[
\capa(A)
= \capa(\lambda B + (1-\lambda)C)
\geq \capa(\lambda B)
= \lambda \capa(B)
= \lambda n^2
\ge n^2(1 - \sqrt{\frac{\Delta(A)}{2n^2}})
= n^2 - n\sqrt{\frac{\Delta(A)}{2}},
\]
as $\capa(B)=s(B)=n^2$ by Lemma~\ref{l:matrix-capUB}.
This proves the lemma for the case when $s=n^2$.
For the general case, we can reduce to the special case by considering 
$(n^2 / s) A$ and get
\[
\capa(A)
= \frac s{n^2} \capa(\frac{n^2}s A)
\geq \frac s{n^2} \Big(n^2 - n\sqrt{\frac{\Delta((n^2/s)A)}{2}}\Big)
= s - \frac{s}{n} \frac{n^2}{s} \sqrt{\frac{\Delta(A)}{2}}
= s - n\sqrt{\frac{\Delta(A)}2}.
\]
\end{proof}

The matrix capacity lower bound for rectangular matrices follows, which we will use in Section~\ref{s:smoothed}.

\begin{proposition} \label{p:matrix-capLB}
If $A \in \mathbb R^{m \times n}$ with $s=\sum_{ij} A_{ij}$, then
\[
\capa(A) \geq s - mn\sqrt{\frac{\Delta(A)}2}.
\]
\end{proposition}
\begin{proof}
Given a non-negative matrix $A \in \R^{m \times n}$,
we apply the reduction in Lemma~\ref{l:reduction-square} to construct a non-negative square matrix $B \in \R^{mn \times mn}$ with $\capa(B)=\capa(A)$, $\Delta(B)=\Delta(A)$ and $s(B)=s(A)$.
Applying Proposition~\ref{p:matrix-cap} on $B$, we have
\[
\capa(A) = \capa(B) \geq s(B) - mn\sqrt{\frac{\Delta(B)}{2}} = s(A)-mn\sqrt{\frac{\Delta(A)}{2}}.
\]
\end{proof}

We can finally prove point (iv) in Subsection~\ref{ss:outline}.

\begin{theorem} \label{t:cap_lower_bound}
Let $\U = \{U_1, \ldots, U_k\}$ with $U_i \in \R^{m \times n}$ for $1 \leq i \leq k$.
\[
\capa(\U) \geq s(\U) - mn\sqrt{\frac{\Delta(\U)}{2}}.
\]
\end{theorem}
\begin{proof}
Given $\U$, 
we apply the reduction in Proposition~\ref{p:reduction-capacity}
to construct a non-negative matrix $A \in \R^{m \times n}$ with $\capa(A) \leq \capa(\U)$ and $\Delta(A) \leq \Delta(\U)$ and $s(A) = s(\U)$.
Applying Proposition~\ref{p:matrix-capLB}, we have
\[
\capa(\U) \geq \capa(A) \geq s(A) - mn\sqrt{\frac{\Delta(A)}{2}} \geq s(\U) - mn\sqrt{\frac{\Delta(\U)}{2}}.
\]
\end{proof}

\subsubsection*{Tight Example}

We give an example to show that the capacity lower bound $\capa(A) \geq s(A) - mn\sqrt{\frac{\Delta(A)}{2}}$ is tight.

\begin{lemma} \label{l:tight}
There is an infinite sequence of $m_k \times n_k$ matrices $A_k$ with 
\[\capa(A_k)=0, \quad s(A_k)=1, \quad {\rm and} \quad \Delta(A_k)=(2+o(1))\frac{1}{m_k^2n_k^2}.
\]
\end{lemma}
\begin{proof}
Let $A_k$ be a $(2k - 1) \times (2k + 1)$ matrix with block structure:
\[
A_k
=
\begin{pmatrix}
0_{k \times k} & x J_{k \times (k + 1)} \\
y J_{(k - 1) \times k} & 0_{(k - 1) \times (k + 1)},
\end{pmatrix}
\]
where $J_{m \times n}$ is the $m \times n$ all one matrix.
It is not difficult to see that $A_k$ has capacity zero using the results in Subsection~\ref{ss:capLB}.
We will choose $x$ and $y$ so that $s(A_k) = 1$ and $\Delta(A_k)$ is minimized.
See Appendix~\ref{a:tight} for the remaining details.
\end{proof}

\subsection{A Bound on the Operator Paulsen Problem} \label{ss:operator-Paulsen}

We will give a bound on the operator Paulsen problem by using our dynamical system.
First, we preprocess the input so that we can obtain a bound on $\Delta^{(0)}$ based on the error $\eps$ as defined in Definition~\ref{d:epsDS}.
Then, we apply our dynamical system to obtain a doubly balanced output, 
and we postprocess the output to obtain a doubly stochastic solution.

\subsubsection*{Preprocessing}

To apply Theorem~\ref{t:total-movement} on the Paulsen problem,
we need a bound on $\Delta^{(0)}$ based on the error $\eps$ as defined in Definition~\ref{d:epsDS}.
We first bound $\Delta^{(0)}$ in the following lemma with the additional assumption that $s^{(0)}=m$,
and then we show how to do a simple preprocessing to ensure that $s^{(0)}=m$ in the next lemma.

\begin{lemma} \label{l:Delta-eps}
Suppose $\U^{(0)}$ is $\eps$-nearly doubly stochastic as defined in Definition~\ref{d:epsDS}, and $s^{(0)}=m$.
Then 
\[
\Delta^{(0)} \leq 2m^2\eps^2.
\]
\end{lemma}
\begin{proof}
The first condition of Definition~\ref{d:epsDS} implies that
\[\frac{1}{m} \tr[(mI_m - m\sum_{i=1}^k U_i U_i^T)^2]
\leq \frac{1}{m} \tr[(m\eps I_m)^2]
= m \tr[(\eps I_m)^2] = m^2 \eps^2.
\]
The second condition of Definition~\ref{d:epsDS} implies that
\[\frac{1}{n} \tr[(mI_n - n\sum_{i=1}^k U_i^T U_i)^2]
\leq \frac{1}{n} \tr[(m\eps I_n)^2] = \frac{m^2}{n} \tr[(\eps I_n)^2] = \frac{m^2}{n} n \eps^2 = m^2 \eps^2.
\]
Therefore, using $s^{(0)}=m$,
\[\Delta^{(0)} = \frac{1}{m} \tr[(mI_m - m\sum_{i=1}^k U_i U_i^T)^2] + 
\frac{1}{n} \tr[(mI_n - n\sum_{i=1}^k U_i^T U_i)^2]
\leq 2m^2\eps^2.\]
\end{proof}

We now describe the preprocessing step to satisfy $s^{(0)}=m$.

\begin{lemma} \label{l:preprocess}
Given $\U$ that is $\eps$-nearly doubly stochastic for $\eps \leq 1/2$ as in Definition~\ref{d:epsDS},
we can scale $\U$ to produce $\U^{(0)}$ such that
$\U^{(0)}$ is $O(\eps)$-nearly doubly stochastic and $s^{(0)}=m$
and $\dist(\U^{(0)},\U) = O(\eps^2 m)$.
\end{lemma}
\begin{proof}
Let $s:=s(\U)$ be the size of $\U$ as in Definition~\ref{d:size}.
We simply scale each entry in $\U$ by a factor of $\sqrt{m/s}$ to produce $\U^{(0)}$.
By construction,
\[s^{(0)} = \sum_{i=1}^k \norm{\sqrt{\frac{m}{s}} U_i}_F^2 
= \frac{m}{s} \sum_{i=1}^k \norm{U_i}_F^2 = m.
\]
Since $\U$ is $\eps$-nearly doubly stochastic as stated in Definition~\ref{d:epsDS},
we have $(1-\eps)m \leq s \leq (1+\eps)m$.
This implies that $1-O(\eps) \leq m/s \leq 1+O(\eps)$ when $\eps \leq 1/2$ and thus
\[\dist(\U,\U^{(0)}) =
\sum_{i=1}^k \norm{U_i - \sqrt{\frac{m}{s}} U_i}_F^2
= \sum_{i=1}^k O(\norm{\eps U_i}_F^2) = O(\eps^2 m).
\]
Finally, it is clear that $\U^{(0)}$ is still $O(\eps)$-nearly doubly stochastic as for instance
\[\sum_{i=1}^k U_i^{(0)} {U_i^{(0)}}^T 
= \frac{m}{s} \sum_{i=1}^k U_i U_i^T 
\preceq (1+O(\eps))I_m
\]
and similarly for the other condition in Definition~\ref{d:epsDS}.
\end{proof}

\subsubsection*{Applying the Dynamical System and Postprocessing}

We will apply Theorem~\ref{t:total-movement} to the preprocessed input $\U^{(0)}$ to obtain an doubly balanced output $\U^{(\infty)}$ with $\Delta(\U^{(\infty)})=0$.
We do a simple postprocessing on $\U^{(\infty)}$ to obtain a doubly stochastic solution $\V$, and then we bound the squared distance of the input $\U$ and the final output $\V$.

\begin{lemma} \label{l:postprocess}
Let $\U^{(0)}$ be $\eps$-nearly doubly stochastic with $s^{(0)}=m$.
Let $\U^{(\infty)}$ be the output of the dynamical system when given $\U^{(0)}$ as the input.
Then, we can move $\U^{(\infty)}$ to $\V$ such that $\V$ is doubly stochastic as in Definition~\ref{d:DS} and $\dist(\U^{(\infty)},\V) \leq O(mn\sqrt{\Delta^{(0)}})$.
\end{lemma}
\begin{proof}
Denote $s:=s^{(\infty)}$.
We consider two cases.
The first case is when $s=0$.
We will argue that in this case the bound $O(mn\sqrt{\Delta^{(0)}})$ is trivial that setting $V$ to be any doubly stochastic operator will do.
First, we argue that $s=0$ implies that $\capa(\U^{(0)})=0$.
By Lemma~\ref{l:capUB} and Theorem~\ref{t:cap_lower_bound}, we have for any $t \geq 0$,
\[s^{(t)} \geq \capa(\U^{(t)}) \geq s^{(t)} - mn\sqrt{\Delta^{(t)}/2}.\]
Since $\Delta^{(\infty)}=0$ and the capacity is unchanged over time by Lemma~\ref{l:cap-unchanged},
we have 
\begin{equation} \label{e:s-infty}
s^{(\infty)} = \capa(\U^{(\infty)}) = \capa(\U^{(0)}),
\end{equation}
and thus $\capa(\U^{(0)})=0$ when $s=0$.
Then it follows from Theorem~\ref{t:cap_lower_bound} that
\[0=\capa(\U^{(0)}) \geq s^{(0)} - mn\sqrt{\Delta^{(0)}/2} = m-mn\sqrt{\Delta^{(0)}/2} \quad \implies \quad n\sqrt{\Delta^{(0)}} \geq \sqrt{2}.
\]
Let $\V$ be any doubly stochastic operator.
We still have
\[
\dist(\U^{(\infty)},\V) = \sum_{i=1}^k \norm{U_i^{(\infty)} - V_i}_F^2
\leq 2\sum_{i=1}^k \Big(\norm{U_i^{(\infty)}}_F^2 + \norm{V_i}_F^2\Big)
= O(m) = O(mn\sqrt{\Delta^{(0)}}).
\]
To summarize, this is the trivial case that moving $\U^{(0)}$ to any doubly stochastic operator will do.
Henceforth, we can assume that $n\sqrt{\Delta^{(0)}} < \sqrt{2}$.

The second case is when $s > 0$.
In this case, we simply scale each entry of $\U^{(\infty)}$ by a factor of $\sqrt{m/s}$ to produce $\V$, that is, we set $V_i = \sqrt{m/s} \cdot U_i^{(\infty)}$ for $1 \leq i \leq k$.
We can check that $\V = \{V_1, \ldots, V_k\}$ is doubly stochastic as $\Delta(\U^{(\infty)})=0$ implies
\[
sI_m - m\sum_{i=1}^k U_i^{(\infty)} {U_i^{(\infty)}}^T = 0
\quad \implies \quad
I_m - \sum_{i=1}^k V_i V_i^T = 0,
\]
and similarly $\Delta(\U^{(\infty)})=0$ implies
\[
sI_n - n\sum_{i=1}^k {U_i^{(\infty)}}^T U_i^{(\infty)} = 0
\quad \implies \quad
I_n - \frac{n}{m} \sum_{i=1}^k V_i^T V_i = 0.
\]
The squared distance between $\U^{(\infty)}$ and $\V$ is
\[
\dist(\U^{(\infty)}, \V) 
= \sum_{i=1}^k \norm{U_i^{(\infty)} - V_i}_F^2
= (\sqrt{\frac{m}{s}}-1)^2 \sum_{i=1}^k \norm{U_i^{(\infty)}}_F^2
= (\sqrt{m}-\sqrt{s})^2.
\]
To bound $s$, we use Lemma~\ref{l:capUB}, Lemma~\ref{l:cap-unchanged} and then Theorem~\ref{t:cap_lower_bound} to get 
\[
s = s^{(\infty)} \geq \capa(\U^{(\infty)}) = \capa(\U^{(0)}) \geq s^{(0)}-mn\sqrt{\Delta^{(0)}/2} = m - mn\sqrt{\Delta^{(0)}/2},
\]
and we can continue the above calculation and get
\[
\dist(\U^{(\infty)}, \V) 
\leq \Big(\sqrt{m}-\sqrt{m-mn\sqrt{\Delta^{(0)}/2}}~\Big)^2
= m\Big(1-\sqrt{1-n\sqrt{\Delta^{(0)}/2}}~\Big)^2
\leq O(mn\sqrt{\Delta^{(0)}}),
\]
where the last inequality follows from the inequality $1-\sqrt{1-x}\leq x$ for $0 \leq x \leq 1$ and the assumption in the second case that $n\sqrt{\Delta^{(0)}/2} < 1$.
\end{proof}

We can finally prove a bound on the operator Paulsen problem.

\begin{theorem} \label{t:operator-Paulsen}
Given any operator $\U$ that is $\eps$-nearly doubly stochastic as in Definition~\ref{d:epsDS},
we can move $\U$ to $\V$ such that
\[
\V {\rm~is~doubly~stochastic}
\quad {\rm and} \quad
\dist(\V, \U) \leq O(m^2 n \eps).
\]
\end{theorem}
\begin{proof}
Given $\U$ that is $\eps$-nearly doubly stochastic,
we first apply the preprocessing step in Lemma~\ref{l:preprocess} 
to produce $\U^{(0)}$ such that $\U^{(0)}$ is $O(\eps)$-nearly doubly stochastic, $s^{(0)}=m$ and $\dist(\U^{(0)},\U) = O(m\eps^2)$.
By Lemma~\ref{l:Delta-eps}, since $s^{(0)}=m$ and $\U^{(0)}$ is $O(\eps)$-nearly doubly stochastic, we have $\Delta^{(0)} = O(m^2 \eps^2)$.
Then, by Theorem~\ref{t:total-movement}, we can apply our dynamical system on $\U^{(0)}$ to produce $\U^{(\infty)}$ such that $\Delta(\U^{(\infty)})=0$ and 
\[\dist(\U^{(0)}, \U^{(\infty)}) = O(mn\sqrt{\Delta^{(0)}}) = O(m^2n\eps).\]
Finally,
we apply Lemma~\ref{l:postprocess} to move $\U^{(\infty)}$ to $\V$ such that $\V$ is doubly stochastic and 
\[\dist(\U^{(\infty)},\V) \leq O(mn\sqrt{\Delta^{(0)}}) = O(m^2n\eps).\]
The distance between $\U$ and $\V$ is
\begin{eqnarray*}
\distance(\V,\U)
& \leq & \distance(\U^{(0)},\U) + \distance(\U^{(\infty)},\U^{(0)}) + \distance(\U^{(\infty)},\V)
\\
& \leq & O(\sqrt{m\eps^2}) + O(\sqrt{m^2 n\eps}) + O(\sqrt{m^2 n\eps}) = O(\sqrt{m^2 n\eps}),
\end{eqnarray*}
and this implies that $\dist(\V,\U) = O(m^2n\eps)$.
\end{proof}

\begin{remark}
As mentioned in Remark~\ref{r:matrix_size_depends_on_gcd}, we can obtain better bound when $g = \gcd(m, n) > 1$.
In such case, we have
\[
\capa(\U)
\ge s(\U) - \frac{mn}g \sqrt{\frac{\Delta(\U)}2}.
\]
Therefore both $\dist(\U^{(0)}, \U^{(\infty)})$ and $\dist(\U^{(\infty)}, \V)$, and hence $\dist(\V, \U)$, are at most $O(m^2 n \eps / g)$.
In particular, when $n$ is a multiple of $m$, our bound becomes $\dist(\V, \U) = O(m n \eps)$.
\end{remark}

\subsection{A Bound on the Paulsen Problem} \label{ss:Paulsen}

Recall that in the Paulsen problem, 
we are given a set of vectors $U=\{u_1, \ldots, u_n\}$ in $\R^d$ that forms an $\eps$-nearly equal norm Parseval frame as described in~(\ref{e:epsParseval}):
\[(1-\eps)I_d \preceq \sum_{i=1}^n u_i u_i^T \preceq (1+\eps)I_d
\quad {\rm and} \quad
(1-\eps) \frac{d}{n} \leq \norm{u_i}_2^2 \leq (1+\eps) \frac{d}{n},
\]
and the Paulsen problem is to find a function $f(d,n,\eps)$ such that there exists a set of vectors $V = \{v_1, \ldots, v_n\}$ in $\R^d$ that forms an equal norm Parseval frame as described in~(\ref{e:Parseval}):
\[
\sum_{i=1}^n v_i v_i^T = I_d \quad {\rm and} \quad
\norm{v_i}_2^2 = \frac{d}{n} {\rm~for~} 1 \leq i \leq n,
\]
and
\[
\dist(U,V) = \sum_{i=1}^k \norm{u_i - v_i}_2^2 \leq f(d,n,\eps).
\]

We will prove that $f(d,n,\eps) = O(d^2n\eps)$,
by following the same steps as in Subsection~\ref{ss:operator-Paulsen} and using a similar reduction to the operator setting as described in Subsection~\ref{ss:reduction}.

\begin{definition}[reduction from frame to operator] \label{d:reduction}
Given $U=\{u_1, \ldots, u_n\}$ where each $u_i \in \R^d$,
we define $\U=\{U_1, \ldots, U_n\}$ where each $U_i$ is a $d \times n$ matrix
with the $i$-th column of $U_i$ being $u_i$ and all other columns of $U_i$ being zero.
Note that $U$ is an $\eps$-nearly equal norm Parseval frame if and only if
$\U$ is $\eps$-nearly doubly stochastic.
\end{definition}

With this reduction, we can define the size of a frame, $\Delta$ of a frame and capacity of a frame using the definitions for operators.
We list them below for reference in the next section.

\begin{definition}[size, $\Delta$, capacity of a frame] \label{d:frame-parameters}
Given $U=\{u_1, \ldots, u_n\}$ where each $u_i \in R^d$,
we define
\[
s(U) = \sum_{i=1}^n \norm{u_i}_2^2 = \tr(\sum_{i=1}^n u_i u_i^T)
\]
as the size of the frame $U$,
\[
\capa(U) = \inf_{X \succeq 0} \frac{m \det(\sum_{i=1}^n X_{ii} u_i u_i^T)}{\det(X)^{1/n}}.
\]
as the capacity of the frame $U$, and
\[
\Delta(U) = \frac{1}{d} \tr\big((sI_d - d\sum_{i=1}^n u_i u_i^T)^2\big) + \frac{1}{n} \sum_{i=1}^n (s-n\norm{u_i}_2^2)^2.
\]
\end{definition}

We follow the proof of Theorem~\ref{t:operator-Paulsen} to obtain the same bound for the Paulsen problem.

\begin{theorem} \label{t:Paulsen}
Given any $U=\{u_1,\ldots,u_n\}$ where each $u_i \in \R^d$ that is an $\eps$-nearly equal norm Parseval frame, we can move $U$ to $V=\{v_1,\ldots,v_n\}$ where each $v_i \in \R^d$ such that
\[
V {\rm~is~an~equal~norm~Parseval~frame~} \quad {\rm and} \quad
\dist(U,V) \leq O(d^2n\eps).
\]
\end{theorem}
\begin{proof}
We apply the reduction in Definition~\ref{d:reduction} to obtain 
$\U = \{U_1, \ldots, U_n\}$ from $U = \{u_1,\ldots,u_n\}$,
where $\U$ is $\eps$-nearly doubly stochastic.
Then we follow the same steps as in Subsection~\ref{ss:operator-Paulsen}.

In the preprocessing step,
we scale $\U=\{U_1, \ldots, U_n\}$ as in Lemma~\ref{l:preprocess} to obtain $\U^{(0)} = \{U_1^{(0)}, \ldots, U_n^{(0)}\}$ such that $\U^{(0)}$ is $O(\eps)$-nearly doubly stochastic, $s(\U^{(0)})=d$, and $\dist(\U,\U^{(0)}) = O(d\eps^2)$.
By Lemma~\ref{l:Delta-eps}, we have $\Delta^{(0)} := \Delta(\U^{(0)}) \leq O(d^2 \eps^2)$. 
Note that this preprocessing step is simply entry-wise scaling.

In the main step from the dynamical system,
we apply Theorem~\ref{t:total-movement} on $\U^{(0)}$ to produce $\U^{(\infty)}$ such that $\Delta(\U^{(\infty)})=0$ and $\dist(\U^{(0)}, \U^{(\infty)}) = O(dn\sqrt{\Delta^{(0)}}) = O(d^2 n \eps)$.
By Lemma~\ref{l:unitary-scaling}, we know that $\U^{(\infty)}$ is an operator scaling of $\U^{(0)}$ as in Definition~\ref{d:scaling}. 

In the postprocessing step,
there are two cases to consider as in Lemma~\ref{l:postprocess}.
In the first case, $s(\U^{(\infty)})=0$.
As shown in Lemma~\ref{l:postprocess},
this implies that $n\sqrt{\Delta^{(0)}} \geq \sqrt{2}$,
and thus 
\[
\sqrt{2} \leq n\sqrt{\Delta^{(0)}} \leq O(nd\eps)
\quad \implies \quad
\eps \geq \Omega(1/nd).
\]
This is the trivial case as we can move $U$ to any equal norm Parseval frame $V$ and $\dist(U,V) = O(d) = O(d^2n\eps)$ as we have shown in Subsection~\ref{ss:alternating}.
In the second case when $s(\U^{(\infty)})>0$,
we can scale $\U^{(\infty)}$ to $\V$ such that $\V$ is doubly stochastic and $\dist(\U^{\infty},\V) \leq O(dn\sqrt{\Delta^{(0)}}) = O(d^2n\eps)$.
Note that in this case the scaling is simply entry-wise scaling.

Therefore, besides the trivial case, we obtain a scaling $\V$ of $\U$ such that $\dist(\V,\U) = O(d^2n\eps)$.
In Lemma~\ref{l:reduction-output}, we will show that given a doubly stochastic scaling $\V$ of $\U$, we can obtain an equal norm Parseval frame $V = \{v_1,\ldots,v_n\}$ such that
$\dist(U,V) \leq \dist(\U,\V) \leq O(d^2n\eps)$.
\end{proof}

The following lemma completes the proof of Theorem~\ref{t:Paulsen}, and thus Theorem~\ref{t:dynamical}.

\begin{lemma} \label{l:reduction-output}
Let $U=\{u_1,\ldots,u_n\}$ where $u_i \in \R^d$ for $1 \leq i \leq n$
and $\U=\{U_1,\ldots,U_n\}$ where $U_i \in \R^{d \times n}$ for $1 \leq i \leq n$ be as defined in Definition~\ref{d:reduction}.
Suppose $\V=\{V_1,\ldots,V_n\}$ where $V_i \in \R^{d \times n}$ is a scaling of $\U$, and $\V$ is doubly stochastic.
Then there exists an equal norm Parseval frame $V=\{v_1,\ldots,v_n\}$ where $v_i \in \R^d$ for $1 \leq i \leq n$ such that
$\dist(U,V) \leq \dist(\U,\V)$.
\end{lemma}
\begin{proof}
Since $\V$ is a scaling of $\U$, by Definition~\ref{d:scaling},
there exist $L \in \R^{d \times d}$ and $R \in \R^{n \times n}$ such that $V_i = L U_i R$ for $1 \leq i \leq n$. 
As $\V$ is doubly stochastic, we have
\[
\sum_{i=1}^n V_i V_i^T =
\sum_{i = 1}^n L U_i R R^T A_i^T L^T = I_d,
\quad {\rm and} \quad
\sum_{i=1}^n V_i^T V_i =
\sum_{i = 1}^n R^T U_i^T L^T L U_i R = \frac{d}{n} I_n.
\]
We claim that we can replace $R$ by a positive diagonal matrix $D$ so that $\tilde V_i = L U_i D$ is still doubly stochastic and $\dist(\U, \tilde \V) \le \dist(\U, \V)$.
Suppose this can be done, then we define 
$v_i$ to be the $i$-th column of $\tilde V_i$ and $V = \{v_1, \ldots, v_n\}$.
Notice that all other columns in $\tilde V_i$ are zeros,
as all but the $i$-th column of $U_i$ are zeros.
It is easy to check that 
\[\dist(U, V) = \dist(\U, \tilde \V) \leq \dist(\U,\V)
\quad {\rm and} \quad
\sum_{i=1}^n v_i v_i^T = I_d
\quad {\rm and} \quad
\sum_{i=1}^n v_i^T v_i = \frac{d}{n} I_n.
\] 
So $V$ is an equal norm Parseval frame with $\dist(U,V) \leq \dist(\U,\V)$ are we are done.

It remains to prove the claim.
Define $D = (R R^T)^{1/2}$ and $\tilde V_i = LU_iD$, so we have
\[
\sum_{i=1}^n \tilde V_i {\tilde V_i}^T
= \sum_{i = 1}^n L U_i D D^T U_i^T L^T
= \sum_{i = 1}^n L U_i R R^T U_i^T L^T
= I_d.
\]
It follows from $\sum_{i = 1}^n R^T U_i^T L^T L U_i R = (d/n) I_n$
that $(d/n) (R R^T)^{-1} = \sum_{i=1}^n U_i^T L^T L U_i$.
Note that $L U_i \in \mathbb R^{d \times n}$ has at most one non-zero column $L u_i$, 
and thus $\sum_{i=1}^n U_i^T L^T L U_i$ is a positive diagonal matrix,
which implies that $(RR^T)^{-1}$ and thus $D=(RR^T)^{1/2}$ are also positive diagonal matrices. 
Furthermore,
\[
\sum_{i=1}^n {\tilde V_i}^T \tilde V_i
= \sum_i D^T U_i^T L^T L U_i D
= \frac{d}{n} D^T (R R^T)^{-1} D
= \frac{d}{n} I_n.
\]
Therefore, we have that $D$ is a diagonal matrix and $\tilde \V = \{\tilde V_1, \ldots, \tilde V_n\}$ is doubly stochastic.

Finally, we check that $\dist(\U, \tilde \V) \le \dist(\U, \V)$.
Note that
\[
\dist(\U, \tilde \V)
= \sum_{i = 1}^n \| U_i - L U_i D \|_F^2
= \sum_{i = 1}^n \Big(\tr(U_i^T U_i) + \tr(L U_i D D^T U_i^T L^T) - 2 \tr(U_i^T L U_i D)\Big)
\]
and
\[
\dist(\U, \V)
= \sum_{i = 1}^n \| U_i - L U_i R \|_F^2
= \sum_{i = 1}^n \Big(\tr(U_i^T U_i) + \tr(L U_i R R^T U_i^T L^T) - 2 \tr(U_i^T L U_i R)\Big)
\]
Since $D D^T = R R^T$ by definition, it remains to prove that
\[
\tr\big(\sum_{i = 1}^n U_i^T L U_i D\big)
\ge \tr\big(\sum_{i = 1}^n U_i^T L U_i R\big).
\]
We consider the singular value decomposition of $R$.
By the definition that $D=(RR^T)^{1/2}$,
we see that the diagonal entries of $D$ are the singular values of $R$,
and so we can write the singular value decomposition of $R := XDY$,
where $X,Y \in \R^{n\times n}$ are orthonormal matrices.
Then $D^2 = R R^T = X D^2 X^T = (X D X^T)^2$.
By the uniqueness of PSD square root of PSD matrices (e.g., see Theorem~7.2.6 of \cite{HJ12}), we have $D = X D X^T$ and thus $X D = D X$.
Hence $R = X D Y = D X Y = D Z$ for some orthonormal matrix $Z$.
Therefore,
\[
\tr(\sum_{i = 1}^n U_i^T L U_i R)
= \tr(\sum_{i = 1}^n U_i^T L U_i D Z) \\
\le \sum_{j = 1}^n \sigma_j (\sum_{i = 1}^n U_i^T L U_i D Z) \\
= \sum_{j = 1}^n \sigma_j (\sum_{i = 1}^n U_i^T L U_i D),
\]
where $\sigma_j(A)$ are the singular values of $A$
and the inequality is by the fact that the trace is at most the sum of singular values (e.g., see Theorem~3.3.13 of \cite{HJ91}).
Since $\sum_{i = 1}^n U_i^T L U_i$ is diagonal, we have $\sum_{j = 1}^n \sigma_i(\sum_{i = 1}^n U_i^T L U_i D) = \tr(\sum_{i = 1}^n U_i^T L U_i D)$, and this completes the proof.
\end{proof}

We remark that there are alternative proofs of Theorem~\ref{t:Paulsen}.
One can use the techniques in Lemma~\ref{l:unitary-scaling} and work out the scaling matrices of the continuous operator scaling algorithm and see that $R$ is a diagonal matrix, and then we do not need Lemma~\ref{l:reduction-output} and the proof will be shorter.
We prefer to use Lemma~\ref{l:reduction-output} even though it is longer because it proves a stronger claim that any operator scaling (not just those from the discrete or continuous operator scaling algorithms) will give us a solution to the Paulsen problem.

\section{Improved Bound through Smoothed Analysis} \label{s:smoothed}

We note that the smoothed analysis only works in the Paulsen problem (not the operator Paulsen problem in Definition~\ref{d:operator-Paulsen}), so in this section we switch back to the frame setting of the Paulsen problem.
Recall from Definition~\ref{d:frame-parameters} for the corresponding definitions.

We first interpret what was done in Section~\ref{s:dynamical} as a reduction from the Paulsen problem to proving capacity lower bound.
Then we see why it could not be improved directly and we motivate the smoothed analysis as a way to go beyond the bound in Section~\ref{s:dynamical}.
Then we will give a detailed outline of the smoothed analysis and the organization of the rest of this section in Subsection~\ref{ss:overview}.

{\bf Reduction to Capacity Lower Bound:}
In Section~\ref{s:dynamical}, we have proved that using the dynamical system in Definition~\ref{d:dynamical}, if there is a capacity lower bound
\[\capa(U^{(t)}) \geq s(U^{(t)}) - p(d,n) \sqrt{\Delta(U^{(t)})} \quad {\rm~for~all~} t \geq 0
\]
where $p(d,n)$ is a function in $d$ and $n$,
then we can move an $\eps$-nearly equal norm Parseval frame $U$ 
to an equal norm Parseval frame $V$ with
\[\dist(U,V) 
\leq O\Big(p(d,n) \sqrt{\Delta(U^{(0)})}\Big)
\leq O\big(p(d,n) \cdot d\eps \big) 
\leq O(d^2 n \eps),
\]
where the second inequality is from Lemma~\ref{l:Delta-eps} and
the last inequality is from Theorem~\ref{t:cap_lower_bound} that shows $p(d,n) = O(dn)$.

In general, the matrix capacity lower bound $\capa(A) \geq s(A) - dn\sqrt{\Delta(A)/2}$ from Proposition~\ref{p:matrix-cap} is tight.
See Lemma~\ref{l:tight} for an example with $s(A)=d$, $\capa(A)=0$ and $\Delta(A)=2/n^2$.
This example implies that the analysis of Theorem~\ref{t:cap_lower_bound} is tight, since the proof is based on a reduction to matrix capacity lower bound.
The same example also shows that the analysis of the dynamical system in Theorem~\ref{t:total-movement} is tight, as we know that $s(U^{(\infty)}) = \capa(U^{(0)})$ from~(\ref{e:s-infty}) and thus the size of the frame has shrank much in our dynamical system and 
\[
\dist(U^{(0)},U^{(\infty)}) 
\geq \Omega(d)
\geq \Omega\big( dn \sqrt{\Delta(U^{(0)})} \big).
\]

{\bf Smoothed Analysis:}
Our intuition is that the instances with $\capa(U) \approx s(U) - dn\sqrt{\Delta(U)}$ are rare.
So our idea is to perturb the input instance and to prove a stronger capacity lower bound on the perturbed instance.
Using some probabilistic arguments, we will prove that with high probability a perturbation $W$ of $U$ satisfies $\Delta(W) \approx \Delta(U)$ and $\capa(W) \gg \capa(U)$.
To prove the lower bound of the capacity in the perturbed instance $W$,
we have developed an interesting method to prove matrix capacity lower bound using the results in our dynamical system.

\subsection{Overview and Organization} \label{ss:overview}

We will present an informal proof outline of Theorem~\ref{t:perturb} in this subsection.  
All the definitions will be formally defined in later subsections,
and all the statements will be formally stated and proved in later subsections.

Given an instance $U = \{u_1, \ldots, u_n\}$ of the Paulsen problem where $u_i \in \R^d$ for $1 \leq i \leq n$,
we would like to perturb $U$ to an instance $W = \{w_1, \ldots, w_n\}$ where $w_i \in \R^d$ for $1 \leq i \leq n$ so that
$U$ and $W$ are close, $\Delta(U) \approx \Delta(W)$ and there is a stronger capacity lower bound for $W$.

{\bf Perturbation:} The perturbation is informally described as follows.
For each $1 \leq i \leq n$, let $g_i$ be a $d$-dimensional vector where each entry is an independent Gaussian random variable $N(0,\sigma^2)$ with mean zero and variance $\sigma^2$.
We let
\begin{equation} \label{e:perturbation}
y_i = P_L(g_i) \quad {\rm and} \quad w_i = u_i + y_i \quad {\rm for~} 1 \leq i \leq n,
\end{equation}
where $L$ is a subspace of codimension $d^2 + n$ and $P_L$ is the orthogonal projection to the subspace $L$.
For technical reasons, we will normalize the vectors so that they have equal norm.
We will choose
\[
\sigma^2 \approx \frac{\sqrt{d\Delta(U)}}{n}.
\]
We will explain the choice of $\sigma^2$ later in this subsection.

{\bf Analysis of Perturbed Instances:}
We will prove that the perturbed instance $W=\{w_1,\ldots,w_n\}$ has the following properties with high probability.
\begin{enumerate}[(i)]
\item 
The squared distance between $U$ and $W$ is small in Proposition~\ref{p:movement}: 
\[\dist(U,W) 
\leq O(dn \sigma^2) 
\leq O(d^{3/2} \sqrt{\Delta(U)}),
\]
where the second inequality is by our choice of $\sigma^2$.
\item
Assuming $n$ is large enough and $\Delta$ is small enough, we bound the increase of $\Delta$ in Proposition~\ref{p:Delta} that
\[
\Delta(W) \leq O(\Delta(U)).
\]
This is the place where we need to use the subspace $L$ in the perturbation process to ensure that $\Delta(W)$ is bounded, and is also the bottleneck of the current proof that requires the assumptions on $n$ and $\Delta$.
\item
Assuming $n$ is large enough, we establish an improved capacity lower bound
in Theorem~\ref{t:capacity} that
\[
\capa(W) 
\geq s(W) - O(\sqrt{\frac{\Delta(W)}{d}}).
\]
This is the heart of the smoothed analysis, where we have removed the dependency on $n$ in the capacity lower bound.
\end{enumerate}

{\bf Paulsen Problem:}
From point (iii) and the reduction of the Paulsen problem to capacity lower bound discussed earlier, we expect that we can set $p(d,n)=1/\sqrt{d}$ and $\Delta(W)=O(\Delta(U))$ to bound the squared distance after the perturbation to be $O(\sqrt{\Delta(U)/d}) = O(\sqrt{d} \eps)$, independent of $n$ when $n$ is large enough.
This is eventually what we will prove.

One subtlety is that in Section~\ref{s:dynamical} we assume the capacity lower bound $\capa(U^{(t)}) \geq s(U^{(t)}) - O(dn\sqrt{\Delta^{(t)}})$ holds for all $t \geq 0$,
but the improved capacity lower bound in point (iii) only holds right after the perturbation and may not hold after we apply the dynamical system on the perturbed instance.

To fix this, we will do the perturbation step (infinitely) many times, following the framework in Proposition~\ref{p:half} and Theorem~\ref{t:total-movement}.
Let $U^{(0)} := U$ be the input to the Paulsen problem and $\Delta:=\Delta(U)$.
We perturb $U^{(0)}$ using the perturbation defined in~(\ref{e:perturbation}).
Then we apply the dynamical system on the perturbed instance until time $T_1$ where $\Delta(U^{(T_1)}) = \Delta/2$.
Then we perturb $U^{(T_1)}$ again, and apply the dynamical system on the perturbed instance until time $T_2$ where $\Delta(U^{(T_2)}) = \Delta/4$, and so on.
Using point (i) and modifying Proposition~\ref{p:half} and Theorem~\ref{t:total-movement}, we can prove that the movement in each step is geometrically decreasing, and it is sufficient to have the capacity lower bound only at time $T_j$ for all $j$ when $\Delta(U^{(T_j)}) = \Delta / 2^j$.

The above will give us the total movement after the perturbation is $O(\sqrt{d} \eps)$.
Note that the total movement in point~(i) is $O(d^{5/2} \eps)$, and this is the bound that we get for the Paulsen problem when $n$ is large enough and $\Delta$ is small enough.
So, in the current proof, the movement in the perturbation step is the bottleneck of the total movement.
The precise step-by-step procedure to move from the initial frame to an equal norm Parseval frame is described in Procedure~\ref{proc:path} of Subsection~\ref{ss:together}. 

{\bf Projection:}
A natural attempt for the perturbation is to add independent noise to each coordinate for each vector.
Unfortunately, it does not work as $\Delta(W)$ would become much bigger than $\Delta(U)$ with high probability.
The linear subspace $L$ consists of $d^2 + n$ linear constraints which are added to enforce that the ``cross terms/first order terms'' become zero to ensure that point (ii) holds.
This comes with the price of the additional assumption that $n \gg d^2$ for point (iii) to hold, basically because the linear subspace $L$ has codimension $d^2 + n$.

\subsubsection*{Matrix Capacity Lower Bound from Dynamical System}

Most of the work in Theorem~\ref{t:perturb} is to prove point (iii).
There are two main ingredients.

{\bf Pseudorandom Property:}
The first ingredient is to identify a pseudorandom property for a frame to have a stronger capacity lower bound.
Instead of doing it directly, we follow the reduction in Proposition~\ref{p:reduction-capacity} to consider the corresponding matrix $A$ defined in~(\ref{e:A}).
The pseudorandom property that we will use of a $d \times n$ matrix $A$ is that every column has at least one entry with value at least $\Omega(\sigma^2)$ and every row has almost all entries with value at least $\Omega(\sigma^2)$.
We will prove in Subsection~\ref{ss:pseudorandom} that after we do the perturbation on the vectors as described in~(\ref{e:perturbation}), the corresponding matrix $A$ defined in~(\ref{e:A}) of Proposition~\ref{p:reduction-capacity} has the pseudorandom property with high probability.
The proof of this lemma is quite technical,
and this is the step that we could not prove in the operator setting, 
and also we need the assumption that $n \gg d^2$ for the proof to go through.

{\bf Bounding Matrix Capacity using Dynamical System:}
The second ingredient is a new method to prove matrix capacity lower bound.
In Subsection~\ref{ss:outline}, we have seen that the capacity lower bound provides an indirect way to argue that $\Delta^{(t)}$ will converge to zero.
We prove the reverse direction to establish matrix capacity lower bound, that a fast convergence of $\Delta^{(t)}$ to zero implies a good capacity lower bound.

To do this, we define a matrix version of the Paulsen problem,
and also a dynamical system from matrix scaling to solve this problem.
We will show that the dynamical system will satisfy the same formulas as outlined in Subsection~\ref{ss:outline} and proved in Subsection~\ref{ss:formulas}.
Assuming the pesudorandom property of a matrix holds in the beginning,
we will show in Proposition~\ref{p:maintain} that it will hold during the execution of the dynamical system.
Proposition~\ref{p:maintain} requires a lower bound on $\sigma^2$ for the proof to go through, and this is the reason for our choice of $\sigma^2$, which is the bottleneck of the current proof as it requires a large movement in the perturbation process.

A key step is a combinatorial argument in Subsection~\ref{ss:combinatorial} that proves that there exists an absolute constant $\kappa$ such that 
\[
-\d \Delta^{(t)} \gtrsim \kappa \sigma^2 n \Delta^{(t)} {\rm~for~all~} t \geq 0
\quad \implies \quad
\Delta^{(t)} \lesssim \exp(\kappa \sigma^2 n t) \cdot \Delta^{(0)} {\rm~for~all~} t \geq 0,
\]
assuming the pesudorandom property of the matrix holds throughout the execution of the dynamical system.
This can be used to lower bound the matrix capacity using the following relations: 
\[
s^{(0)} - \capa(A) 
= s^{(0)} - s^{(\infty)}
= -\int_0^{\infty} \d s~dt
= \int_0^{\infty} 2\Delta^{(t)} dt
\lesssim \Delta^{(0)} \int_0^{\infty} 2\exp(-\kappa \sigma^2 n t) dt
= \frac{\Delta^{(0)}}{\kappa \sigma^2 n},
\]
where the first equality is by Proposition~\ref{p:matrix-cap} and
the third equality is by an identity analogous to that in Lemma~\ref{l:s'}.
This implies that
\[
\capa(A) \gtrsim s(A) - \frac{\Delta(A)}{\kappa \sigma^2 n} 
\quad \implies \quad
\capa(W) \gtrsim s(W) - \frac{\Delta(W)}{\kappa \sigma^2 n},
\] 
where the implication follows from the reduction in Proposition~\ref{p:reduction-capacity}.
This completes the outline of the proof of point (iii).

\subsubsection*{Organization}

The proof of the smoothed analysis can be roughly divided into two parts.
The first part is to prove a stronger matrix capacity lower bound assuming the pseudorandom property.
The second part is to study the perturbation process and proves that the perturbed instances satisfies the properties described above.
The two parts are mostly independent and so the reader can choose to read which part first.

For the first part, the proof is divided into three subsections.
In Subsection~\ref{ss:matrix-Paulsen},
we first define the matrix version of the Paulsen problem. 
Then, we define a dynamical system based on a continuous version of the matrix scaling algorithm to solve the problem.
Then, we prove the formulas of the dynamical system that are analogous to those in Subsection~\ref{ss:formulas}.
In Subsection~\ref{ss:Delta'}, we define the pseudorandom property of a matrix, and show that the pesudorandom property is maintained throughout the execution of the dynamical system.
And we prove that a lower bound on $-\d \Delta^{(t)}$ implies a lower bound on the matrix capacity as outlined above.
Finally, in Subsection~\ref{ss:combinatorial}, we prove the combinatorial lemma that the pseudorandom property of a matrix implies a lower bound on $-\d \Delta^{(t)}$.

For the second part, the proof is also divided into three subsections.
In Subsection~\ref{ss:perturb}, we define the perturbation process for the vectors.
Then we bound the total movement in the perturbation process as in point~(i), 
and we list the probabilistic tools for the rest of the analysis.
In Subsection~\ref{ss:Delta},
we prove point~(ii) by bounding the increase of $\Delta$ in the perturbation process.
Finally, in Subsection~\ref{ss:pseudorandom}, we prove point~(iii) that after the perturbation on the vectors, the corresponding matrix $A$ in~(\ref{e:A}) of Proposition~\ref{p:reduction-capacity} has the pseudorandom property.
Both proofs of point~(ii) and point~(iii) are quite involved.

We complete the proof in Subsection~\ref{ss:together}, 
where we describe a step-by-step procedure to move from the initial frame to an equal norm Parseval frame,
and we put all the pieces together and prove Theorem~\ref{t:perturb}.

\subsection{The Matrix Paulsen Problem and the Dynamical System} \label{ss:matrix-Paulsen}

We define the matrix Paulsen problem as an analog to the Paulsen problem in the matrix setting.
We would like to mention that it is through this simpler problem that we find our solution to the Paulsen problem.

Recall from Subsection~\ref{ss:capLB} the definitions of doubly balanced and doubly stochastic matrices, the row sum and the column sum, the size of a matrix, the capacity of a matrix, and the $\Delta$ of a matrix.

\begin{definition}[$\eps$-nearly doubly stochastic matrix] \label{d:matrix-epsDS}
A matrix $A \in \R^{m \times n}$ is $\eps$-nearly doubly stochastic if 
\[
1-\eps \leq r_i(A) \leq 1+\eps {\rm~for~} 1 \leq i \leq m
\quad {\rm and} \quad
(1-\eps) \frac{m}{n} \leq c_j(A) \leq (1+\eps) \frac{m}{n} {\rm~for~} 1 \leq j \leq n.
\]
\end{definition}

\begin{definition}[Hadamard product] \label{d:Hadamard-product}
Given two matrices $A,B \in \R^{m \times n}$,
the Hadamard product $A \circ B$ is an $m \times n$ matrix with $(A \circ B)_{ij} = A_{ij} B_{ij}$.
We use the shorthand $\AA$ to denote $A \circ A$.
\end{definition}

\begin{definition}[the matrix Paulsen problem] \label{d:matrix-Paulsen}
The matrix Paulsen problem asks for the best function $g(m,n,\eps)$ such that given any matrix $A \in \R^{m \times n}$ with $\AA$ being $\eps$-nearly doubly stochastic, 
\[
\inf_B \dist(A,B) := \inf_B \norm{A-B}_F^2 \leq g(m,n,\eps),
\]
where the infimum $B$ is over the set of matrices where $\BB$ is doubly stochastic.
\end{definition}

In the standard version of the matrix scaling problem, we are given a matrix $A \in \R^{m \times n}$ and the goal is to find a scaling so that the resulting matrix $B$ is doubly stochastic.
The corresponding definition of the matrix Paulsen problem should be to find a $B$ that is close to $A$ so that $B$ is doubly stochastic.
Note that in our definition of the matrix Paulsen problem in Definition~\ref{d:matrix-Paulsen}, our goal is to find a matrix $B$ that is close to $A$ so that $\BB$ (but not $B$) is doubly stochastic.
This definition is to be consistent with the Paulsen problem, where we change the vectors in $U = \{u_1, \ldots, u_n\}$ to $V = \{v_1, \ldots, v_n\}$ with small squared distance so that $\norm{v_i}_2^2$ are equal and $\sum_{i=1}^n v_i v_i^T = I$, where both requirements are about the {\em squares} of the vectors in $V$.

As in our approach for the operator Paulsen problem, we would like to find a scaling solution to the problem.

\begin{definition}[matrix scaling] \label{d:matrix-scaling}
Given a matrix $A \in \R^{m \times n}$, we say a matrix $B$ is a scaling of $A$ if there exist diagonal matrices
\[
L \in \R^{m \times m} \quad {\rm and} \quad
R \in \R^{n \times n} \quad
{\rm such~that~} B=LAR.
\]
\end{definition}

Our dynamical system is a continuous version of the alternating algorithm for matrix scaling in Subsection~\ref{ss:matrix}.
Note the similarity with the dynamical system for operator scaling as defined in Definition~\ref{d:dynamical}.

\begin{definition}[dynamical system for matrix scaling] \label{d:matrix-dynamical}
The following differential equation defines how $A^{(t)}$ changes over time:
\[
\d A_{ij} := \big(s(\AA) - mr_i(\AA)\big) A_{ij} + \big(s(\AA) - nc_j(\AA)\big) A_{ij}.
\]
\end{definition}

To measure the progress of the dynamical system, we keep track of 
\[\Delta(\AA) = \frac{1}{m} \sum_{i=1}^m (s(\AA) - mr_i(\AA))^2 
+ \frac{1}{n} \sum_{j=1}^n (s(\AA) - nc_j(\AA))^2,
\]
as $\Delta(\AA)=0$ if and only if $\AA$ is doubly balanced.
And we measure the size of the matrix as
\[s(\AA) = \sum_{i=1}^m \sum_{j=1}^n A_{ij}^2.
\]

We are going to prove the analogous statements for matrix scaling as in those statements for operator scaling in Subsection~\ref{ss:formulas}.
We will prove
\begin{enumerate}
\item 
in Lemma~\ref{l:matrix-s'} that
\[\d s(\AAt) = -2 \Delta(\AAt),\]
which implies that the size of the matrix $\AA$ is decreasing over time;
\item 
in Lemma~\ref{l:matrix-Delta'} that
\[\d \Delta(\AAt) = -4\sum_{i=1}^m \sum_{j=1}^n \Big(2s\big( \AAt \big)-nr_i\big(\AAt \big)-mc_j\big(\AAt \big)\Big)^2 \cdot (A_{ij}^{(t)})^2,\]
which implies that the $\Delta(\AAt)$ is decreasing over time;
\item
in Lemma~\ref{l:matrix-cap-unchanged} that
$\capa(\AAt)$ is unchanged over time.
\end{enumerate}

\begin{remark}
It is possible to reduce the matrix Paulsen problem to the operator Paulsen problem, by using one matrix for each entry.
Through this reduction, we can obtain the formulas in the matrix case as corollaries of the formulas in the operator case in Subsection~\ref{ss:formulas}.
In the following, we simply present the direct proofs of the formulas in the matrix case, as it is more straightforward than to present the reduction.
\end{remark}

\begin{definition}[matrix shorthand] \label{d:matrix-shorthand}
We use the shorthands 
\[s := s(\AA), \quad \Delta := \Delta(\AA), \quad 
r_i := r_i(\AA), \quad {\rm and~} c_j := c_j(\AA),\]
and write
\[s = \sum_{i=1}^m \sum_{j=1}^n A_{ij}^2, \quad r_i = \sum_{j=1}^n A_{ij}^2, 
\quad c_j = \sum_{i=1}^m A_{ij}^2,
\]
\[\d A_{ij} = (2s-mr_i-nc_j)A_{ij} \quad {\rm and} \quad
\Delta = \frac{1}{m} \sum_{i=1}^m (s-mr_i)^2 + \frac{1}{n} \sum_{j=1}^n (s-nc_j)^2.
\]
All the quantities change over time $t$, but we will drop the superscript for ease of notation.
\end{definition}

\subsubsection*{Formula for the change of $s$}

\begin{lemma} \label{l:matrix-s'}
\[
\frac{d}{dt} s =-2\Delta.
\]
\end{lemma}
\begin{proof}
By the definition of the dynamical system in Definition~\ref{d:matrix-dynamical},
\[\d A_{ij} = (s-mr_i+s-nc_j)A_{ij}.\]
Therefore,
\begin{align*}
\d s & 
= 2\sum_{i=1}^m \sum_{j=1}^n A_{ij} \d A_{ij} 
\\
& = 2\sum_{i=1}^m \sum_{j=1}^n (s-mr_i) A_{ij}^2 + 2\sum_{i=1}^m \sum_{j=1}^n (s-mc_j) A_{ij}^2
\\
& = 2 \sum_{i=1}^m (s-mr_i)r_i + 2\sum_{j=1}^n (s-nc_j)c_j
\\
& = -\frac{2}{m} \sum_{i=1}^m (s-mr_i)(-mr_i) + -\frac{2}{n} \sum_{j=1}^n (s-nc_j)(-nc_j)
\\ 
& = -\frac{2}{m} \sum_{i=1}^m (s-mr_i)^2 + -\frac{2}{n} \sum_{j=1}^n (s-nc_j)^2
\\
& = -2\Delta,
\end{align*}
where the second last equality follows from $\sum_{i=1}^m (s-mr_i) = \sum_{j=1}^n (s-nc_j) = 0$.
\end{proof}

\subsubsection*{Formula for change of $\Delta$}

\begin{lemma} \label{l:matrix-Delta'}
\[
\d \Delta = -4\sum_{i=1}^m \sum_{j=1}^n (2s-mr_{i}-nc_{j})^{2}\cdot A_{ij}^{2}.
\]
\end{lemma}
\begin{proof}
Starting from the definition $\Delta = \frac{1}{m} \sum_{i=1}^m (s-mr_i)^2 + \frac{1}{n} \sum_{j=1}^n (s-nc_j)^2$, we have
\begin{align*}
\d \Delta & = 
\frac{2}{m} \sum_{i=1}^m (s-mr_i)(\d s - m \d r_i) + \frac{2}{n} \sum_{j=1}^n (s-nc_j)(\d s - n \d c_j)
\\
& = -2 \sum_{i=1}^m (s-mr_i) \d r_i - 2 \sum_{j=1}^n (s-nc_j) \d c_j
\\
& = -4\sum_{i=1}^m \sum_{j=1}^n (s-mr_i) A_{ij} \d A_{ij} - 4 \sum_{i=1}^m \sum_{j=1}^n (s-nc_j) A_{ij} \d A_{ij}
\\
& = -4\sum_{i=1}^m \sum_{j=1}^n (s-mr_i+s-nc_j) A_{ij} \d A_{ij}
\\
& = -4\sum_{i=1}^m \sum_{j=1}^n (s-mr_i+s-nc_j)^2 A_{ij}^2,
\end{align*}
where the second equality follows from $\sum_{i=1}^m (s-mr_i) = \sum_{j=1}^n (s-nc_j) = 0$,
and the last equality follows from the definition of the dynamical system in Definition~\ref{d:matrix-dynamical}.
\end{proof}

\subsubsection*{Capacity Unchanged}

The proof is similar but simpler than that in Subsection~\ref{ss:formulas} as we have an explicit formula for $A_{ij}^{(t)}$, which is useful in later proofs.

\begin{lemma} \label{l:A-formula}
At time $T \geq 0$, we have
\[
A_{ij}^{(T)}=\exp\Big(\int_{0}^{T} (2s^{(t)}-mr_{i}^{(t)}-nc_{j}^{(t)}) dt\Big) \cdot A_{ij}^{(0)},
\]
and
\[
A^{(T)} = \diag\left( \exp\Big(\int_{0}^{T} \big(s^{(t)}-mr_{i}^{(t)} \big) dt\Big)\right) \cdot A^{(0)} \cdot 
\diag\left( \exp\Big(\int_{0}^{T} \big(s^{(t)}-nc_{j}^{(t)} \big) dt \Big) \right).
\]
\end{lemma}
\begin{proof}
It is easy to check that this solution satisfies the differential equation in Definition~\ref{d:matrix-dynamical} as
\begin{eqnarray*}
\frac{d}{dT} A_{ij}^{(T)} 
& = & (2s^{(T)} - mr_i^{(T)} - nc_j^{(T)}) \cdot \exp\Big(\int_{0}^{T} (2s^{(t)}-mr_{i}^{(t)}-nc_{j}^{(t)}) dt\Big) \cdot A_{ij}^{(0)}\\
& = & (2s^{(T)} - mr_i^{(T)} - nc_j^{(T)}) A_{ij}^{(T)}.
\end{eqnarray*}
By standard theory (see Theorem~2.1 of~\cite{Bjo}),
there is a unique solution to the differential equation $\d z^{(t)} = c^{(t)} z^{(t)}$ at an initial value $z^{(0)}=A_{ij}^{(0)}$ when $c^{(t)}$ is Lipschitz continuous over $t$.
So the first part of the lemma follows from the uniqueness of the solution of our dynamical system.
Let 
\[x_i = \exp\Big(\int_{0}^{T} \big( s^{(t)}-mr_{i}^{(t)} \big) dt \Big)
\quad {\rm and} \quad 
y_j = \exp\Big(\int_{0}^{T} \big( s^{(t)}-nc_{j}^{(t)} \big) dt \Big).
\]
Let $X \in \R^{m \times m}$ be the diagonal matrix with $X_{ii}=x_i$ for $1 \leq i \leq m$ and $Y \in \R^{n \times n}$ be the diagonal matrix with $Y_{jj}=y_j$ for $1 \leq j \leq n$.
Then we see the second part of the lemma as 
\[A_{ij}^{(T)} = x_i A_{ij}^{(0)} y_j \quad {\rm and~thus} \quad
A^{(T)} = X A^{(0)} Y.
\]
\end{proof}

We see how the capacity changes after scaling.
The proof is basically the same as that in Proposition~2.7 of~\cite{operator}, with a slightly different definition of capacity.

\begin{lemma} \label{l:matrix-cap-change}
Let $A \in \R^{m \times n}$, $X \in \R^{m \times m}$ be a positive diagonal matrix and $Y \in \R^{n \times n}$ be a positive diagonal matrix.
Then 
\[
\capa(XAY) = (\prod_{i=1}^m X_{ii})^{1/m} (\prod_{j=1}^n Y_{jj})^{1/n} \capa(A).
\]
\end{lemma}
\begin{proof}
Let $x_i=X_{ii}$ and $y_j=Y_{jj}$.
By Definition~\ref{d:matrix-capacity},
\begin{eqnarray*}
\capa(XAY) 
& = & \inf_{z > 0} \frac{\prod_{i=1}^m \big((XAYz)_i\big)^{1/m}}{\prod_{j=1}^n z_j^{1/n}}
\\
& = & \Big(\prod_{i=1}^m x_i\Big)^{1/m} \inf_{z > 0} \frac{\prod_{i=1}^m \big((AYz)_i\big)^{1/m}}{\prod_{j=1}^n z_j^{1/n}}
\\
& = &
\Big(\prod_{i=1}^m x_i\Big)^{1/m} \inf_{z > 0} \frac{\prod_{i=1}^m \big((Az)_i\big)^{1/m}}{\prod_{j=1}^n (z_j/y_j)^{1/n}}
\\
& = &
\Big(\prod_{i=1}^m x_i\Big)^{1/m} \Big(\prod_{j=1}^n y_j\Big)^{1/n} \inf_{z > 0} \frac{\prod_{i=1}^m \big((Az)_i\big)^{1/m}}{\prod_{j=1}^n (z_j)^{1/n}}
\\
& = &
\Big(\prod_{i=1}^m x_i\Big)^{1/m} \Big(\prod_{j=1}^n y_j\Big)^{1/n} \capa(A).
\end{eqnarray*}
\end{proof}

We are ready to check that the capacity of $\capa(\AAt)$ is unchanged over time $t$.

\begin{lemma} \label{l:matrix-cap-unchanged}
For any $T \geq 0$, we have
\[ \capa\big((A^{(T)})^{\circ 2}\big) = \capa\big((A^{(0)})^{\circ 2}\big).\]
\end{lemma}
\begin{proof}
Let 
\[
X=\diag\left( \exp\Big(\int_{0}^{T} \big(s^{(t)}-mr_{i}^{(t)} \big) dt\Big)\right)
\quad {\rm and} \quad
Y=\diag\left( \exp\Big(\int_{0}^{T} \big(s^{(t)}-nc_{j}^{(t)} \big) dt \Big) \right).
\]
By Lemma~\ref{l:A-formula}, we have $A^{(T)} = XA^{(0)}Y$.
Since $X$ and $Y$ are diagonal matrices,
we have
\[(A^{(T)})^{\circ 2}
= A^{(T)} \circ A^{(T)} 
= (XA^{(0)}Y) \circ (XA^{(0)}Y) 
= X^{\circ 2} \cdot (A^{(0)})^{\circ 2} \cdot Y^{\circ 2}.\]
Therefore, by Lemma~\ref{l:matrix-cap-change},
we have
\[
\capa\big( (A^{(T)})^{\circ 2} \big) 
= \Big(\prod_{i=1}^m \big({X^{\circ 2}}\big)_{ii}\Big)^{1/m} \Big(\prod_{j=1}^n \big(Y^{\circ 2}\big)_{jj}\Big)^{1/n} \capa\big( (A^{(0)})^{\circ 2} \big).
\]
To prove that the capacity is unchanged, it remains to check that 
$\prod_{i=1}^m \big(X^{\circ 2}\big)_{ii} =  \prod_{j=1}^n \big(Y^{\circ 2}\big)_{jj} = 1$.
Note that
\[
\sqrt{\prod_{i=1}^m \big(X^{\circ 2}\big)_{ii}}
= \prod_{i=1}^m X_{ii}
= \prod_{i=1}^m \exp\Big(\int_{0}^{T} \big(s^{(t)}-mr_{i}^{(t)} \big) dt \Big)
= \exp\Big( \int_{0}^{T} \sum_{i=1}^m \big(s^{(t)}-mr_{i}^{(t)} \big) dt \Big) 
= 1,
\]
as $\sum_{i=1}^m (s-mr_i) = ms - m\sum_{i=1}^m r_i = ms - ms = 0$.
Similarly, we can check that $\prod_{j=1}^n \big(Y^{\circ 2}\big)_{jj} = 1$ and the lemma follows.
\end{proof}

\subsection{Matrix Capacity Lower Bound from Dynamical System} \label{ss:Delta'}

In this subsection, we present a new method to prove matrix capacity lower bound, and use it to prove that a pseudorandom matrix has a stronger capacity lower bound.
First, we show how to prove matrix capacity lower bound using a lower bound on the convergence rate of the dynamical system.
Then, we define the pseudorandom property of a matrix and state the main results about matrix capacity lower bound.
Finally, we show that the pesudorandom property is maintained throughout the execution of the dynamical system.
In the next subsection,
we will prove the combinatorial lemma that the pseudorandom property of a matrix implies a lower bound on the convergence rate of the dynamical system, 
and this will complete the first part of Section~\ref{s:smoothed}.

\subsubsection*{Matrix Capacity Lower Bound from Convergence Rate Lower Bound}

In Subsection~\ref{ss:outline}, the operator capacity lower bound is used to show that the size of the operator will not decrease much (see point (iv) in Subsection~\ref{ss:outline}), and this implies that $\Delta$ must decrease quick enough, as otherwise the size will decrease too much by the formula $\d s = -2\Delta$ from Lemma~\ref{l:Delta'}, contradicting the capacity lower bound.
To summarize, the capacity lower bound provides an indirect way to argue about the fast convergence of $\Delta$ to zero.

In this subsection, we establish the reverse direction to prove capacity lower bound.
The following proposition shows that if we have a lower bound on the convergence rate of $\Delta$,
then we can use it to prove a lower bound on the capacity.
Recall the shorthand notation defined in Definition~\ref{d:matrix-shorthand}.

\begin{proposition} \label{p:cap-from-Delta'}
For the dynamical system defined in Definition~\ref{d:matrix-dynamical} with $A^{(0)}$ as the input matrix,
if the following convergence lower bound holds throughout the execution of the dynamical system
\[
-\d \Delta^{(t)} \geq \mu \Delta^{(t)} {\rm~for~all~} t \geq 0,
\]
then we have the capacity lower bound
\[
\capa^{(0)} \geq s^{(0)} - \frac{2\Delta^{(0)}}{\mu}.
\]
\end{proposition}
\begin{proof}
The matrix capacity lower bound is obtained from the formulas in the dynamical system in Subsection~\ref{ss:matrix-Paulsen} as follows.
\begin{eqnarray*}
s^{(0)} - \capa^{(0)} 
& = & s^{(0)} - \capa^{(\infty)}
\\
& = & s^{(0)} - s^{(\infty)} 
\\
& = & -\int_0^{\infty} \d s~dt
\\
& = & 2 \int_0^{\infty} \Delta^{(t)}~dt
\\
& \leq & 2\Delta^{(0)} \int_0^{\infty} e^{-\mu t} dt
\\
& = & \frac{2\Delta^{(0)}}{\mu}.
\end{eqnarray*}
The first line is because the matrix capacity is unchanged by Lemma~\ref{l:matrix-cap-unchanged}.
To see the second line, we first claim that $\lim_{t \to \infty} \Delta^{(t)} = 0$.
Otherwise, as $\Delta^{(t)}$ is non-increasing over time by Lemma~\ref{l:matrix-Delta'}, there exists a constant $\eta>0$ such that $\Delta^{(t)} \geq \eta$ for all $t \geq 0$.
Since $\d s^{(t)} = -2\Delta^{(t)}$ by Lemma~\ref{l:matrix-s'}, 
this would imply that $s^{(\infty)}$ is unbounded below (in particular, negative), a contradiction to that $s$ is non-negative.
With $\Delta^{(\infty)}=0$ proved,
the second line then follows from $s(\AA) \geq \capa(\AA) \geq s(\AA) - mn\sqrt{\Delta(\AA)/2}$ applied to $A := A^{(\infty)}$, where the first inequality is by the capacity upper bound in Lemma~\ref{l:matrix-capUB} applied to $\AA$ and the second inequality is by the capacity lower bound in Proposition~\ref{p:matrix-capLB} applied to $\AA$.
The fourth line is by Lemma~\ref{l:matrix-s'}, and the fifth line is by $\Delta^{(t)} \leq e^{-\mu t} \Delta^{(0)}$ which follows from the assumption.
The third line and the last line are simple calculations,
and rearranging completes the proof.
\end{proof}

\subsubsection*{Pseudorandom Property}

We now define the pseudorandom property that we will use,
and state the combinatorial lemma that shows a convergence lower bound for pseudorandom matrices.

\begin{definition}[pesudorandom property of a matrix] \label{d:pseudorandom}
A non-negative matrix $A \in \R^{m \times n}$ is $(\alpha,\beta)$-pseudorandom, denoted by $A\gtrsim_{\beta} \alpha$, if it satisfies the following two properties:
\begin{itemize}
\item Every column has at least one entry with value at least $\alpha$.
\item Every row has at least $(1-\beta)n$ entries with value at least $\alpha$.
\end{itemize}
\end{definition}

The following is a key step which we will prove in the next subsection.
This allows us to prove that the pseudorandom property will be maintained throughout the execution of the dynamical system in Proposition~\ref{p:maintain}, so that we can apply Proposition~\ref{p:cap-from-Delta'} to prove a capacity lower bound for a pseudorandom matrix.

\begin{proposition} \label{p:combinatorial}
If $(A^{(t)})^{\circ 2}\gtrsim_{\beta} \alpha$ for a constant $\beta \leq 10^{-9}$ for all $t \geq 0$, then there exists an absolute constant $\kappa \geq 10^{-7}$ such that
\[
- \d \Delta^{(t)} \geq \kappa \alpha n \Delta^{(t)}  {\rm~for~all~} t \geq 0.
\]
A corollary is that 
\[\Delta^{(t)} \leq e^{-\kappa \alpha n t} \cdot \Delta^{(0)} {\rm~for~all~} t \geq 0.
\]
\end{proposition}

The proof of Proposition~\ref{p:combinatorial} is combinatorial and has an involved case analysis, which is deferred to the next subsection.
Assuming Proposition~\ref{p:combinatorial} and Proposition~\ref{p:maintain},
we prove the main result in this subsection.
In our application to the Paulsen problem in Subsection~\ref{ss:together},
we will argue that all the assumptions will be satisfied.

\begin{theorem} \label{t:pseudorandom}
Suppose $(A^{(0)})^{\circ 2}$ is $(\alpha,\beta)$-pseudorandom for a constant $\beta \leq 10^{-9}$ and all the column sums are the same, i.e.
\[(A^{(0)})^{\circ 2} \gtrsim_{\beta} \alpha 
\quad {\rm and} \quad
c_j^{(0)} = \frac{s^{(0)}}{n} {\rm~for~} 1 \leq j \leq n,
\]
where $\kappa$ is the absolute constant in Proposition~\ref{p:combinatorial}.
Also assume that 
\[
\alpha \geq \frac{80\sqrt{m\Delta^{(0)}}}{\kappa n}, \quad
s^{(0)} = m \quad {\rm and} \quad 
\Delta^{(0)} \leq 1/10.\]
Then
\[
\capa^{(0)} \geq s^{(0)} - \frac{\Delta^{(0)}}{5\kappa \alpha n}.
\] 
\end{theorem}
\begin{proof}
The assumptions are needed to apply Proposition~\ref{p:maintain} to prove that when $A^{(0)}$ is given as the input matrix to the dynamical system,
all the matrices $(A^{(t)})^{\circ 2}$ throughout the execution of the dynamical system is still $(\frac{1}{10} \alpha,\beta)$-pseudorandom.
Thus, by Proposition~\ref{p:combinatorial}, we have
\[
-\d \Delta^{(t)} \geq \frac{1}{10} \kappa \alpha n \Delta^{(t)} {\rm~for~all~} t \geq 0.
\]
We can then apply Proposition~\ref{p:cap-from-Delta'} with $\mu = \kappa \alpha n / 10$ to get
\[
\capa^{(0)} \geq s^{(0)} - \frac{\Delta^{(0)}}{5 \kappa \alpha n}.
\]
\end{proof}

In the remaining of this subsection, we will prove Proposition~\ref{p:maintain}, and then Proposition~\ref{p:combinatorial} in the next subsection.

\subsubsection*{Invariance of Pseudorandom Property}

The following proposition proves that the pseudorandom property is maintained throughout the execution of the dynamical system.
This is an important proposition for the matrix capacity lower bound,
and the assumption on $\alpha$ is the reason for our choice of $\sigma^2$ in the perturbation step.
We will also explain how the proof of the following proposition can be used to bound the running time of the matrix scaling algorithms in~\cite{Cohen}.

\begin{proposition} \label{p:maintain}
Suppose $(A^{(0)})^{\circ 2}$ is $(\alpha,\beta)$-pseudorandom and all the initial column sums are the same, i.e.
\[(A^{(0)})^{\circ 2} \gtrsim_{\beta} \alpha
\quad {\rm and} \quad
c_j^{(0)} = \frac{s^{(0)}}{n} {\rm~for~} 1 \leq j \leq n,
\]
where $\kappa$ is the constant in Proposition~\ref{p:combinatorial}.
Also assume that 
\[
\alpha \geq \frac{80\sqrt{m\Delta^{(0)}}}{\kappa n}, \quad
s^{(0)} = m \quad {\rm and} \quad \Delta^{(0)} \leq 1/10.\]
Then the matrix in the dynamical system is still $(\frac{1}{10} \alpha,\beta)$-pesudorandom at any time $t$, i.e.
\[(A^{(t)})^{\circ 2} \gtrsim_{\beta} \frac{\alpha}{10} {\rm~~for~all~} t \geq 0.
\]
\end{proposition}
\begin{proof}
Consider the set of indices $S = \{(i, j) ~|~ (A^{(0)}_{ij})^2 \geq \alpha\}$.
Let $T$ be the supremum such that $(A_{ij}^{(t)})^2 \geq \alpha/10$ for all $(i, j) \in S$ and $0 \leq t \leq T$.
Our goal is to prove that $T$ is unbounded, which will imply that $A^{(t)}$ is $(\alpha/10, \beta)$-pesudorandom for all $t \geq 0$.

From Lemma~\ref{l:A-formula}, we know the explicit formula for
\[
A_{ij}^{(T)}=\exp\Big(\int_{0}^{T} (s^{(t)}-mr_{i}^{(t)}) dt\Big) \cdot A_{ij}^{(0)} \cdot
\exp\Big(\int_{0}^{T} (s^{(t)}-nc_{j}^{(t)}) dt\Big).
\]
Let $x_i := \exp\Big(\int_{0}^{T} (s^{(t)}-mr_{i}^{(t)}) dt\Big)$ and 
$y_j := \exp\Big(\int_{0}^{T} (s^{(t)}-nc_{j}^{(t)}) dt\Big)$.
Then $A_{ij}^{(T)} = x_i A_{ij}^{(0)} y_j$.
We would like to bound the exponents of $x_i$ and $y_j$.
Using $(s-mr_i)^2 \leq m\Delta$ that follows from the definition of $\Delta$ in Definition~\ref{d:matrix-Delta}, we have
\begin{eqnarray*}
\left|\int_{0}^{T} (s^{(t)}-mr_{i}^{(t)}) dt\right|
& \leq &
\int_{0}^T |s^{(t)} - mr_i^{(t)}| dt
\\
& \leq &
\int_{0}^T \sqrt{m \Delta^{(t)}} dt
\\
& \leq &
\sqrt{m \Delta^{(0)}} \int_0^T e^{-\frac{\kappa \alpha n t}{20}} dt
\\
& \leq &
\frac{20\sqrt{m \Delta^{(0)}}}{\kappa \alpha n},
\end{eqnarray*}
where the third inequality is from the corollary in Proposition~\ref{p:combinatorial}
and our assumption that $(A^{(t)})^{\circ 2}$ is $(\alpha/10,\beta)$-pesudorandom for $0 \leq t \leq T$.
Since we assume that $\alpha \geq 80\sqrt{m \Delta^{(0)}}/(\kappa n)$, we have
\[
-\frac{1}{4} \leq \int_{0}^{T} (s^{(t)}-mr_{i}^{(t)}) dt \leq \frac{1}{4}
\quad \implies \quad
e^{-1/4} \leq x_i \leq e^{1/4}.
\]
So we have that $x_i$ is a constant for $1 \leq i \leq n$.
To bound $(A_{ij}^{(T)})^2$ in terms of $(A_{ij}^{(0)})^2$, it remains to lower bound $y_j$ by a constant.
Note that
\[(A_{ij}^{(T)})^2 = x_i^2 (A_{ij}^{(0)})^2 y_j^2 \leq e^{1/2} (A_{ij}^{(0)})^2 y_j^2 \leq 2 (A_{ij}^{(0)})^2 y_j^2.\]
To lower bound $y_j$, we consider the column sum by summing the above inequality 
\[
c_j^{(T)} = \sum_{i=1}^m (A_{ij}^{(T)})^2 \leq 2y_j^2\sum_{i=1}^m (A_{ij}^{(0)})^2 = 2y_j^2 c_j^{(0)} = \frac{2s^{(0)}y_j^2}{n}
\quad \implies \quad
y_j^2 \geq \frac{c_j^{(T)} n}{2s^{(0)}},
\]
where the last equality follows from our assumption about the initial column sum.
Therefore, to lower bound $y_j^2$, it suffices to show that the column sum at time $T$ is not much smaller than the initial column sum which is $s^{(0)}/n$.
To lower bound $c_j^{(T)}$,
we keep track of the change of $c_j$ over time.
By the definition of the dynamical system in Definition~\ref{d:matrix-dynamical},
\begin{eqnarray*}
& &\d c_j = \d \sum_{i=1}^m A_{ij}^2 = 2 \sum_{i=1}^m A_{ij} \d A_{ij}
= 2 \sum_{i=1}^m A_{ij}^2 (s - mr_i + s - nc_j).
\\
& \implies & \frac{1}{2} \d c_j = \sum_{i=1}^m (s-mr_i)A_{ij}^2 + \sum_{i=1}^m (s-nc_j)A_{ij}^2 
= \sum_{i=1}^m (s-mr_i)A_{ij}^2 + (s-nc_j)c_j.
\end{eqnarray*}
Since $(s-mr_i)^2 \leq m\Delta \leq m\Delta^{(0)}$ as $\Delta$ is non-increasing over time by Lemma~\ref{l:matrix-Delta'},
we have $s-mr_i \geq -\sqrt{m \Delta^{(0)}}$ and thus
\[
\frac{1}{2} \d c_j 
\geq -\sum_{i=1}^m \sqrt{m\Delta^{(0)}} A_{ij}^2 + (s-nc_j)c_j
= (s-nc_j - \sqrt{m\Delta^{(0)}}) c_j.
\]
This implies that $c_j$ is non-decreasing whenever $s-nc_j - \sqrt{m\Delta^{(0)}} \geq 0 \iff c_j \leq (s - \sqrt{m \Delta^{(0)}})/n$.
Note that $c_j^{(0)} \geq (s^{(0)} - \sqrt{m\Delta^{(0)}})/n$ by the definition of $\Delta^{(0)}$.
Since $s^{(T)}$ is non-increasing, the column sum can never go below this value and in particular 
\[c_j^{(T)} \geq \frac{s^{(T)}-\sqrt{m\Delta^{(0)}}}{n}.
\]
So, to lower bound the column sum, it suffices to lower bound the size of the matrix at time $T$.
We can do this directly by using the result in the dynamical system, very similar to the proof in Proposition~\ref{p:cap-from-Delta'}.
Note that
\[s^{(0)}-s^{(T)} 
= -\int_{0}^T \d s~dt
= 2\int_0^T \Delta~dt 
\leq 2 \Delta^{(0)} \int_0^T e^{-\frac{\kappa \alpha n t}{10}} dt 
\leq \frac{20\Delta^{(0)}}{\kappa \alpha n}
\leq \frac{\sqrt{\Delta^{(0)}}}{4\sqrt{m}},
\]
where the second equality is by Lemma~\ref{l:matrix-s'},
the first inequality is by our assumption that $A$ is $(\alpha/10,\beta)$-pesudorandom and the corollary in Proposition~\ref{p:combinatorial},
and the last inequality is by our assumption about $\alpha$.
Therefore, 
\[
c_j^{(T)}  
\geq \frac{s^{(0)} - \frac14 \sqrt{\Delta^{(0)}/m} - \sqrt{m \Delta^{(0)}}}{n}
\quad \implies \quad
y_j^2 \geq \frac{c_j^{(T)} n}{2 s^{(0)}} \geq
\frac{s^{(0)} - \frac14 \sqrt{\Delta^{(0)}/m} - \sqrt{m \Delta^{(0)}}}{2s^{(0)}}.
\]
Using our assumptions that $s^{(0)} = m$ and $\Delta^{(0)} \leq 1/10$,
we conclude that
\[
y_j^2 \geq \frac{1}{4} \quad \implies \quad
(A_{ij}^{(T)})^2 = x_i^2 (A_{ij}^{(0)})^2 y_j^2 
\geq \frac{e^{-1/2} (A_{ij}^{(0)})^2}{4} \geq \frac{1}{8} (A_{ij}^{(0)})^2
\geq \frac{\alpha}{8}.
\]
Since $A_{ij}^{(t)}$ is a continuous function of $t$ during the execution of the dynamical system, this implies that $(A_{ij}^{(T+\xi)})^2 > \alpha/10$ for a small enough $\xi > 0$, which contradicts that $T$ is the supremum that $(A_{ij}^{(t)})^2 \geq \alpha/10$ for all $0 \leq t \leq T$.
Therefore, $T$ is unbounded and the pseudorandom property is maintained throughout the execution of the dynamical system.
\end{proof}

\subsubsection*{Bounding Optimal Matrix Scaling Solution}

Before we move on to the combinatorial lemma in the next subsection, we discuss how the proofs in Proposition~\ref{p:maintain} could be used to bound the running time of the algorithms in~\cite{Cohen}.
This part can be skipped as this is not related to the Paulsen problem.

\begin{remark}[bounding optimal matrix scaling solution] \label{r:condition}
Given a non-negative matrix $A \in \R^{n \times n}$, the matrix scaling problem in~\cite{Cohen,ALOW} is to find diagonal matrices $L$ and $R$ such that $LAR$ is doubly stochastic (i.e. row and column sums are all one).
In~\cite{Cohen}, the time complexities of the algorithms depend on the parameter $\kappa$ (do not confuse with our $\kappa$ which is an absolute constant), which is defined as the ratio of the maximum entry to the minimum entry in $L$ and $R$.

Our matrix scaling problem is slightly different but closely related, in which we are given a matrix $A \in \R^{n \times n}$ and we would like to find diagonal matrices $L$ and $R$ such that $(LAR)^{\circ 2}$ is doubly stochastic.
We can define $\kappa$ similarly as in~\cite{Cohen}, and note that the $\kappa$ in our problem is just the square of the $\kappa$ in~\cite{Cohen}.

In the proof of Proposition~\ref{p:maintain}, we consider the quantities 
$x_i(T) := \exp\Big(\int_{0}^{T} (s^{(t)}-mr_{i}^{(t)}) dt\Big)$ and 
$y_j(T) := \exp\Big(\int_{0}^{T} (s^{(t)}-nc_{j}^{(t)}) dt\Big)$ for $1 \leq i,j \leq n$.
Let $X$ be the diagonal matrix with $X_{ii}=x_i(\infty)$ for $1 \leq i \leq n$, and $Y$ be the diagonal matrix with $Y_{jj}=y_j(\infty)$ for $1 \leq j \leq n$.
Then Lemma~\ref{l:A-formula} shows that $XAY$ is doubly balanced as $A^{(\infty)}$ is doubly balanced.
Note that the scaling solution for matrix scaling is unique under very mild assumption (i.e. the underlying graph is connected), and so we only need to bound the ratio of the maximum entry to the minimum entry of $X$ and similarly for $Y$.
Assuming the conditions in Proposition~\ref{p:maintain}, the proof of Proposition~\ref{p:maintain} shows that $e^{-1/4} \leq x_i(T) \leq e^{1/4}$ for $1 \leq i \leq n$ for all $T$, and the same argument can be used to show that $e^{-1/4} \leq y_j(T) \leq e^{1/4}$ for all $T$ for $1 \leq j \leq n$ when the input is a square matrix.
This implies that the $\kappa$ in~\cite{Cohen} is bounded by a constant when the input matrix $A$ satisfies the assumptions of Proposition~\ref{p:maintain}, and it follows that the algorithms in~\cite{Cohen} have near linear time for these instances.
The only known previous result about bounding $\kappa$ is in~\cite{KK} when $A$ is strictly positive, which is incomparable to our assumptions.

To extend the results in this section to bound $\kappa$, what we need to do is to find a property/condition of the input matrix under which we can lower bound $-\d \Delta^{(t)}$ as stated in Proposition~\ref{p:combinatorial}, and to prove that the property/condition will be maintained as in Proposition~\ref{p:maintain}.
\end{remark}

\subsection{Convergence Rate Lower Bound from Pseudorandom Property} \label{ss:combinatorial}

In this subsection, we prove Proposition~\ref{p:combinatorial} that establishes the convergence rate lower bound using the pseudorandom property.
We first begin with some notations for the proof,
and then we will prove a stronger lower bound using a stronger pseudorandom property, and finally we will prove Proposition~\ref{p:combinatorial}.

Recall from Lemma~\ref{l:matrix-Delta'} that
\[
\d \Delta = -4\sum_{i=1}^m \sum_{j=1}^n (s-mr_{i}+s-nc_{j})^{2}\cdot A_{ij}^{2},
\]
where from the shorthand in Definition~\ref{d:matrix-shorthand} that 
\[
\Delta = \frac{1}{m} \sum_{i=1}^m (s-mr_i)^2 + \frac{1}{n} \sum_{j=1}^n (s-nc_j)^2, \quad
s=\sum_{i=1}^m \sum_{j=1}^n A_{ij}^2, \quad r_i = \sum_{j=1}^n A_{ij}^2, \quad c_j = \sum_{i=1}^m A_{ij}^2.
\]
We will divide the rows and columns into buckets in the proof and the following are the notations.

\begin{definition}[buckets] \label{d:buckets}
Let $\AA$ be a non-negative $m\times n$ matrix.
We say a column $j$ is positive if $s-nc_j$ is positive, and a column $j$ is negative if $s-nc_j$ is negative.
We denote the set of positive columns by $C^+$, the set of negative columns by $C^-$, and the set of all columns by $C$.
Similarly, we say a row $i$ is positive if $s-mr_i$ is positive, and is negative if $s-mr_i$ is negative.
We denote the set positive rows by $R^+$, the set of negative rows by $R^-$ and the set of all rows by $R$.
We assume without loss of generality that $\sum_{j \in C^+} (s-nc_j)^2 \geq \sum_{j \in C^-} (s-nc_j)^2$.

We divide the positive columns into buckets, so that for each column $j$ in bucket $C_l$ for $l \in \mathbb{Z}$, the column $j$ satisfies $2^l \leq s-nc_j \leq 2^{l+1}$.
Similarly, we divide the negative rows into buckets, so that for each row $i$ in bucket $R_l$ for $l \in \mathbb{Z}$, the row $i$ satisfies $-2^l \geq s-mr_i \geq -2^{l+1}$.
Note that a bucket could be an empty set.
We use the notations
\[
C_{l \pm 1} := C_{l-1} \cup C_{l} \cup C_{l+1} \quad {\rm and} \quad
C_{l \pm 2} := C_{l-2} \cup C_{l-1} \cup C_{l} \cup C_{l+1} \cup C_{l+2}
\]
as the unions of the adjacent buckets of bucket $l$, and also
\[
\ol{C_{l \pm 1}} = C \setminus C_{l \pm 1} \quad {\rm and} \quad 
\ol{C_{l \pm 2}} = C \setminus C_{l \pm 2}
\]
as the complements of the unions.
Likewise, we use the analogous notations for the buckets of rows such as $R_{l \pm 1}, \ol{R_{l \pm 1}}, R_{l \pm 2}, \ol{R_{l \pm 2}}$.
The main reason for these definitions is the following:
For a column $j$ in bucket $C_l$ and a row $i$ in bucket $\ol{R_{l\pm 1}}$, it follows from the definitions that
\begin{equation} \label{e:difference}
| s - mr_i + s - nc_j | \geq \frac{1}{2} | s - nc_j |
\quad \implies \quad
(s - mr_i + s - nc_j)^2 \geq \frac{1}{4} (s - nc_j)^2.
\end{equation}
Likewise, we have the same inequality for a column $j$ in $C_{l \pm 1}$ and a row $i$ in $\ol{R_{l \pm 2}}$.
This will help us to bound the summands in $\d \Delta$.
\end{definition}

Recall from Definition~\ref{d:pseudorandom} that $\AA$ is $(\alpha,\beta)$-pseudorandom if every column of $\AA$ has at least one entry with value at least $\alpha$ and every row of $\AA$ has at least $(1-\beta)n$ entries with value at least $\alpha$.

We will first prove the following stronger conclusion assuming a stronger pseudorandom property, where we also require each column to have most entries with value at least $\alpha$. 
The statement will not be used in other places, 
but the proof is useful for the proof of Proposition~\ref{p:combinatorial},
and we think that the statement may be useful to improve our results.

\begin{proposition} \label{p:pseudorandom-strong}
Let $\AA$ be an $m \times n$ matrix.
Suppose $\AA$ satisfies the following properties:
\begin{enumerate}
\item Every row has at most $n/8000$ entries smaller than $\alpha$. 
\item Every column has at most $m/8000$ entries smaller than $\alpha$. 
\end{enumerate}
Then 
\[
- \d \Delta \geq \frac{\alpha m n \Delta}{32000}.
\]
\end{proposition}

Finally we will prove the following precise result that will imply Proposition~\ref{p:combinatorial}.

\begin{proposition} \label{p:pseudorandom-weak}
Let $\AA$ be an $m \times n$ matrix which is $(\alpha,\beta)$-pesudorandom for $\beta \leq 10^{-9}$.
Then
\[
- \d \Delta \geq \frac{\alpha n \Delta}{8192000}.
\]
\end{proposition}

The following definition will be useful in the proof of Proposition~\ref{p:pseudorandom-weak}.

\begin{definition}[strong/weak columns] \label{d:strong-weak}
Let $\AA$ be an $m \times n$ matrix which is $(\alpha,\beta)$-pseudorandom.
We say a column of $\AA$ is strong if it has at most $m / 128000$ entries with value less than $\alpha$; otherwise we say a column is weak.
Since $\AA$ is $(\alpha,\beta)$-pseudorandom, each row of $\AA$ has at most $\beta n$ entries with value less than $\alpha$, and so 
there are at most $\beta m n / (m /128000) = 128000 \beta n$ weak columns in $\AA$.
In Proposition~\ref{p:combinatorial}, we assume $\beta \leq 10^{-9}$, and so there will be at most $n/7500$ columns of $\AA$ which are weak.
\end{definition}

To lower bound $\d \Delta$, we can think of $\d \Delta$ as the inner product of two matrices $\inner{B}{\AA}$ where $B_{ij} = (s - mr_i + s - nc_j)^2$.
Our strategy is to find a large area of $B$ with large value, and then use the pseudorandom property of $\AA$ to conclude that the inner product is large.

\subsubsection*{Strong Lower Bound from Strong Pseudorandom Property}

We prove Proposition~\ref{p:pseudorandom-strong} in this subsubsection.
The following lemma will be useful in both proofs.
We will use the following lemma with $\gamma$ being an absolute constant at least $1/64$.

\begin{lemma} \label{l:subset-strong}
Suppose $C' \subseteq C^+$ is a set of columns with $\sum_{j \in C'} (s-nc_j)^2 \geq \gamma n \Delta$ for some $1 \geq \gamma > 0$,
and every column in $C'$ has at most $\gamma m / 2000$ entries of value smaller than $\alpha$.
Then
\[
-\d \Delta \geq \frac{\alpha \gamma^2 m n \Delta}{2000}.
\]
\end{lemma}
\begin{proof}
In this proof, we will restrict our attention only to the columns in $C'$.
We use $C'_l$ to denote $C_l \cap C'$.
As stated in~(\ref{e:difference}) of Definition~\ref{d:buckets},
for columns in $C_l$ and rows in $\ol{R_{l \pm 1}}$, we have
\[
\sum_{j \in C'} \sum_{i \in R} (s-mr_i + s-nc_j)^2 \cdot A_{ij}^2
\geq \sum_{l \in \mathbb{Z}} \sum_{j \in C'_l} \sum_{i \in \ol{R_{l \pm 1}}} (s-mr_i + s-nc_j)^2 \cdot A_{ij}^2
\geq \sum_{l \in \mathbb{Z}} \sum_{j \in C'_l} \sum_{i \in \ol{R_{l \pm 1}}} \frac{1}{4} (s-nc_j)^2 \cdot A_{ij}^2.
\]

{\bf Case 1:}
Suppose $|\ol{R_{l \pm 1}}| \geq \gamma m / 1000$ for all $l \in \mathbb{Z}$.
As each column $j \in C'_l$ has at most $\gamma m / 2000$ entries with value smaller than $\alpha$, there are at least $\gamma m / 2000$ entries which belong to $\ol{R_{l \pm 1}}$ with value at least $\alpha$.
Therefore,
\begin{align*}
- \d \Delta 
& \geq 4\sum_{i \in R} \sum_{j \in C'} (s - mr_i + s - nc_j)^2 \cdot A_{ij}^2
\geq \sum_{l \in \mathbb{Z}} \sum_{j \in C'_l} \sum_{i \in \ol{R_{l \pm 1}}} (s-nc_j)^2 \cdot A_{ij}^2\\
& \geq \sum_{l \in \mathbb{Z}} \sum_{j \in C'_l} \frac{\gamma m}{2000} \cdot (s-nc_j)^2 \cdot \alpha 
= \frac{\alpha \gamma m}{2000} \sum_{j \in C'} (s-nc_j)^2
\geq \frac{\alpha \gamma^2 m n \Delta}{2000},
\end{align*}
where the last inequality follows from the assumption of the lemma,
and we are done in this case.

{\bf Case 2:}
Otherwise, there exists an $l^*$ such that $|\ol{R_{l^*\pm 1}}| \leq \gamma m / 1000 \leq m / 1000$, which implies that $|R_{l^*-1} \cup R_{l^*} \cup R_{l^*+1}| \geq 999m / 1000$.
We consider $C'_{l^* \pm 2}$ and divide into two subcases.

\begin{enumerate}[(a)]
\item
The first subcase is when $\sum_{j \in C'_{l^* \pm 2}} (s-nc_j)^2 \leq \frac12 \sum_{j \in C'} (s-nc_j)^2$.
In this subcase, we consider the contribution to $\d \Delta$ from the columns in $\ol{C'_{l^* \pm 2}} := C' \setminus C'_{l^* \pm 2}$.
Our assumption in this subcase implies that 
\[
\sum_{j \in \ol{C'_{l^* \pm 2}}} (s-nc_j)^2
\geq \frac{1}{2} \sum_{j \in C'} (s-nc_j)^2 \geq \frac12 \gamma n \Delta.
\]
Furthermore, for every bucket $C'_l \in \ol{C'_{l^* \pm 2}}$,
we now have $R_{l^*-1} \cup R_{l^*} \cup R_{l^*+1} \subseteq \ol{R_{l \pm 1}}$ and thus $|\ol{R_{l \pm 1}}| \geq 999m/100 \geq 3m/4$ for all such $l$.
We can apply a similar argument as in case~1 to conclude that
\begin{align*}
- \d \Delta 
& \geq 4\sum_{i \in R} \sum_{j \in \ol{C'_{l^* \pm 2}}} (s - mr_i + s - nc_j)^2 \cdot A_{ij}^2
\geq \sum_{l \notin l^* \pm 2}~ \sum_{j \in C'_l} ~\sum_{i \in \ol{R_{l \pm 1}}} (s-nc_j)^2 \cdot A_{ij}^2\\
& \geq \sum_{l \notin l^* \pm 2}~ \sum_{j \in C'_l} \frac{m}{2} \cdot (s-nc_j)^2 \cdot \alpha 
= \frac{\alpha m}{2} \sum_{j \in \ol{C'_{l^* \pm 2}}} (s-nc_j)^2
\geq \frac{\alpha \gamma m n \Delta}{4},
\end{align*}
where the second inequality is from~(\ref{e:difference}) in Definition~\ref{d:buckets}, the third inequality is because $|\ol{R_{l \pm 1}}| \geq 3m/4$ and there are at most $\gamma m / 1000 \leq m / 1000$ entries in each column with value smaller than $\alpha$, and the final inequality is from our assumption in this subcase.
Therefore, the lemma also holds in this subcase.

\item
The remaining subcase is when $\sum_{j \in C'_{l^* \pm 2}} (s-nc_j)^2 \geq \frac12 \sum_{j \in C'} (s-nc_j)^2 \geq \frac12 \gamma n \Delta$.
We will rule this subcase out by deriving a contradiction, which will complete the proof of the lemma.
We will derive a contradiction by showing that $\sum_{i \in \ol{R_{l^* \pm 1}}} (s-mr_i)^2$ is too large.

Note that every row $i \in R_{l^* \pm 1}$ satisfies $|s-mr_i| \geq \frac{1}{16} |s-nc|$ where $|s-nc|$ is the maximum $|s-nc_j|$ term in $C'_{l^* \pm 2}$,
as $|s-nc_j| \leq 2^{l^*+3}$ for $j \in C_{l^* \pm 2}$ and $|s-mr_i| \geq 2^{l^*-1}$ for $i \in R_{l^* \pm 1}$.
Since we have $\sum_{i=1}^m (s - mr_i) = ms - m \sum_{i=1}^m r_i = 0$,
we know that
\[
\Big| \sum_{i \in \ol{R_{l^* \pm 1}}} (s-mr_i) \Big|
= \Big| \sum_{i \in R_{l^* \pm 1}} (s-mr_i) \Big|
= \sum_{i \in R_{l^* \pm 1}} \big| s-mr_i \big|
\geq \frac{1}{16} \big|R_{l^* \pm 1}\big| (s-nc).
\]
As $\big|\ol{R_{l^* \pm 1}}\big| \leq m/1000$, we use Cauchy-Schwarz to lower bound
\[
\sum_{i \in \ol{R_{l^* \pm 1}}} (s-mr_i)^2 
\geq \frac{1}{\big|\ol{R_{l^* \pm 1}}\big|} \Big( \sum_{i \in \ol{R_{l^* \pm 1}}} (s-mr_i) \Big)^2
\geq \frac{\big|R_{l^* \pm 1}\big|^2 (s-nc)^2}{256\big|\ol{R_{l^* \pm 1}}\big|.}
\]
Note that $\big| C'_{l^* \pm 2} \big| (s-nc)^2 
\geq \sum_{j \in C'_{l^* \pm 2}} (s-nc_j)^2 \geq \frac{1}{2} \gamma n \Delta$, 
which implies that $(s-nc)^2 \geq \frac{1}{2} \gamma \Delta$.
Hence,
\[
\sum_{i \in \ol{R_{l^* \pm 1}}} (s-mr_i)^2
\geq \frac{\gamma \Delta \big|R_{l^* \pm 1}\big|^2}{512 \big|\ol{R_{l^* \pm 1}}\big|}
\geq \frac{\gamma \Delta (999m/1000)^2}{512 (\gamma m/1000)} > m\Delta,
\]
contradicting to the definition of $\Delta$.
Therefore, case (2b) cannot happen.
\end{enumerate}
\end{proof}

Proposition~\ref{p:pseudorandom-strong} follows rather easily from Lemma~\ref{l:subset-strong}.

\begin{proof}[{\bf Proof of Proposition~\ref{p:pseudorandom-strong}}]
Recall that $\Delta = \frac{1}{m} \sum_{i=1}^m (s-mr_i)^2 + \frac{1}{n} \sum_{j=1}^n (s-nc_j)^2$, so either we have $\sum_{i=1}^m (s-mr_i)^2 \geq m\Delta/2$ or $\sum_{j=1}^n (s-nc_j)^2 \geq n\Delta/2$.

If $\sum_{j=1}^n (s-nc_j)^2 \geq n\Delta/2$, then we can assume without loss of generality that $\sum_{j \in C^+} (s-nc_j)^2 \geq n\Delta/4$.
Then, we apply Lemma~\ref{l:subset-strong} with $C' = C^+$ and $\gamma=1/4$.
Note that the assumption in Lemma~\ref{l:subset-strong} that every column in $C'$ has at most $\gamma m/2000 = m/8000$ is satisfied by the assumption in Proposition~\ref{p:pseudorandom-strong}.
So we can conclude from Lemma~\ref{l:subset-strong} that
\[-\d \Delta \geq \frac{\alpha mn\Delta}{32000}.
\]

The other case is the same. 
We have not used any property of the rows in Lemma~\ref{l:subset-strong}.
When $\sum_{i=1}^m (s-mr_i)^2 \geq m\Delta/2$, we can simply exchange the roles of $m$ and $n$ (or consider the transpose of the matrix) and apply Lemma~\ref{l:subset-strong} as in the previous paragraph to get the same conclusion.
\end{proof}

\subsubsection*{Lower Bound from Pseudorandom Property}

We prove Proposition~\ref{p:pseudorandom-weak} in this subsubsection, and then use it to prove Proposition~\ref{p:combinatorial}.
The case analysis is more involved in the proof of Proposition~\ref{p:pseudorandom-weak}, in some cases we can just apply Lemma~\ref{l:subset-strong} to get the (stronger) conclusion, but in the final case we can only use the weaker property to get the (weaker) conclusion.

\begin{proof}[{\bf Proof of Proposition~\ref{p:pseudorandom-weak}}]
As in the proof of Proposition~\ref{p:pseudorandom-strong},
we start with either $\sum_{i=1}^m (s-mr_i)^2 \geq m\Delta/2$ or $\sum_{j=1}^n (s-nc_j)^2 \geq n\Delta/2$.
In the former case, we can apply Lemma~\ref{l:subset-strong} as in Proposition~\ref{p:pseudorandom-strong} because all rows have at most $\beta n \ll n/8000$ entries with value smaller than $\alpha$.
Henceforth, we consider the latter case and assume without loss of generality that 
\[\sum_{j=1}^n (s-nc_j)^2 \geq n\Delta/2 \quad {\rm and} \quad 
\sum_{j \in C^+} (s-nc_j)^2 \geq n\Delta /4.\]
As in the proof of Lemma~\ref{l:subset-strong}, we will restrict our attention only to the columns in $C^+$, and consider the buckets in $C^+$.

\begin{definition}[strong/weak column buckets]
We call a column bucket $C_l$ strong if at least half of the columns in it are strong; otherwise we call the bucket $C_l$ weak.

Let $C_S$ be the union of the columns of the strong buckets, and let $C_W$ be the union of the columns of the weak buckets.
\end{definition}

Since $\sum_{j \in C^+} (s-nc_j)^2 \geq n\Delta /4$, either
\[\sum_{j \in C_S} (s-nc_j)^2 \geq n\Delta / 8 \quad {\rm or} \quad
 \sum_{j \in C_W} (s-nc_j)^2 \geq n\Delta/8.\]
We consider the two cases separately.

{\bf Case 1:}
Suppose we are in the former case when $\sum_{j \in C_S} (s-nc_j)^2 \geq n\Delta/8$.
For each strong bucket $C_l$, we throw away the weak columns and call the remaining columns $C'_l$.
Since each $(s-nc_j)^2$ in the same bucket is within a factor $4$ of each other and we throw away at most half the columns, we still have 
$\sum_{j \in C'_l} (s-nc_j)^2 \geq \frac18 \sum_{j \in C_l} (s-nc_j)^2$.
Call the union of the remaining columns in the strong buckets $C'_S$.
Then $\sum_{j \in C'_S} (s-nc_j)^2 \geq \frac{1}{8} \sum_{l \in C_S} (s-nc_j)^2 \geq n\Delta/64$.
Since each column in $C'_S$ is strong, we can apply Lemma~\ref{l:subset-strong} with $C' := C'_S$ and $\gamma=1/64$ to get 
\[
-\d \Delta \geq \frac{\alpha mn\Delta}{8192000}.
\]
This is the case where we require each strong column to have at most $m/128000$ 
entries of value smaller than $\alpha$, so that the assumption in Lemma~\ref{l:subset-strong} that every column in $C'$ has at most $\gamma m/2000$ entries of value smaller than $\alpha$ is satisfied.
In other cases, we apply Lemma~\ref{l:subset-strong} with larger $\gamma$.

{\bf Case 2:}
Suppose we are in the latter case when $\sum_{j \in C_W} (s-nc_j)^2 \geq n\Delta/8$.
We distinguish the weak column buckets into two types.

\begin{definition}[close to a big bucket]
We say a column bucket is big if it has at least $n/10$ columns.
Note that a big column bucket must be a strong column bucket, as there are at most $n/7500$ weak columns as stated in Definition~\ref{d:strong-weak}.

We say a weak bucket $C_l$ is close to a big bucket if there is a big bucket in $C_{l \pm 2}$; otherwise it is not close to a big bucket.

We denote the union of the columns of the weak buckets that are close to a big bucket by $C_{W^*}$, and the union of the columns of the weak buckets that are not close to a big bucket by $C_{\ol{W^*}}$.
\end{definition}

Since $\sum_{j \in C_W} (s-nc_j)^2 \geq n\Delta/8$ in this case,
either 
\[
\sum_{j \in C_{W^*}} (s-nc_j)^2 \geq \frac{n\Delta}{16} \quad {\rm or} \quad
\sum_{j \in C_{\ol{W^*}}} (s-nc_j)^2 \geq \frac{n\Delta}{16}.
\]
We consider these two subcases separately.

\begin{enumerate}[(a)]
\item The first subcase is when $\sum_{j \in C_{W^*}} (s-nc_j)^2 \geq n\Delta/16$.
We will reduce this subcase to case 1 (with different parameters).
For each weak bucket $C_l$ that is close a big bucket, let $\tilde{C}$ be a big bucket in $C_{l \pm 2}$.
We claim that 
\[
\sum_{j \in \tilde{C}} (s-nc_j)^2 \geq \frac{375}{8} \sum_{j \in C_l} (s-nc_j)^2.
\] 
The reasons are as follows.
Firstly, each summand on the LHS is at least a factor of $1/8$ of each summand on the RHS, as the worst case is when $\tilde{C} = C_{l-2}$. 
Secondly, the number of terms on the LHS is at least $375$ times the number of terms on the RHS, as $C_l$ has at most $2n/7500$ columns since it is a weak bucket and there are at most $n/7500$ weak columns as stated in Definition~\ref{d:strong-weak}, while $\tilde{C}$ has at least $n/10$ columns.
So we have the claim.

Let $C_B$ be the union of the columns of the big buckets that are close to a weak bucket contained in $C_{W^*}$.
Note that each big bucket contained in $C_B$ can be close to at most four weak buckets.
Therefore,
\[
\sum_{j \in C_B} (s-nc_j)^2 
\geq \frac{1}{4} \frac{375}{8} \sum_{j \in C_{W^*}} (s-nc_j)^2
\geq \frac{375}{32} \frac{n\Delta}{16} = \frac{375}{512} n\Delta.
\]
Since each big bucket is a strong bucket, we can apply the same argument as in case~1.
We throw away the weak columns in $C_B$ and call the set of the remaining columns $C'_B$.
Then 
\[\sum_{j \in C'_B} (s-nc_j)^2 
\geq \frac18 \sum_{j \in C_B} (s-nc_j)^2 
\geq \frac{375}{4096} n \Delta
\quad \implies \quad
-\d \Delta \geq \frac{\alpha m n \Delta}{800000},
\]
where the implication follows from Lemma~\ref{l:subset-strong} with $C' := C'_B$ and $\gamma = 375/4096 \geq 1/20$.

\item The second subcase is when $\sum_{j \in C_{\ol{W^*}}} (s-nc_j)^2 \geq n\Delta/16$.
This is the remaining case that we can only use the weak property.
For each weak bucket $C_l$ contained in $C_{\ol{W^*}}$,
we consider the corresponding row buckets in $R_{l \pm 1}$.
There are two situations.

The first situation is when $R_{l \pm 1}=\emptyset$.
Since each column in $C_l$ has at least one entry with value at least $\alpha$,
we have
\begin{equation} \label{e:one-entry}
\sum_{j \in C_l} \sum_{i=1}^m (s-mr_i + s-nc_j)^2 \cdot A_{ij}^2
\geq \sum_{j \in C_l} \sum_{i=1}^m \frac{1}{4} (s-nc_j)^2 \cdot A_{ij}^2
\geq \sum_{j \in C_l} \frac{\alpha}{4} (s-nc_j)^2,
\end{equation}
where the first inequality is from~(\ref{e:difference}) in Definition~\ref{d:buckets}.

The second situation is when $R_{l \pm 1} \neq \emptyset$.
In this situation, we instead consider the contribution of the entries in $R_{l \pm 1}$ to $\d \Delta$.
Again using~(\ref{e:difference}) in Definition~\ref{d:buckets}, we have
\[
\sum_{i \in R_{l \pm 1}} \sum_{j=1}^n (s-mr_i+s-nc_j)^2 \cdot A_{ij}^2
\geq \sum_{i \in R_{l \pm 1}} \sum_{j \in \ol{C_{l \pm 2}}} (s-mr_i+s-nc_j)^2 \cdot A_{ij}^2
\geq \sum_{i \in R_{l \pm 1}} \frac{1}{4} \big| \ol{C_{l \pm 2}} \big| (s-mr_i)^2 \cdot A_{ij}^2.
\]
As $C_l$ is not close to a big bucket, we have $\big| \ol{C_{l \pm 2}} \big| \geq n/2$.
As each row has at most $\beta n$ entries with value smaller than $\alpha$ and there are at most $n/7500$ weak columns,
there are at least $n/4$ entries in each row which belong to $\ol{C_{l \pm 2}}$ and have value at least $\alpha$ and moreover do not belong to the weak columns.
Therefore,
\[
\sum_{i \in R_{l \pm 1}} \sum_{j \in \ol{C_{l \pm 2}}} (s-mr_i+s-nc_j)^2 \cdot A_{ij}^2
\geq \sum_{i \in R_{l \pm 1}} \frac{1}{4} \big| \ol{C_{l \pm 2}} \big| (s-mr_i)^2 \cdot A_{ij}^2
\geq \sum_{i \in R_{l \pm 1}} \frac{\alpha n}{16} (s-mr_i)^2.
\]
In the worst case, there is only one row in $R_{l \pm 1}$,
but we still have
\begin{equation} \label{e:row}
\sum_{i \in R_{l \pm 1}} \sum_{j \in \ol{C_{l \pm 2}}} (s-mr_i+s-nc_j)^2 \cdot A_{ij}^2
\geq \sum_{i \in R_{l \pm 1}} \frac{\alpha n}{16} (s-mr_i)^2
\geq \frac{\alpha n}{256} (s-nc_j)^2,
\end{equation}
where we used $|s-mr_i| \geq \frac14 |s-nc_j|$ for $i \in R_{l \pm 1}$ and $j \in C_l$.

Now we combine the two situations together.
For a weak column bucket $C_l$ that is not close to a big bucket,
if $R_{l \pm 1} = \emptyset$,
we lower bound the contribution to $- \d \Delta$ from $C_l$ using~(\ref{e:one-entry}).
Otherwise, we lower bound the contribution to $- \d \Delta$ from $R_{l \pm 1}$ using~(\ref{e:row}). 
In either case, the contribution is at least
\[
\min\Big\{\sum_{j \in C_l} \frac{\alpha}{4} (s-nc_j)^2,
\frac{\alpha n}{256} (s-nc_j)^2 \Big\} 
\geq \sum_{j \in C_l} \frac{\alpha}{256} (s-nc_j)^2.
\]
Note that the contribution from any row in $R_{l \pm 1}$ is counted at most $5$ times, when we consider the buckets in $C_{l \pm 2}$.
And when we consider the rows, we consider those entries not in the weak columns, and so the contributions from the rows and from the columns are disjoint.
To summarize and to conclude, we have
\begin{align*}
- \d \Delta 
& = \sum_{i=1}^m \sum_{j=1}^n (s-mr_i + s-nc_j)^2 \cdot A_{ij}^2
\\
& \geq \frac{1}{5} \sum_{l \in \ol{W^*}} 
\min\Big\{ \sum_{j \in C_l} \sum_{i=1}^m (s-mr_i + s-nc_j)^2 \cdot A_{ij}^2,
\sum_{i \in R_{l \pm 1}} \sum_{j \in \ol{C_{l \pm 2}}} (s-mr_i+s-nc_j)^2 \cdot A_{ij}^2\Big\}
\\
& \geq \frac{1}{5} \sum_{l \in \ol{W^*}} 
\min\Big\{ \sum_{j \in C_l} \frac{\alpha}{4} (s-nc_j)^2,
\sum_{j \in C_l} \frac{\alpha}{256} (s-nc_j)^2\Big\}
\\
& \geq \frac{1}{5} \sum_{l \in \ol{W^*}} \sum_{j \in C_l} \frac{\alpha}{256} (s-nc_j)^2
\\
& = \frac{\alpha}{1280} \sum_{j \in C_{\ol{W^*}}} (s-nc_j)^2
\\
& \geq \frac{\alpha n \Delta}{20480},
\end{align*}
where the last inequality is by the assumption in this subcase.
\end{enumerate}
We have considered all cases, and the proposition follows by taking the minimum contribution to $-\d \Delta$ of these cases, which is achieved by Case~1.
\end{proof}

It is clear that Proposition~\ref{p:combinatorial} follows from Proposition~\ref{p:pseudorandom-weak} with $\kappa \geq 1/8192000$.
So, we have completed the first part of the smoothed analysis, that a pseudorandom matrix has a stronger capacity lower bound.
Next, we move on to the perturbation process and its analysis.

\subsection{The Perturbation Process} \label{ss:perturb}

In this subsection, we describe the perturbation process of the vectors.
Then we bound the movement during the perturbation process. 
We end this subsection by listing the facts and results that we will use for the rest of the analysis of the perturbation process.

\begin{procedure}[perturbation process] \label{proc:perturb}
We consider the following perturbation process:
\begin{enumerate}

\item (Preprocess) 
Rescale the vectors $u_{j} \in \R^d$ such that $\norm{u_{j}}^2=\frac{d}{n}$ for $1 \leq j \leq n$. 

\item (Gaussian noise) 
Let $x \in \R^{d \times n}$ be the concatenation of the vectors $x_1, \ldots, x_n \in \R^d$ such that $x_{i,j} := (x_j)_i$.
Let \[x \sim N(0,\sigma^2 I_{dn}),\]
where $x$ is sampled from a multivariate Gaussian distribution with zero mean and the covariance being the identity matrix.
Explicitly, $x_j \in \R^d$ is a vector with each coordinate being an independent Gaussian random variable in $N(0,\sigma^2)$ with mean zero and variance $\sigma^2$.

\item (First subspace)
Let $y \in \R^{d \times n}$ be the concatenation of the vectors $y_1, \ldots, y_n \in \R^d$ such that $y_{i,j} := (y_j)_i$.
Let 
\[L_{1}=\{y \in \R^{d\times n}\text{ such that } \inner{u_{j}}{y_{j}}=0\text{ for all } 1 \leq j \leq n\} 
\quad {\rm and} \quad 
y \sim N(0,\sigma^2 P_{L_1}),
\]
where $y$ is sampled from a multivariate Gaussian distribution with the covariance matrix being $P_{L_1} \in \R^{(d \times n) \times (d \times n)}$, which is the orthogonal projection matrix to the subspace $L_1$.
Note that $L_1$ is of co-dimension $n$ and $\rank(P_{L_1}) \ge (d-1)n$.
Explicitly, $y_j$ is the orthogonal projection of $x_j$ to the subspace perpendicular to $u_j$, such that
\[y_j = x_j - \frac{\inner{u_j}{x_j} u_j}{\norm{u_j}^2} = x_j - \frac{n}{d} \inner{u_j}{x_j} u_j.
\]

\item (Second subspace) 
Let $z \in \R^{d \times n}$ be the concatenation of the vectors $z_1, \ldots, z_n \in \R^d$ such that $z_{i,j} := (z_j)_i$.
Let 
\[L_{2}=\{z \in\R^{d\times n}\text{ such that }\sum_{j=1}^n u_{j}z_{j}^{T}=0\}
\quad {\rm and} \quad
z \sim N(0, \sigma^2 P_{L_1 \cap L_2}),
\]
where $z$ is sampled from a multivariate Gaussian distribution with the covariance matrix being $P_{L_1 \cap L_2}$, which is the orthogonal projection matrix to the subspace $L_1 \cap L_2$.
Note that $L_2$ is of co-dimension at most $d^2$ and thus $\rank(P_{L_1 \cap L_2}) \geq dn - n - d^2$.
Equivalently, we can think of $z = P_{L_1 \cap L_2} y$. 

\item (Noise adding)
The vector $z$ is the noise vector that we generate.
For $1 \leq j \leq n$, let the vectors $v_1, \ldots, v_n \in \R^d$ be
\[v_j := u_j + z_j {\rm~for~} 1 \leq j \leq n.\]

\item (Postprocess) 
We rescale each perturbed vector to have squared norm $d/n$ by letting
\[w_{j}=\sqrt{\frac{d}{n}} \frac{v_{j}}{\norm{v_{j}}}.
\]
The vectors $w_1, \ldots, w_n$ are the perturbed vectors that we generate.
\end{enumerate}
\end{procedure}

In the remainder of this section, we assume that the size of a frame is exactly $d$, and this is more convenient for the calculations as we do not need to keep track of $s$ in the definition of $\Delta$.

\begin{definition}[normalization for the Paulsen problem] \label{d:normal-Delta}
Given a frame $U = \{u_1, \ldots, u_n\}$ where $u_i \in \R^d$ for $1 \leq i \leq n$, when $s(U)=d$,
\[
\Delta(U)
=d\tr(I_d-\sum_{i=1}^{n}u_{i}u_{i}^{T})^{2} + n\sum_{i=1}^{n}(\frac{d}{n}-\norm{u_{i}}^{2})^{2}
=d\norm{I_d-\sum_{i=1}^{n}u_{i}u_{i}^{T}}_F^2 + n\sum_{i=1}^{n}(\frac{d}{n}-\norm{u_{i}}^{2})^{2}.
\]
Recall that $\Delta(U) \leq d^{2}\varepsilon^{2}$ by Lemma~\ref{l:Delta-eps}.
\end{definition}

\subsubsection*{Bounding the Movement in the Perturbation Process}

We first bound the movement and the increase in $\Delta$ in the preprocessing step of the perturbation process.

\begin{lemma} \label{l:perturb-preprocess}
Let $U = \{u_1, \ldots, u_n\}$ where $u_i \in \R^d$ for $1 \leq i \leq n$
and $s(U)=d$.
Let $V = \{v_1, \ldots, v_n\}$ where 
\[v_{i}=\sqrt{\frac{d}{n}}\frac{u_{i}}{\norm{u_{i}}} {\rm~for~} 1 \leq i \leq n 
\quad {\rm so~that} \quad
\norm{v_i}^2 = \frac{d}{n} {\rm~for~} 1 \leq i \leq n.
\] 
Assuming that $\Delta(U) \leq d/16$, we have
\[
\dist(U,V) = \sum_{i=1}^n \norm{v_{i}-u_{i}}^{2}\leq\frac{\Delta(U)}{d}
\quad {\rm and} \quad
\Delta(V) \leq 20\Delta(U).
\]
\end{lemma}
\begin{proof}
For the movement, 
\begin{eqnarray*}
\sum_{i=1}^n \norm{v_{i}-u_{i}}^{2}
& = & \sum_{i=1}^n \norm{(\sqrt{\frac{d}{n}} \frac{1}{\norm{u_i}} -1) u_i}^2
= \sum_{i=1}^n \left(\sqrt{\frac{d}{n}}-\norm{u_{i}}\right)^{2}
\\
& = & \sum_{i=1}^n \left(\frac{d}{n}-\norm{u_{i}}^2 \right)^{2} / \left(\sqrt{\frac{d}{n}} + \norm{u_i}\right)^2
\leq \frac{n}{d}\sum_{i=1}^n \left(\frac{d}{n}-\norm{u_{i}}^{2}\right)^{2}
\leq \frac{\Delta(U)}{d},
\end{eqnarray*}
where the last inequality follows from Definition~\ref{d:normal-Delta}.

For the bound on $\Delta(V)$, since $\norm{v_i}^2 = d/n$ and $s(V)=d$, it follows from Definition~\ref{d:normal-Delta} that
\begin{align*}
\Delta(V) & = d \tr(I_d-\sum_{i=1}^{n}v_{i}v_{i}^{T})^{2}
= d\norm{I_d - \sum_{i=1}^n v_i v_i^T}_F^2
\\
& \leq 2d\norm{I_d-\sum_{i=1}^{n} u_{i}u_{i}^{T}}_F^{2} 
+ 2d\norm{\sum_{i=1}^{n}u_{i}u_{i}^{T} -\sum_{i=1}^{n}v_{i}v_{i}^{T}}_F^{2},
\end{align*}
where the last inequality uses that $\norm{A+B}_F^2 \leq 2\norm{A}_F^2 + 2\norm{B}_F^2$ for two symmetric matrices $A$ and $B$.

Let $\mathcal{I}$ be the set of $i$ such that $\norm{u_{i}}^{2}\leq\frac{d}{2n}$.
Note that 
\[\Delta(U) \geq n \sum_{i \in \mathcal{I}} (\frac{d}{n} - \norm{u_i}^2)^2
\geq n\cdot(\frac{d}{2n})^{2}|\mathcal{I}|
\quad \implies \quad
|\mathcal{I}|\leq\frac{4n}{d^{2}}\Delta(U).
\]
Let $\gamma_{i}=1-\frac{d}{n}\norm{u_{i}}^{-2}$
if $i\notin\mathcal{I}$ and $\gamma_{i}=0$ for $i\in\mathcal{I}$.
By triangle inequality and the definition of $v_i$, 
\begin{align*}
\norm{\sum_{i=1}^{n}u_{i}u_{i}^{T}-\sum_{i=1}^{n}v_{i}v_{i}^{T}}_{F} & \leq\norm{\sum_{i\in\mathcal{I}}u_{i}u_{i}^{T}-\sum_{i\in\mathcal{I}}v_{i}v_{i}^{T}}_{F}+\norm{\sum_{i=1}^{n}\gamma_{i}u_{i}u_{i}^{T}}_{F}.
\end{align*}
For the first term, using that $\Delta(U)\leq\frac{d}{16}$, it follows that
\[
\norm{\sum_{i\in\mathcal{I}}u_{i}u_{i}^{T}-\sum_{i\in\mathcal{I}}v_{i}v_{i}^{T}}_{F}\leq\sum_{i\in\mathcal{I}}\left(\norm{u_{i}u_{i}^{T}}_{F}+\norm{v_{i}v_{i}^{T}}_{F}\right)\leq\frac{2d}{n}|\mathcal{I}|\leq\frac{8}{d}\Delta(U)\leq2\sqrt{\frac{\Delta(U)}{d}}.
\]
For the second term, we have that 
\begin{align*}
\norm{\sum_{i=1}^{n}\gamma_i u_{i}u_{i}^{T}}_{F}^{2} 
=\sum_{i=1}^{n} \sum_{j=1}^n \gamma_{i}\gamma_{j}(u_{i}^{T}u_{j})^{2}
\leq\sum_{i=1}^{n} \sum_{j=1}^n \left(\frac{\gamma_{i}^{2}}{2}+\frac{\gamma_{j}^{2}}{2}\right)(u_{i}^{T}u_{j})^{2}
 =\sum_{i=1}^{n}\gamma_{i}^{2}\tr(u_{i}u_{i}^{T}\sum_{j=1}^{n}u_{j}u_{j}^{T}).
\end{align*}
Using that $\Delta(U)\leq d$, it follows that $\sum_{j=1}^{n}u_{j}u_{j}^{T}\preceq2I$ and hence 
\begin{align*}
\norm{\sum_{i=1}^{n}\gamma_i u_{i}u_{i}^{T}}_{F}^{2} 
\leq2\sum_{i=1}^{n}\gamma_{i}^{2}\tr(u_{i}u_{i}^{T})
=2\sum_{i\notin\mathcal{I}}\frac{(\norm{u_{i}}^{2}-\frac{d}{n})^{2}}{\norm{u_{i}}^{2}}
\leq4\frac{n}{d}\sum_{i\notin\mathcal{I}}(\norm{u_{i}}^{2}-\frac{d}{n})^{2}
\leq\frac{4}{d}\Delta(U).
\end{align*}
Combining both terms, we have that
\[
\norm{\sum_{i=1}^n u_{i}u_{i}^{T}-\sum_{i=1}^n v_{i}v_{i}^{T}}_{F}\leq4\sqrt{\frac{\Delta(U)}{d}}.
\]
Therefore,
\[
\Delta(V) \leq 2d\norm{I_d-\sum_{i=1}^{n} u_{i}u_{i}^{T}}_F^{2} 
+ 2d\norm{\sum_{i=1}^{n}u_{i}u_{i}^{T} -\sum_{i=1}^{n}v_{i}v_{i}^{T}}_F^{2}
\leq 2\Delta(U) + 16\Delta(U).
\]
\end{proof}

Next, we bound the movement in the rest of the perturbation process.

\begin{lemma} \label{l:perturb-movement}
Let $U = \{u_1, \ldots, u_n\}$ where $u_i \in \R^d$ and $\norm{u_i}^2 = d/n$ 
for $1 \leq i \leq n$ be the vectors after the preprocessing step of the perturbation process.
Then the expected squared distance between $U$ and the output $W$ of the perturbation process in Procedure~\ref{proc:perturb} is
\[
\E \dist(U,W) = \E\sum_{i=1}^n \norm{u_{i}-w_{i}}^{2}
\leq 2 \sigma^2 d n.
\]
\end{lemma}
\begin{proof}
By the triangle inequality, for any $1 \leq j \leq n$,
\begin{align*}
\norm{u_{j}-w_{j}} 
& \leq\norm{u_{j}-v_{j}}+\norm{v_{j}-w_{j}}
= \norm{z_j} + \norm{v_j - \sqrt{\frac{d}{n}} \frac{v_j}{\norm{v_j}}}
\\
& =\norm{z_{j}} + \left|\norm{v_{j}}-\sqrt{\frac{d}{n}}\right|
=\norm{z_{j}} + \Big|\norm{u_{j}+z_j}-\norm{u_j}\Big| \leq 2\norm{z_j},
\end{align*}
where we used $v_j = u_j + z_j$ in the perturbation process and $\norm{u_j}^2 = d/n$.
Therefore,
\[
\E \dist(U,W) 
= \E \sum_{j=1}^n \norm{u_j-w_j}^2 
\leq 2\E \sum_{j=1}^n \norm{z_j}^2 
\leq 2 \sum_{j=1}^n \sigma^2 d 
= 2\sigma^2 d n,
\]
where the second inequality follows because $\E \norm{x_j}^2 = d\sigma^2$ by definition and $z_j$ is a projection of $x_j$ as stated in Procedure~\ref{proc:perturb}.
\end{proof}

\subsubsection*{Facts and Results for the Rest of the Analysis}

We list some facts and results for the analysis of the perturbation process.
First we start with some facts about projection matrices.

\begin{fact}[projection matrices] \label{f:projection}
Some basic facts about projection matrices that we will use:
\begin{enumerate}
\item
An orthogonal projection matrix $P \in \R^{n \times n}$ has only two eigenvalues $1$ and $0$, with multiplicity $\rank(P)$ and $n-\rank(P)$ respectively. 
In particular, $\tr(P) = \rank(P)$ as trace is equal to the sum of eigenvalues.
\item
For a real symmetric matrix $A \in \R^{n \times n}$, it holds that $\norm{A}_F^2 = \sum_{i=1}^n \lambda_i(A)^2$ where $\lambda_i(A)$ is the $i$-th eigenvalue of $A$.
In particular, for an orthogonal projection matrix $P$,
it follows that $\norm{P}_F^2 = \rank(P)$.
\item
For an orthogonal projection matrix $P \in \R^{n \times n}$, it holds that $0 \preceq P \preceq I_n$, and any principle submatrix $Q$ of $P$ satisfies $0 \preceq Q \preceq I$.  In particular, if the principle submatrix $Q$ of $P$ is a $k \times k$ matrix, then $\norm{Q}_F^2 \leq k$.
\item
For a positive semidefinite matrix $A \in \R^{n \times n}$, it holds that
$\norm{A}_F \leq \tr(A)$. 
This follows from $\tr(A) = \sum_{i=1}^n \lambda_i(A)$, $\lambda_i(A) \geq 0$ for all $i$ by our assumption and the previous fact.
\end{enumerate}
\end{fact}

Next we state some results about Gaussian distributions that we will use. 

\begin{fact}[Gaussian distribution] \label{f:Gaussian}
Some results about Gaussian distributions that we will use:
\begin{enumerate}
\item \label{t:moments} (Moment Bound: Theorem 5.22 in~\cite{LV07}) 
For any mean $0$ log-concave distribution $p$, it holds that 
\[
\E_{x \sim p}\norm x^{k} \leq (2k)^{k}\left(\E_{x\sim p}\norm x^{2} \right)^{k/2}.
\]
Since multivariate Gaussian distributions are log-concave, we can apply this results to $x,y,z$ in the perturbation process.

\item \label{t:Isserlis} (Isserlis' Theorem)
If $(x_1,x_2,x_3,x_4)$ is a zero-mean multivariate normal random vector, then
\[
\E[x_1 x_2 x_3 x_4] = \E[x_1 x_2] \E[x_3 x_4] + \E[x_1 x_3] \E[x_2 x_4] + \E[x_1 x_4] \E[x_2 x_3].
\]

\item \label{t:covariance}
Let $y \sim N(0,\Sigma)$.  Then
\[{\rm Var}[\norm{y}_2^2] = \E[\norm{y}_2^4] - (\E[\norm{y}_2^2])^2 = 2\norm{\Sigma}_F^2.\]
This can be seen by reducing to the diagonal case and using the fact that ${\rm Var}[x^2] = 2$ for $x \sim N(0,1)$.
\end{enumerate}
\end{fact}

Finally, we state some concentration inequalities.

\begin{fact}[concentration inequalities] \label{f:concentration}
Some concentration inequalities that we will use:
\begin{enumerate}
\item (Chernoff bound)
Let $X_1, \ldots, X_n$ be independent random variables with $X_i \in \{0,1\}$.
Let $X = \sum_{i=1}^n X_i$.
Then, for any $\delta>0$
\[
\P( X \geq (1+\delta)\E X) \leq \left( \frac{e^{\delta}}{(1+\delta)^{\delta}} \right)^{\E X}.
\]
\item
(Theorem 6.1 of~\cite{LM00})
For any $\delta \geq 0$, we have
\[
\P_{z\sim N(0,I_{n})}\left(z^{T}Az\geq\tr A+2\norm A_{F}\sqrt{\delta}+2\norm A_{2}\delta\right)\leq e^{-\delta}.
\]
\item
If $I\succeq A\succeq0$, for any $0 \leq \delta \leq \frac{1}{e^2} \tr A$,
\[
\P_{z\sim N(0,I_{n})}\left(z^{T}Az\leq \delta\right)\leq\left(\frac{\delta}{\tr A}\right)^{\frac{1}{4}\tr A}.
\]
\end{enumerate}
\end{fact}
\begin{proof}
We include a proof of the last concentration inequality for completeness.
First, we can assume that $A$ is a diagonal matrix.
Otherwise, we can write $A= U D U^T$ where the columns in $U$ are orthonormal and $D$ is a diagonal matrix.
Then $\tr(A) = \tr(U D U^T) = \tr(D U^T U) = \tr(D)$ and $U^Tz$ is still distributed as $N(0,I_n)$, both are because the columns in $U$ are orthonormal.
Therefore, the problem is reduced to proving the inequality for the diagonal matrix $D$.

So we assume $A$ is diagonal and let $A_{ii}=a_{i}$ where $0 \leq a_i \leq 1$.
Then $z^{T}Az=\sum_{i}a_{i}z_{i}^{2}$. 
By the moment generating function of the chi-squared distribution, 
we know that for $u < 1/2$,
\[
\E_{z_i} e^{uz_i^{2}}=(1-2u)^{-1/2}.
\]
Since $z_{i}$ are independent, for any $u > -1/2$,
\[
\E_{z}e^{-u(\sum_{i}a_{i}z_{i}^{2})}
=\prod_{i}(1+2ua_{i})^{-1/2}
= e^{-\frac{1}{2}\sum_{i}\log(1+2ua_{i})}.
\]
By Markov's inequality,
\[
\P_z\Big(\sum_{i}a_{i}z_{i}^{2}\leq \delta\Big)
= \P_z\Big(e^{-u(\sum_i a_i z_i^2)} \leq e^{-u\delta}\Big)
\leq \frac{\E_{z}e^{-u(\sum_{i}a_{i}z_{i}^{2})}}{e^{-u\delta}}
\leq e^{u\delta-\frac{1}{2}\sum_{i}\log(1+2ua_{i})}.
\]
Let $\lambda\defeq\sum a_{i}$. 
To upper bound the RHS, we would like to lower bound the term $\sum_{i}\log(1+2ua_{i})$ for any $0 \leq a_i \leq 1$ with $\sum a_i = \lambda$.
Since $\sum_{i}\log(1+2ua_{i})$ is concave, 
the minimum is achieved when we set $\lfloor \lambda \rfloor$ of $a_i$ to be one, one $a_i$ to be $\lambda - \lfloor \lambda \rfloor$ and the rest to be zero,
and this gives us a lower bound
\[
\sum_{i}\log(1+2ua_{i})
\geq \lambda\log(1+2u).
\]
Plugging this back into the previous inequality and choosing $u=\frac{\lambda}{2\delta}-\frac{1}{2}$, we get 
\begin{align*}
\P_z(\sum_{i}a_{i}z_{i}^{2}\leq \delta)
\leq  \exp(u\delta-\frac{\lambda}{2}\log(1+2u))
\leq  \exp(\frac{\lambda}{2}-\frac{\delta}{2}-\frac{\lambda}{2}\log(\frac{\lambda}{\delta})).
\end{align*}
Therefore, when $\delta \leq \lambda / e^2$, we conclude that 
\begin{align*}
\P_z(z^T A z \leq \delta) 
= \P_z(\sum_{i}a_{i}z_{i}^{2} \leq \delta)
\leq \exp(-\frac{\lambda}{4}\log(\frac{\lambda}{\delta}))
= \left(\frac{\delta}{\lambda}\right)^{\frac{\lambda}{4}}
= \left( \frac{\delta}{\tr A} \right)^{\frac14 \tr A}.
\end{align*}
\end{proof}

\subsection{Bounding the Increase of $\Delta$ in the Perturbation Process}
\label{ss:Delta}

In this subsection, we bound the increase of $\Delta$ in the perturbation process after the preprocessing step.
The proof is a bit long as there are many terms to keep track of.

\begin{proposition} \label{p:perturb-Delta}
Let $U = \{u_1, \ldots, u_n\}$ where $u_i \in \R^d$ and $\norm{u_i}^2 = d/n$ 
for $1 \leq i \leq n$ be the vectors after the preprocessing step of the perturbation process.
Assume $\Delta(U)\leq 1$ and $\sigma^2 \leq\frac{1}{n}$.
Then the expected value of $\Delta(W)$ for the output $W$ of the perturbation process in Procedure~\ref{proc:perturb} is
\[
\E\Delta(W) \leq 6\Delta(U) + 40 \sigma^4 n^{2} \sqrt{\Delta(U)} + 10^{7} \sigma^4 d^{3} n + 10^{14} \sigma^6 d^{3} n^{3}.
\]
\end{proposition}
\begin{proof}
Since $\norm{w_i}^2 = d/n$ by the postprocessing step of the perturbation process, by Definition~\ref{d:normal-Delta},
\[
\Delta(W) = d \tr(I_d - \sum_{i=1}^n w_i w_i^T)^2 = d \norm{I_d - \sum_{i=1}^n w_i w_i^T}_F^2.
\]
By the property of the first subspace $L_1$,
\[
\norm{v_i}^2 = \norm{u_i + z_i}^2 = \norm{u_i}^2 + 2\inner{u_i}{z_i} + \norm{z_i}^2 = \norm{u_i}^2 + \norm{z_i}^2,
\]
as $\inner{u_i}{z_i} = 0$ for $1 \leq i \leq n$ by construction.
Therefore, by the definition of $w_i$ and $\norm{u_i}^2 = d/n$,
\begin{align*}
\sum_{i=1}^n w_{i} w_{i}^{T} 
= \frac{d}{n} \sum_{i=1}^n \frac{v_i v_i^T}{\norm{v_i}^2} 
= \frac{d}{n} \sum_{i=1}^n \frac{v_i v_i^T}{\norm{u_i}^2 + \norm{z_i}^2}
& = \sum_{i=1}^n \frac{v_i v_i^T}{1 + \frac{n}{d} \norm{z_i}^2}.
\end{align*}
Using that $\frac{1}{1+x} = 1-x+\frac{x^2}{1+x}$,
we split the right hand side of the above equality into three terms so that
\[
\sum_{i=1}^n w_{i} w_{i}^{T} =
\sum_{i=1}^n v_i v_i^T - \frac{n}{d} \sum_{i=1}^n v_i v_i^T \norm{z_i}^2 
+ \frac{n^2}{d^2} \sum_{i=1}^n v_i v_i^T \frac{\norm{z_i}^4}{1+\frac{n}{d} \norm{z_i}^2}.
\]
We will show that the first term $\sum_{i=1}^n v_i v_i^T \approx (1+\sigma^2 \frac{n}{d} (d-1))I$, and the second term $\sum_{i=1}^n v_iv_i^T \norm{z_i}^2 \approx \sigma^2 (d-1)I$.
So, our plan is to bound
\begin{align}
\frac1d \E \Delta(W) 
& = \E \norm{\sum_{i=1}^n w_i w_i^T - I}_F^2
= \E \norm{\sum_{i=1}^n v_i v_i^T - I - \frac{n}{d} \sum_{i=1}^n v_i v_i^T \norm{z_i}^2 
+ \frac{n^2}{d^2} \sum_{i=1}^n v_i v_i^T \frac{\norm{z_i}^4}{1+\frac{n}{d} \norm{z_i}^2}}_F^2 \nonumber
\\
& \leq 
3 \E \norm{\sum_{i=1}^n v_i v_i^T - (1+\sigma^2 \frac{n}{d} (d-1)I }_F^2
+ \frac{3n^2}{d^2} \E \norm{\sum_{i=1}^n v_i v_i^T \norm{z_i}^2 - \sigma^2 (d-1)I}_F^2 \nonumber
\\
& \quad \quad \quad + \frac{3n^4}{d^4} \E \norm{\sum_{i=1}^n v_i v_i^T \frac{\norm{z_i}^4}{1+\frac{n}{d} \norm{z_i}^2}}_F^2,
\label{e:master}
\end{align}
where we used the inequality that $\norm{A+B+C}_F^2 \leq 3\norm{A}_F^2 + 3\norm{B}_F^2 + 3\norm{C}_F^2$ for symmetric matrices $A,B,C$.

We will bound the three terms separately in the following three claims,
and then we will combine the bounds to prove the lemma.
The following claim relies on the second subspace in the perturbation process.

\begin{claim} \label{c:first-term}
Let $u_1, \ldots, u_n \in \R^d$ be such that $\norm{u_{j}}^{2}=\frac{d}{n}$ for $1 \leq j \leq n$.
Assume $n \geq d^{2}$ and $\Delta(U)\leq d^{4}$. Then
\[
\E\norm{\sum_{i=1}^ n v_{i} v_{i}^{T} - (1+\sigma^2 \frac{n}{d}(d-1))I}_{F}^{2}
\leq 2\norm{\sum_{i=1}^n u_{i} u_{i}^{T}-I}_{F}^{2}+ 
8 \sigma^4 \left(\frac{n^{2}}{d}\sqrt{\Delta(U)}+nd^{2}\right).
\]
\end{claim}
\begin{proof}
Recall that $v_i = u_i + z_i$.
Since $z$ is in the second subspace $L_2$,
we have the important property that the ``cross terms''
\[
\sum_{i=1}^n u_i z_i^T = \sum_{i=1}^n z_i u_i^T = 0.
\]
This implies that
\[
\sum_{i=1}^n v_i v_i^T = \sum_{i=1}^n \big( u_i u_i^T + u_i z_i^T + z_i u_i^T + z_i z_i^T \big) =
\sum_{i=1}^n u_i u_i^T + \sum_{i=1}^n z_i z_i^T.
\]
The main work of this claim is to bound
\[
\E \norm{\sum_{i=1}^n z_i z_i^T}_F^2
= \E \tr\Big( \sum_{i=1}^n z_i z_i^T \sum_{j=1}^n z_j z_j^T \Big)
= \E \sum_{i=1}^n \sum_{j=1}^n \inner{z_i}{z_j}^2
= \E \sum_{i=1}^n \sum_{j=1}^n \sum_{l_1 = 1}^d \sum_{l_2=1}^d z_{l_1,i} z_{l_1,j} z_{l_2,i} z_{l_2,j},
\]
where we recall that $z_{l,i} = (z_i)_l$ as defined in Procedure~\ref{proc:perturb}.
We apply Isserlis' theorem in Fact~\ref{f:Gaussian} to break the right hand side into three terms so that
\begin{equation} \label{e:Isserlis}
\E \norm{\sum_{i=1}^n z_i z_i^T}_F^2
= 
\sum_{i,j,l_1,l_2} \E(z_{l_1,i} z_{l_1,j}) \E(z_{l_2, i} z_{l_2,j}) +
\sum_{i,j,l_1,l_2} \E(z_{l_1,i} z_{l_2,i}) \E(z_{l_1, j} z_{l_2,j}) +
\sum_{i,j,l_1,l_2} \E(z_{l_1,i} z_{l_2,j}) \E(z_{l_1, j} z_{l_2,i}).
\end{equation}
To bound these terms, we consider the $(d \times n) \times (d \times n)$ matrix $Z$ where 
\[Z_{(l_1,i),(l_2,j)} = \E(z_{l_1,i} z_{l_2, j}).\] 
By the perturbation process, 
\[Z = \sigma^2 P_{L_1 \cap L_2}.\]
To bound the first term in~(\ref{e:Isserlis}),
let $Z^{(l)}$ be the $n \times n$ matrix where 
$(Z^{(l)})_{ij} := \E(z_{l,i} z_{l,j})$.
Then, the first term is
\[
\sum_{i,j,l_1,l_2} \E(z_{l_1,i} z_{l_1,j}) \E(z_{l_2, i} z_{l_2,j})
= \sum_{i,j,l_1,l_2} Z^{(l_1)}_{ij} Z^{(l_2)}_{ij}
= \sum_{l_1,l_2} \tr( Z^{(l_1)} Z^{(l_2)} )
\leq \sum_{l_1,l_2} \norm{Z^{(l_1)}}_F \norm{Z^{(l_2)}}_F,
\]
where the inequality is by Cauchy-Schwarz.
Note that $Z^{(l)}$ is a principle submatrix of the matrix $Z = \sigma^2 P_{L_1 \cap L_2}$ where $P_{L_1 \cap L_2}$ is a projection matrix, so $Z^{(l)} \preceq \sigma^2 I_n$ and hence $\norm{Z^{(l)}}_F^2 \leq \sigma^4 n$ by Fact~\ref{f:projection}(2).
Therefore, the first term is bounded by 
\[
\sum_{i,j,l_1,l_2} \E(z_{l_1,i} z_{l_1,j}) \E(z_{l_2, i} z_{l_2,j})
\leq \sum_{l_1=1}^d \sum_{l_2=1}^d \norm{Z^{(l_1)}}_F \norm{Z^{(l_2)}}_F
\leq \sigma^4 d^2 n.
\]
To bound the third term in~(\ref{e:Isserlis}),
let $Z^{(i,j)}$ be the $d \times d$ matrix where $Z^{(i,j)}_{l_1 l_2} = \E(z_{l_1,i} z_{l_2,j})$.
Then, the third term is
\begin{align*}
\sum_{i,j,l_1,l_2} \E(z_{l_1,i} z_{l_2,j}) \E(z_{l_1, j} z_{l_2,i})
& = \sum_{i,j,l_1,l_2} Z^{(i,j)}_{l_1,l_2} Z^{(j,i)}_{l_1,l_2}
= \sum_{i,j} \tr(Z^{(i,j)} Z^{(j,i)})
\\
& \leq \sum_{i,j} \norm{Z^{(i,j)}}_F \norm{Z^{(j,i)}}_F
= \sum_{i,j} \norm{Z^{(i,j)}}_F^2 = \norm{Z}_F^2 \leq \sigma^4 dn.
\end{align*}
where the first inequality is by Cauchy-Schwarz
and the second inequality is by Fact~\ref{f:projection}(2) as $Z = \sigma^2 P_{L_1 \cap L_2}$.

The second term in~(\ref{e:Isserlis}) requires more care.
Let $Z^{(i)}$ be the $d \times d$ matrix where $Z^{(i)}_{l_1,l_2} = E(z_{l_1,i} z_{l_2,i})$.
Then, the second term is
\[
\sum_{i,j,l_1,l_2} \E(z_{l_1,i} z_{l_2,i}) \E(z_{l_1, j} z_{l_2,j})
= \sum_{i,j,l_1,l_2} Z^{(i)}_{l_1,l_2} Z^{(j)}_{l_1,l_2}
= \sum_{i,j} \tr(Z^{(i)} Z^{(j)})
= \norm{\sum_{i=1}^n Z^{(i)}}_F^2.
\]
Note that we can bound the second term as in the first term to get a bound $\sigma^4 d n^2$, but we could not afford the $n^2$ factor.
We bound $\norm{\sum_{i=1}^n Z^{(i)}}_F^2$ by looking at the special structure in the subspace $L_1$.
Let $y \in \R^{d \times n}$ be the (noise) vector in the perturbation process,
and let $Y$ be the $(d \times n) \times (d \times n)$ matrix and $Y^{(i)}$ be the $d \times d$ matrix where
\[Y_{(l_1,i),(l_2,j)} = \E(y_{l_1,i} y_{l_2,j}) {\rm~so~that~} Y = \sigma^2 P_{L_1},
\quad {\rm and} \quad
Y^{(i)}_{l_1,l_2} = \E(y_{l_1,i} y_{l_2,i}).\]
Note that $0\preceq Z^{(i)}\preceq Y^{(i)}$ for all $i$ and hence $\norm{\sum_{i=1}^{n}Z^{(i)}}_{F}^{2}\leq\norm{\sum_{i=1}^{n}Y^{(i)}}_{F}^{2}$.
Now, we use the special structure of $L_1$ to bound $\norm{\sum_{i=1}^n Y^{(i)}}_F$.
As stated in step~(3) in the perturbation process in Procedure~\ref{proc:perturb}, $Y^{(i)}$ can be described explicitly as 
\[Y^{(i)} = \sigma^2(I_d - \frac{n}{d} u_i u_i^T).\]
\begin{align*}
{\rm So,~} \norm{\sum_{i=1}^n Y^{(i)}}_F
& = \sigma^2 \norm{\sum_{i=1}^n (I_d - \frac{n}{d} u_i u_i^T)}_F
= \sigma^2 \norm{nI_d - \frac{n}{d} \sum_{i=1}^n u_i u_i^T}_F
\\
& \leq \sigma^2 \left( \norm{\frac{n}{d} (I_d-\sum_{i=1}^n u_i u_i^T)}_F + \norm{(n-\frac{n}{d})I_d}_F \right) 
= \sigma^2 \left( \frac{n}{d} \sqrt{\frac{\Delta(U)}{d}} + n(1 - \frac{1}{d}) \sqrt{d} \right),
\end{align*}
where the last equality is by Definition~\ref{d:normal-Delta} and the assumption that $\norm{u_i}^2 = d/n$ for $1 \leq i \leq n$.
Hence,
\[
\norm{\sum_{i=1}^n Z^{(i)}}_F 
\leq \sigma^2 \left( \frac{n}{d} \sqrt{\frac{\Delta(U)}{d}} + n(1 - \frac{1}{d}) \sqrt{d} \right).
\]
Squaring both sides, we get that the second term is
\begin{align*}
\norm{\sum_{i=1}^n Z^{(i)}}_F^2
\leq \sigma^4 \left( \frac{n^2 \Delta(U)}{d^3} + dn^2(1 - \frac{1}{d})^2 + \frac{2n^2}{d} (1-\frac{1}{d})\sqrt{\Delta(U)} \right)
\leq \sigma^4 \left( dn^2(1 - \frac{1}{d})^2 + \frac{3n^2}{d}\sqrt{\Delta(U)} \right),
\end{align*}
where the last inequality follows from the assumption that $\Delta(U) \leq d^4$.
Putting all three bounds back to~(\ref{e:Isserlis}), we finally have
\[
\E \norm{\sum_{i=1}^n z_i z_i^T}_F^2
\leq \sigma^4 \left(
dn^2(1 - \frac{1}{d})^2 + \frac{3n^2}{d}\sqrt{\Delta(U)} + 2 d^2 n \right).
\]
Going back to the left hand side of the claim,
\begin{align}
\E\norm{\sum_{i=1}^ n v_{i} v_{i}^{T} - (1+\sigma^2 \frac{n}{d}(d-1))I}_{F}^{2}
& \leq
2\E\norm{\sum_{i=1}^n u_i u_i^T - I}_F^2 + 2\E\norm{\sum_{i=1}^n z_i z_i^T - \sigma^2 \frac{n}{d} (d-1) I}_F^2,
\label{e:first-first}
\end{align}
where the second term in~(\ref{e:first-first}) is
\begin{align*}
& 2\E \norm{\sum_{i=1}^n z_i z_i^T}_F^2 - 4 \sigma^2 n (1-\frac{1}{d}) \E \sum_{i=1}^n \norm{z_i}^2 + 2\sigma^4 d n^2 (1-\frac{1}{d})^2.
\end{align*}
Note that
\[
\E \sum_{i=1}^n \norm{z_i}^2 
= \sum_{i=1}^n \tr(Z^{(i)}) 
= \tr(Z) 
= \sigma^2 \tr(P_{L_1 \cap L_2})
= \sigma^2 \rank(P_{L_1 \cap L_2}) 
\geq \sigma^2 (dn - n - d^2),
\]
where the last equality is by Fact~\ref{f:projection}(1) and the inequality is because the codimension of $L_1$ is at most $n$ and the codimension of $L_2$ is at most $d^2$.
Using the bound for $\E \norm{\sum_{i=1}^n z_i z_i^T}_F^2$,
the second term in~(\ref{e:first-first}) is at most
\begin{align*}
&~ 2\sigma^4 \left( dn^2(1 - \frac{1}{d})^2 + \frac{3n^2}{d} \sqrt{\Delta(U)} + 2 d^2 n - 2n (1-\frac{1}{d}) (dn - n - d^2) + dn^2(1-\frac{1}{d})^2
\right)
\\
= &~ 2\sigma^4 \left( \frac{3n^2}{d} \sqrt{\Delta(U)} + 2d^2n +2nd^2(1-\frac{1}{d}) \right) 
\leq 2\sigma^4 \left(\frac{3n^2}{d} \sqrt{\Delta(U)} + 4 nd^{2} \right).
\end{align*}
Putting this back into the second term of~(\ref{e:first-first}) proves the claim.
\end{proof}

In the second term of~(\ref{e:master}), we show that $\sum_{i=1}^n v_{i}v_{i}^{T}\norm{z_{i}}^{2}$ is close to the matrix $\sigma^2(d-1)I$.

\begin{claim} \label{c:second-term}
Let $u_1, \ldots, u_n \in \R^d$ be such that $\norm{u_{j}}^{2}=\frac{d}{n}$ for $1 \leq j \leq n$.
Assuming $\sigma^2 \leq \frac{1}{n}$ and $\Delta(U) \leq d$, then
\[
\E\norm{\sum_{i=1}^n v_{i} v_{i}^{T} \norm{z_{i}}^{2} - \sigma^2 (d-1)I}_{F}^{2}
\leq 4\sigma^{4} d \Delta(U)+10^{6}\frac{d^{4}}{n}\sigma^{4} + 10^{8} \sigma^{6}d^{4}n.
\]
\end{claim}
\begin{proof}
Recall that $v_i = u_i + z_i$, so
\begin{equation*}
\sum_{i=1}^n v_{i} v_{i}^{T} \norm{z_{i}}^{2}
= \sum_{i=1}^n u_{i} u_{i}^{T} \norm{z_{i}}^{2} 
+ \sum_{i=1}^n (u_{i}z_{i}^{T} + z_{i}u_{i}^{T} + z_{i}z_{i}^{T})\norm{z_{i}}^{2}.
\end{equation*}
To bound the left hand side of the claim,
we use the above equation to split the left hand side into two terms
\begin{align}
&~\E\norm{\sum_{i=1}^n v_{i} v_{i}^{T} \norm{z_{i}}^{2} - \sigma^2 (d-1)I}_{F}^{2} \nonumber
\\
\leq &~
2\E\norm{\sum_{i=1}^n u_{i} u_{i}^{T} \norm{z_{i}}^{2} - \sigma^2 (d-1)I}_{F}^{2}
+ 2\E\norm{\sum_{i=1}^n (u_{i}z_{i}^{T} + z_{i}u_{i}^{T} + z_{i}z_{i}^{T})\norm{z_{i}}^{2}}_F^2.
\label{e:2nd-master}
\end{align}

To bound the first term on the right hand side of (\ref{e:2nd-master}), we let $$r_i = \norm{z_i}^2 - \norm{y_i}^2.$$
We further split the first term in~(\ref{e:2nd-master}) into two terms such that
\begin{align}
\E\norm{\sum_{i=1}^n u_{i}u_{i}^{T} \norm{z_{i}}^{2} - \sigma^2(d-1)I}_{F}^{2}
\leq 
 2\E\norm{\sum_{i=1}^n u_{i} u_{i}^{T} \norm{y_{i}}^{2} - \sigma^2(d-1)I}_{F}^{2}+2\E\norm{\sum_{i=1}^n r_i u_{i} u_{i}^{T}}_{F}^{2}. \label{eq:2n_step_sub}
\end{align}
Recall from point (3) of Procedure~\ref{proc:perturb} that
$y_i = x_i - \frac{n}{d} \inner{x_i}{u_i} u_i$,
this allows us to compute $\E \norm{y_i}^2$ exactly as
\begin{equation} \label{e:y^2}
\E \norm{y_i}^2 = \E \norm{x_i}^2 - \E \frac{n}{d} \inner{x_i}{u_i}^2
= \sigma^2 d - \frac{n}{d} \sigma^2 \norm{u_i}^2
= \sigma^2 (d - 1),
\end{equation}
where the second equality is because $x_i$ is a $d$-dimensional vector in which each coordinate is an independent Gaussian variable with variance $\sigma^2$.
Therefore, the first term in (\ref{eq:2n_step_sub}) is
\begin{align*}
&~\E\norm{\sum_{i=1}^n u_{i} u_{i}^{T}\norm{y_{i}}^{2} - \sigma^2 (d-1)I}_{F}^{2} \\
= &~ \sum_{i=1}^n \sum_{j=1}^n \inner{u_{i}}{u_{j}}^{2} \cdot \E\norm{y_{i}}^{2}\norm{y_{j}}^{2}
- 2 \sigma^2(d-1) \sum_{i=1}^n \norm{u_{i}}^{2}\E\norm{y_{i}}^{2}
+  \sigma^{4}(d-1)^{2}d
\\
= &~ \sigma^{4}(d-1)^{2}\left(\sum_{i=1}^n \sum_{j=1}^n \inner{u_{i}}{u_{j}}^{2}- 2\sum_{i=1}^n \norm{u_{i}}^{2}+d\right)
+ \sum_{i=1}^n \norm{u_i}^4 \left(\E \norm{y_i}^4_2 - (\E \norm{y_i}^2)^2 \right)
\\
= &~ \sigma^{4}(d-1)^{2}\norm{\sum_{i=1}^n u_{i} u_{i}^{T}-I}_{F}^{2}
+ \frac{d^2}{n^2} \sum_{i=1}^n \left(\E \norm{y_i}^4_2 - (\E \norm{y_i}^2)^2 \right)
\\
= &~ \sigma^4 (d-1)^2 \frac{\Delta(U)}{d} + \frac{d^2}{n^2} \sum_{i=1}^n \left(\E \norm{y_i}^4_2 - (\E \norm{y_i}^2)^2 \right) \\
\leq &~ \sigma^4 d \Delta(U) + \frac{d^2}{n^2} \sum_{i=1}^n \left(\E \norm{y_i}^4_2 - (\E \norm{y_i}^2)^2 \right),
\end{align*}
where the second equality uses that $y_i$ and $y_j$ are independent random variables for $i\neq j$ and $\E\norm{y_i}^2 = \sigma^2(d-1)$ as calculated above,
the third equality is by our assumption that $\norm{u_i}^2 = d/n$ for $1 \leq i \leq n$,
and the last equality follows from Definition~\ref{d:normal-Delta} and the same assumption.

For the variance term, since $y$ follows a multivariate distribution with covariance matrix $\sigma^2 (I_d - \frac{n}{d} u_i u_i^{T})$, Fact~\ref{f:Gaussian}(3) implies that
$$
\E\norm{y_{i}}^{4}-\left(\E\norm{y_{i}}^{2}\right)^{2}
= 2\sigma^{4}\norm{I_d-\frac{n}{d}u_{i}u_{i}^{T}}_{F}^{2}
= 2\sigma^{4}(d-1),
$$
where the last equality uses the assumption that $\norm{u_i}^2 = d/n$ for $1 \leq i \leq n$.
Therefore, the first term in (\ref{eq:2n_step_sub}) is
\begin{equation} \label{e:c2-1}
\E\norm{\sum_{i=1}^n u_{i} u_{i}^{T}\norm{y_{i}}^{2} - \sigma^4 (d-1)I}_{F}^{2}
\leq  \sigma^4 d \Delta(U) +  2 \sigma^4 \frac{d^3}{n}.
\end{equation}

For the second term in (\ref{eq:2n_step_sub}), we note as in Lemma~\ref{l:perturb-preprocess} that
\begin{align*}
\E\norm{\sum_{i}r_{i}u_{i}u_{i}^{T}}_{F}^{2} & =\E\sum_{i,j=1}^{n}r_{i}r_{j}(u_{i}^{T}u_{j})^{2}\leq\E\sum_{i,j=1}^{n}\left(\frac{r_{i}^{2}}{2}+\frac{r_{j}^{2}}{2}\right)(u_{i}^{T}u_{j})^{2}
 =\E\sum_{i=1}^{n}r_{i}^{2}\tr(u_{i}u_{i}^{T}\sum_{j=1}^{n}u_{j}u_{j}^{T}).
\end{align*}
Using that $\Delta(U)\leq d$, we have that $\sum_{j=1}^{n}u_{j}u_{j}^{T}\preceq2I$
and hence
\begin{align*}
\E\norm{\sum_{i}r_{i}u_{i}u_{i}^{T}}_{F}^{2} & 
\leq2\E\sum_{j=1}^{n}r_{j}^{2}\norm{u_{j}}^{2}
=\frac{2d}{n}\E\sum_{j=1}^{n}\left(\norm{z_{j}}^{2}-\norm{y_{j}}^{2}\right)^{2}.
\end{align*}
Note that
\begin{align*}
\left|\norm{z_{i}}^{2}-\norm{y_{i}}^{2}\right| & \leq\left(\norm{y_{i}}+\norm{z_{i}-y_{i}}\right)^{2}-\norm{y_{i}}^{2}=2\norm{z_{i}-y_{i}}\norm{y_{i}}+\norm{z_{i}-y_{i}}^{2},
\end{align*}
which implies that
\begin{align*}
\E\left(\norm{z_{i}}^{2}-\norm{y_{i}}^{2}\right)^{2} & \leq
8\E\norm{z_{i}-y_{i}}^{2}\norm{y_{i}}^{2}+2\E\norm{z_{i}-y_{i}}^{4}.
\end{align*}
By the AM-GM inequality,
\[
\norm{z_{i}-y_{i}}^{2}\norm{y_{i}}^{2}\leq\frac{\E\norm{y_{i}}^{2}}{2\E\norm{z_{i}-y_{i}}^{2}}\norm{z_{i}-y_{i}}^{4}+\frac{\E\norm{z_{i}-y_{i}}^{2}}{2\E\norm{y_{i}}^{2}}\norm{y_{i}}^{4}.
\]
Therefore,
\begin{align*}
\E\left(\norm{z_{i}}^{2}-\norm{y_{i}}^{2}\right)^{2} 
& \leq\frac{4\E\norm{y_{i}}^{2}}{\E\norm{z_{i}-y_{i}}^{2}}\E\norm{z_{i}-y_{i}}^{4}+
\frac{4\E\norm{z_{i}-y_{i}}^{2}}{\E\norm{y_{i}}^{2}}\E\norm{y_{i}}^{4}+2\E\norm{z_{i}-y_{i}}^{4} 
\\
& \leq \frac{4\E\norm{y_{i}}^{2}}{\E\norm{z_{i}-y_{i}}^{2}} \cdot 2^{12} \left(\E\norm{z_{i}-y_{i}}^{2}\right)^{2}
+\frac{4\E\norm{z_{i}-y_{i}}^{2}}{\E\norm{y_{i}}^{2}} \cdot 2^{12} \left(\E\norm{y_{i}}^{2}\right)^{2}
+2 \cdot 2^{12} (\E\norm{z_{i}-y_{i}}^{2})^{2}
\\
 & =2^{15}\E\norm{y_{i}}^{2}\E\norm{z_{i}-y_{i}}^{2}+2^{14}(\E\norm{z_{i}-y_{i}}^{2})^{2},
\end{align*}
where we used Fact~\ref{f:Gaussian} in the second inequality.
Let $P_{L_1 \cap L_2^{\perp}}$ be the projector to the subspace $L_1 \cap L_2^{\perp}$ such that $z-y = \sigma^2 P_{L_1 \cap L_2^{\perp}} x$.
Let $Q^{(i)}$ be the $n \times n$ matrix where $(Q^{(i)})_{l_1,l_2} = (P_{L_1 \cap L_2^{\perp}})_{(l_1,i),(l_2,i)}$ such that $\sigma^2 \tr(Q^{(i)}) = \E\norm{z_i-y_i}^2$; in other words, $Q^{(i)} \in \R^{d \times d}$ is the $(i,i)$-th block of $P_{L_1 \cap L_2^{\perp}}$.
So the second term in~(\ref{eq:2n_step_sub}) is 
\begin{align}
\E\norm{\sum_{i}r_{i}u_{i}u_{i}^{T}}_{F}^{2} & 
\leq2^{16}\frac{d}{n}\sum_{i=1}^{n}\left(\E\norm{y_{i}}^{2}\E\norm{z_{i}-y_{i}}^{2}+(\E\norm{z_{i}-y_{i}}^{2})^{2}\right) \nonumber
\\
 & \leq2^{16}\frac{d}{n}\sigma^{4}\sum_{i=1}^{n}
\left(d\cdot\tr(Q^{(i)})+(\tr(Q^{(i)}))^{2}\right) \nonumber
\\
 & \leq2^{17}\frac{d^{2}}{n}\sigma^{4}\sum_{i=1}^{n}\tr(Q^{(i)}) \nonumber
\\
 & = 2^{17}\frac{d^{2}}{n}\sigma^{4}\tr(P_{L_{1}\cap L_{2}^{\perp}}) \nonumber
\\
& \leq2^{17}\frac{d^{4}}{n}\sigma^{4}, \label{e:c2-2}
\end{align}
where the second inequality is by~(\ref{e:y^2}) and the definition of $Q^{(i)}$,
the third inequality uses Fact~\ref{f:projection}(3) $\tr(Q^{(i)}) \leq d$ as $Q^{(i)}$ is a $d \times d$ principle submatrix of a projector matrix, 
and the last inequality uses Fact~\ref{f:projection}(1) that $L_2$ is of codimension at most $d^2$.
Plugging back the right hand sides of~(\ref{e:c2-1}) and~(\ref{e:c2-2}) into (\ref{eq:2n_step_sub}), we get that 
\begin{equation}
\E\norm{\sum_{i=1}^n u_{i} u_{i}^{T}\norm{z_{i}}^{2}-\sigma^2(d-1)I}_{F}^{2}
\leq 2 \sigma^4 d \Delta(U) +  4 \sigma^4 \frac{d^3}{n} + 2^{18}\frac{d^{4}}{n}\sigma^{4} .
\label{eq:2n_step_master_1}
\end{equation}
For the second term in (\ref{e:2nd-master}), using triangle inequality and Cauchy-Schwarz, 
\begin{align*}
\norm{\sum_{i=1}^n (u_{i}z_{i}^{T}+z_{i}u_{i}^{T}+z_{i}z_{i}^{T})\norm{z_{i}}^{2}}_{F}
\leq & \sum_{i=1}^n (2\norm{u_{i}}_{2}\norm{z_{i}}^{3}+\norm{z_{i}}^{4})
= 2\sqrt{\frac{d}{n}}\sum_{i=1}^n \norm{z_{i}}^{3}+\sum_{i=1}^n\norm{z_{i}}_{2}^{4}.
\end{align*}
Again we bound the higher moments of $\norm{z_i}$ using Fact~\ref{f:Gaussian}(1).
Since $\E \norm{z_i}^2 \leq \sigma^2 d$, 
we can bound
\begin{align}
& \E\norm{\sum_{i=1}^n (u_{i}z_{i}^{T} + z_{i}u_{i}^{T} + z_{i}z_{i}^{T})\norm{z_{i}}^{2}}_{F}^{2}
\leq 8 \frac{d}{n} \E \Big(\sum_{i=1}^n \norm{z_i}^3\Big)^2 + 2 \E \Big(\sum_{i=1}^n \norm{z_i}^4\Big)^2 \nonumber
\\
\leq &~ 8 \frac{d}{n} \Big(6^3 n (\sigma^2 d)^{3/2}\Big)^2 + 
2\Big( 8^4 n (\sigma^2 d)^2\Big)^2
\leq 10^{6} \sigma^{6} d^{4} n 
+ 4 \cdot 10^{7} \sigma^{8} d^4 n^2.
\label{eq:2n_step_master_2}
\end{align}

Combining (\ref{eq:2n_step_master_1}) and (\ref{eq:2n_step_master_2})
into (\ref{e:2nd-master}), we get that
\begin{align*}
\E\norm{\sum_{i=1}^n u_{i}u_{i}^{T}\norm{z_{i}}^{2}-\sigma^2(d-1)I}_{F}^{2} 
& \leq 4 \sigma^4 d \Delta(U) +  8 \sigma^4 \frac{d^3}{n} + 2^{19}\frac{d^{4}}{n}\sigma^{4} 
+ 2 \cdot 10^{6} \sigma^{6} d^{4} n + 8 \cdot 10^{7} \sigma^{8} d^4 n^2
\\
& \leq 4 \sigma^4 d \Delta(U) +  10^{6}\frac{d^{4}}{n}\sigma^{4} + 10^{8} \sigma^{6} d^{4} n,
\end{align*}
where we used the assumption that $\sigma^2 \leq\frac{1}{n}$.
\end{proof}

Finally, we bound the third term of the right hand side of (\ref{e:master}).

\begin{claim} \label{c:third-term}
Let $u_1, \ldots, u_n \in \R^d$ be such that $\norm{u_{j}}^{2}=\frac{d}{n}$ for $1 \leq j \leq n$.
Assume $\sigma^2 \leq \frac{1}{n}$. Then,
\[
\E \norm{\sum_{i=1}^n v_{i}v_{i}^{T}\frac{\norm{z_{i}}^{4}}{1+\frac{n}{d}\norm{z_{i}}^{2}}}^2_{F}
\leq10^{13} \sigma^{8} d^{6}.
\]
\end{claim}
\begin{proof}
By triangle inequality,
\begin{align*}
\norm{\sum_{i=1}^n v_{i}v_{i}^{T}\frac{\norm{z_{i}}^{4}}{1+\frac{n}{d}\norm{z_{i}}^{2}}}_{F}
 & \leq\sum_{i=1}^n \frac{\norm{z_{i}}^{4}}{1+\frac{n}{d}\norm{z_{i}}^{2}}\norm{v_{i}}^{2}
 \leq\sum_{i=1}^n \norm{z_{i}}^{4}(\norm{u_{i}}^{2}+\norm{z_{i}}^{2})
 = \frac{d}{n}\sum_{i=1}^n \norm{z_{i}}^{4}+\sum_{i=1}^n \norm{z_{i}}^{6},
\end{align*}
where we used that $\norm{v_i}^2 = \norm{u_i}^2 + \norm{z_i}^2$ by the property of the subspace $L_1$ and $\norm{u_i}^2 = d/n$ for $1 \leq i \leq n$.
Therefore, by Fact~\ref{f:Gaussian}(1), 
with $\E\norm{z_{i}}_{2}^{2}\leq \sigma^2 d$
and $\sigma^2 \leq\frac{1}{n}$, we can bound 
\begin{align*}
\E \norm{\sum_{i=1}^n v_{i}v_{i}^{T}\frac{\norm{z_{i}}^{4}}{1+\frac{n}{d}\norm{z_{i}}^{2}}}_{F}^2
& \leq \frac{2d^2}{n^2} \E \Big(\sum_{i=1}^n \norm{z_{i}}^{4}\Big)^2 + \E \Big(\sum_{i=1}^n \norm{z_{i}}^{6} \Big)^2
\\
& \leq \frac{2d^2}{n^2} \big( 8^4 n (\sigma^2 d)^2 \Big)^2
+ \big( (12)^6 n (\sigma^2 d)^3 )^2
\leq10^{13} \sigma^{8} d^6.
\end{align*}
\end{proof}

Using Claim~\ref{c:first-term}, Claim~\ref{c:second-term} and Claim~\ref{c:third-term} in~(\ref{e:master}), we have that
\begin{align*}
\E\Delta(W)
\leq &~ 3d \left(2\norm{\sum_{i=1}^n u_{i} u_{i}^{T}-I}_{F}^{2}+8 \sigma^4 \left(\frac{n^{2}}{d}\sqrt{\Delta(U)}+nd^2\right)\right)
\\
& \quad \quad + \frac{3n^2}{d} \left( 4\sigma^{4} d \Delta(U)+10^{6}\frac{d^{4}}{n}\sigma^{4} + 10^{8} \sigma^{6}d^{4}n \right)
+ \frac{3n^4}{d^3} \left( 10^{13} \sigma^{8} d^6 \right)
\\
\leq &~ 
6\Delta(U) + 24 \sigma^4 n^2 \sqrt{\Delta(U)} + 24 \sigma^4 d^3 n
+ 12 \sigma^4 n^2 \Delta(U) 
\\
& \quad \quad + 3 \cdot 10^6 \sigma^4 n d^3 + 3 \cdot 10^8 \sigma^6 d^3 n^3 + 3 \cdot 10^{13} \sigma^8 d^3 n^4
\\
\leq &~ 6\Delta(U) + 40 \sigma^4 n^2 \sqrt{\Delta(U)} + 10^7 \sigma^4 d^3 n + 10^{14} \sigma^6 d^3 n^3,
\end{align*}
where we used our assumptions that $\sigma^2 \leq\frac{1}{n}$ and $\Delta(U) \leq 1$. 
\end{proof}

\subsection{Improved Capacity Lower Bound for a Perturbed Frame} \label{ss:pseudorandom}

Our goal in this subsection is to establish the foundation to prove a strong capacity lower bound on a frame $W$ produced by the perturbation process in Procedure~\ref{proc:perturb}.
To prove a frame capacity lower bound, as in Subsection~\ref{ss:capLB}, we use the reduction in Proposition~\ref{p:reduction-capacity} to construct a $d \times n$ matrix $B$ from $W$ such that $\capa(B) \leq \capa(W)$ and $\Delta(B) \leq \Delta(W)$, and then our goal is reduced to proving a strong lower bound on $\capa(B)$ in terms of $\Delta(B)$.
In Subsection~\ref{ss:Delta'}, we proved in Theorem~\ref{t:pseudorandom} that if a matrix $B$ is pesudorandom, then $\capa(B)$ has a stronger lower bound than the general lower bound in Proposition~\ref{p:matrix-capLB}.
The main work in this subsection is to prove that $B$ is pseudorandom if $W$ is the output of the perturbation process.

Our plan is as follows.
First, using the special structure in the frame setting, we show that the matrix $B$ has a simpler form than that in the general reduction from Proposition~\ref{p:reduction-capacity}.
Then, the main work is a probabilistic analysis to prove that $B$ is pseudorandom.
In the next subsection, we use the pseudorandom property to establish an improved capacity lower bound for a perturbed frame.

\begin{lemma} \label{l:simple}
Given a frame $U=\{u_1, \ldots, u_n\}$ where $u_i \in \R^d$ for $1 \leq i \leq n$, there exists an orthonormal basis $g_1, \ldots, g_d \in \R^d$ so that the $d \times n$ non-negative matrix $A$ with
\[
A_{ij} = \inner{g_i}{u_j}^2 \quad {\rm for~} 1 \leq i \leq d, 1 \leq j \leq n
\]
satisfies the properties
\[\capa(A) \leq \capa(U) \quad {\rm and} \quad 
\Delta(A) \leq \Delta(U) \quad {\rm and} \quad
s(A) = s(U).
\]
Furthermore, if $\norm{u_j}^2 = d/n$ for $1 \leq j \leq n$, then $c_j(A) = d/n$ for $1 \leq j \leq n$.
\end{lemma}
\begin{proof}
The proof is mostly similar to the proof of Proposition~\ref{p:reduction-capacity},
and so we focus on the observation that simplifies the matrix.
Recall in Definition~\ref{d:reduction} that we reduce a frame $U=\{u_1, \ldots, u_n\}$ to an operator $\U=\{U_1, \ldots, U_n\}$ where each $U_l$ is a $d \times n$ matrix with the $l$-th column being $u_l$ and other columns being zero.
From Definition~\ref{d:frame-parameters} the capacity of a frame $U$ is defined as
\[
\capa(U) = \capa(\U) = \inf_{X \succeq 0} \frac{d \det(\sum_{l=1}^n U_l X U_l^T)^{1/d}}{\det(X)^{1/n}} = \inf_{X \succeq 0} \frac{d \det(\sum_{l=1}^n X_{ll} u_l u_l^T)^{1/d}}{\det(X)^{1/n}},
\]
where the last equality follows from the specific $\U$ from a frame $U$ as described above.
As in Proposition~\ref{p:reduction-capacity}, we let $X \in \R^{n \times n}$ be an approximate minimizer to this optimization problem such that 
\[
X \succ 0 \quad {\rm and} \quad \frac{d \det(\sum_{l=1}^n X_{ll} u_l u_l^T)^{1/d}}{\det(X)^{1/n}} \leq \capa(\U) + \delta {\rm~for~some~} \delta>0.
\]
The observation is that $X$ can be assumed to be a diagonal matrix.
Consider the $n \times n$ diagonal matrix $\ol{X}$ with $\ol{X}_{ii} = X_{ii}$ for $1 \leq i \leq n$.
We will show that $\ol{X}$ is a feasible solution to the optimization problem with objective value at most that of $X$.
Firstly, as $X \succ 0$, each diagonal entry of $X$ is positive, and thus $\ol{X} \succ 0$.
Secondly, it is clear that the numerators are the same.
Finally, Hadamard's inequality (see~\cite{HJ91} Theorem~7.8.1) says that for a positive definite matrix $Y \in \R^{n \times n}$, we have $\det(Y) \leq \prod_{i=1}^n Y_{ii}$, which implies that $\det(X) \leq \prod_{i=1}^n X_{ii} = \det(\ol{X})$.
Therefore,
\[
\ol{X} \succ 0 \quad {\rm and} \quad 
\frac{d \det(\sum_{l=1}^n \ol{X}_{ll} u_l u_l^T)^{1/d}}{\det(\ol{X})^{1/n}} \leq
\frac{d \det(\sum_{l=1}^n X_{ll} u_l u_l^T)^{1/d}}{\det(X)^{1/n}} \leq \capa(\U) + \delta {\rm~for~some~} \delta>0.
\]
Then, in the proof of Proposition~\ref{p:reduction-capacity}, we consider the eigen-decomposition of $\ol{X}$, and we can thus assume that the orthonormal basis is simply the standard basis, i.e. $f_j = e_j$ for $1 \leq j \leq n$.
The $(i,j)$-th entry of the associated matrix $A \in \R^{d \times n}$ is defined in~(\ref{e:A}) as
\[
A_{ij} = g_i^T \big( \sum_{l=1}^n U_l f_j f_j^T U_l^T \big) g_i
\]
where $g_1, \ldots, g_d$ is an orthonormal basis. 
Using the special structure of $U_l$ and $f_j=e_j$ for $1 \leq j \leq n$,
the $(i,j)$-th entry of the associated matrix thus becomes
\[
A_{ij} = g_i^T \big( \sum_{l=1}^n U_l e_j e_j^T U_l^T \big) g_i 
= g_i^T \big( U_j e_j e_j^T U_j \big) g_i = \inner{g_i}{u_j}^2.
\]
The rest of the proof is the same as the proof of Proposition~\ref{p:reduction-capacity}, and we have $\capa(A) \leq \capa(\U) = \capa(U)$, $\Delta(A) \leq \Delta(\U) = \Delta(U)$ and $s(A) = s(\U) = s(U)$.
Finally, if $\norm{u_j}^2 = d/n$ for $1 \leq j \leq n$, then notice that $\sum_{j=1}^n U_j^T U_j = \frac{d}{n} I_n$.
It follows from the furthermore part in Proposition~\ref{p:reduction-capacity} that $c_j(A) = d/n$ for $1 \leq j \leq n$. 
\end{proof}

\subsubsection*{Proving the Pseudorandom Property}

The main result in this subsection is the following proposition.

\begin{proposition} \label{p:perturb-pseudorandom}
Let $W = \{w_1, \ldots, w_n\}$ where $w_j \in \R^d$ for $1 \leq j \leq n$ be the output of the perturbation process in Procedure~\ref{proc:perturb} when $U=\{u_1, \ldots, u_n\}$ where $u_j \in \R^d$ for $1 \leq j \leq n$ is given as the input with $\norm{u_j}^2 = d/n$ for $1 \leq j \leq n$.
Let $B$ be a $d \times n$ matrix with 
\[B_{ij} = \inner{g_i}{w_j}^2 {\rm~for~} 1 \leq i \leq d {\rm~and~} 1 \leq j \leq n \quad {\rm where~} g_1, \ldots, g_d \in \R^d {\rm~form~an~orthonormal~basis}.
\]
Assume that 
\[\norm{u_j}^2 = \frac{d}{n} {\rm~for~} 1 \leq j \leq n
\quad {\rm and} \quad 
\sigma^2 \leq \frac{1}{1600n}
\quad {\rm and} \quad
n \geq \frac{1600 d^2}{\beta \zeta},
\]
where $\zeta = (\beta/(800e))^{160}$ and $\beta = 10^{-9}$ are absolute constants.
Then $B$ is $(\zeta \sigma^2, \beta)$-pseudorandom with probability at least $0.9$.
\end{proposition}
\begin{proof}
Recall from Definition~\ref{d:pseudorandom} for a matrix to be $(\alpha,\beta)$-pseudorandom.
The first property is easy to check.
The $j$-th column sum of $B$ is
\[c_j(B) = \sum_{i=1}^d \inner{g_i}{c_j}^2 = \norm{c_j}^2 = \frac{d}{n} \geq d \sigma^2,
\]
where the second equality is because $g_1, \ldots, g_d$ is an orthonormal basis, the third equality is by the postprocessing step in the perturbation process, and the inequality is by our assumption that $\sigma^2 \leq 1/n$.
Therefore, there must be an entry with value at least $\sigma^2$, satisfying the first property in Definition~\ref{d:pseudorandom}.

For the second property in Definition~\ref{d:pseudorandom},
we will prove that with probability at least $0.99$ over the perturbation process,
for every row $1 \leq i \leq d$, there are at most $\beta n$ entries with $\inner{g_i}{w_j}^2 < \zeta \sigma^2$.
We will establish this by proving a stronger statement that
with probability at least $0.99$ over the perturbation process,
for any unit vector $g \in \R^d$, there are at most $\beta n$ entries with $\inner{g}{w_j}^2 < \zeta \sigma^2$.

To do this, we follow the perturbation process and keep track of the pseudorandom property.
Recall that in the perturbation process, we first generate some correlated Gaussian noise $z_1, \ldots, z_n \in \R^d$ and set $v_j = u_j + z_j$ for $1 \leq j \leq n$, and then we rescale each $v_j$ to form $w_i$. 
In Lemma~\ref{l:perturb-v}, we will prove that with probability at least $0.99$ over the perturbation process, for any unit vector $g \in \R^d$, there are at most $\beta n / 2$ entries with $\inner{g}{v_j}^2 < 2\zeta \sigma^2$.

Assuming Lemma~\ref{l:perturb-v}, we then consider the rescaling step.
Since $w_j = \sqrt{\frac{d}{n}} \frac{v_j}{\norm{v_j}}$ for $1 \leq j \leq n$,
we have 
\[\inner{g}{w_j}^2 = \frac{d}{n} \norm{v_j}^{-2} \inner{g}{v_j}^2 \quad {\rm for~} 1 \leq j \leq n.\]
If we can prove that $\norm{v_j}^2 \leq 2d/n$ for at most $\beta n / 2$ entries with probability at least $0.99$, then this implies the lemma that there are at most $\beta n$ entries with $\inner{g}{w_j}^2 < \zeta \sigma^2$ with probability at least $0.9$.
Note that
\[
\norm{v_j}^2 = \norm{u_j + z_j}^2 = \norm{u_j}^2 + 2\inner{u_j}{z_j} + \norm{z_j}^2 = \norm{u_j}^2 + \norm{z_j}^2,
\]
where the last equality is by the first subspace in the perturbation process that $\inner{u_j}{z_j}=0$ for all $1 \leq j \leq n$.
Note that $\norm{z_j}^2 \leq \norm{x_j}^2$ where $x_j \in \R^d$ with each entry sampled from $N(0,\sigma^2)$, as $z_j$ is a projection of $x_j$ in the perturbation process.  
By our assumption, $\E \norm{x_j}^2 = d\sigma^2 \leq d/ (100n)$.
By Fact~\ref{f:concentration}(2), $\P(\norm{x_j}^2 \geq d/n) \leq e^{-100}$,
and therefore the expected number of entries with $\norm{x_j}^2 \leq d/n$ is at most $e^{-100} n \ll \beta n / 2$.
By Chernoff bound, there are at most $\beta n / 2$ entries with $\norm{x_j}^2 > d/n$ with probability at least $0.99$,
and therefore $\norm{v_j}^2 \leq 2d/n$ for at most $\beta n / 2$ entries with probability at least $0.99$.
\end{proof}

To finish the proof of Proposition~\ref{p:perturb-pseudorandom},
it remains to prove Lemma~\ref{l:perturb-v} which is the main work in this subsection.
To prove Lemma~\ref{l:perturb-v},
we first consider $\v_j = u_j + y_j$ and prove in Lemma~\ref{l:perturb-L1} that for any unit vector $g \in \R^d$, there are a small number of entries with small $\inner{g}{\v_j}^2$, and then we use it to prove Lemma~\ref{l:perturb-v}.
The idea is to use the special property of the subspace $L_1$ in the perturbation process to prove Lemma~\ref{l:perturb-L1}, and then we use that $L_2$ is a subspace with codimension $d^2$ to argue that there are at most $O(d^2)$ more small entries in $\inner{g}{v_j}^2$.

The core of the technical argument in this subsection is in the following lemma.

\begin{lemma} \label{l:perturb-L1}
Assume $\sigma^2 \leq 1/(1600n)$ and $n \geq 4(d \log d)/(\beta \zeta)$ and $\zeta \leq (\beta/(800e))^{160}$.
The probability over the perturbation process 
that there exists a unit vector $g \in \R^d$ 
and a subset $S \subseteq [n]$ with $|S| \geq \beta n / 4$ such that
\[
\inner{g}{u_j+y_j}^2 \leq 4\zeta \sigma^2 \quad {\rm for~all~} j \in S
\]
is at most $0.001$.
\end{lemma}
\begin{proof}
Our plan is to use the concentration inequality in Fact~\ref{f:concentration}(3) to prove that the probability is small for a given $g \in \R^d$ and a given large subset $S$,
and then we apply a union bound on an epsilon-net in $\R^d$ and all large subsets $S$.

We first assume that $\norm{y_j}^2 \geq 100 d \sigma^2$ for at most $\beta n / 8$ vectors.
The reason is as follows.
Note that $\norm{y_j}^2 \leq \norm{x_j}^2$ where $x_j \in \R^d$ with each entry independently sampled from $N(0,\sigma^2)$, as $y_j$ is a projection of $x_j$ in the perturbation process.  
By definition, $\E \norm{x_j}^2 = d\sigma^2$.
By Fact~\ref{f:concentration}(2), $\P(\norm{x_j}^2 \geq 100d\sigma^2) \leq e^{-100}$,
and therefore the expected number of vectors with $\norm{x_j}^2 \geq 100d\sigma^2$ is at most $e^{-100} n \ll \beta n / 8$.
By Chernoff bound, there are at most $\beta n / 8$ vectors with $\norm{x_j}^2 \geq 100d \sigma^2$ with probability at least $0.9999$, and hence there are at most $\beta n / 8$ vectors with $\norm{y_j}^2 \geq 100d\sigma^2$ with probability at least $0.9999$.
In the following, we condition on this event.

For an entry with $|\inner{g}{u_j}| \geq \norm{u_j}/2$ and $\norm{y_i} \leq 10\sqrt{d} \sigma$,
we have 
\[
|\inner{g}{u_j + y_j}| 
\geq |\inner{g}{u_j}| - |\inner{g}{y_j}|
\geq \frac12 \norm{u_j} - \norm{y_j}
\geq \frac12 \norm{u_j} - 10\sqrt{d} \sigma
\geq 10\sqrt{d} \sigma 
\gg 2\sqrt{\zeta} \sigma,
\]
where the second last inequality is by our assumptions that 
$\norm{u_j}^2=d/n \geq 1600d\sigma^2$.
So, for an entry with $|\inner{g}{u_j}| \geq \norm{u_j}/2$,
it will be too small after the perturbation process only if $\norm{y_i} \geq 10\sqrt{d} \sigma$.
Therefore, by the above assumption, there are at most $\beta n / 8$ entries with $|\inner{g}{u_j}| \geq \norm{u_j}/2$ that could become too small after the perturbation process.

Hence, we can restrict our attention to those entries with $|\inner{g}{u_j}| \leq \norm{u_j}/2$.
The following claim contains the key concentration argument.

\begin{claim} \label{c:one-subset}
Assume that $n \geq 2d / \beta$.
For a unit vector $g$ and a subset $S \subseteq [n]$ with $|S| \geq \beta n$ with $|\inner{g}{u_j}| \leq 3\norm{u_j}/4$ for all $j \in S$,
\[
\P_y[ \inner{g}{u_j + y_j}^2 \leq  16\zeta \sigma^2 {\rm~for~all~} j \in S] \leq (100 \zeta)^{|S|/20}.
\]
\end{claim}
\begin{proof}
We will reduce our setting to applying Fact~\ref{f:concentration}(3).
Assume without loss of generality that $\inner{g}{u_j+y_j}^2 \leq 16 \zeta \sigma^2$ for $1 \leq j \leq |S|$.
Let $\gamma \in \R^{|S|}$ be the vector with the $j$-th entry being $\gamma_j = \inner{g}{u_j}$.
Let $y^g \in \R^{|S|}$ be the vector with the $j$-th entry being $y^g_j = \inner{g}{y_j}$.
Let $h \in \R^{|S|}$ be the vector with the $j$-th entry being
\[
h_j = \inner{g}{u_j + y_j} = \inner{g}{u_j} + \inner{g}{y_j} = \gamma_j + y^g_j.
\]
Note that $\gamma$ is in a fixed $d$-dimensional subspace $L \subseteq \R^{|S|}$, where $\gamma = U_S^T g$ where $U_S$ is the $d \times |S|$ matrix with the $j$-th column being $u_j$.
Let $L^{\perp}$ being the orthogonal subspace of $L$ in $\R^{|S|}$, so the dimension of $L^{\perp}$ is at least $|S|-d \geq |S|/2$ by our assumptions $|S| \geq \beta n$ and $n \geq 2d/\beta$.
Let $P_{L^{\perp}}$ be the projector matrix to the subspace $L^{\perp}$.
Our assumption is equivalent to 
\[\norm{h}_{\infty} \leq 4\sqrt{\zeta} \sigma,\]
which implies that 
\[
\norm{P_{L^{\perp}} y^g}_2 
= \norm{P_{L^{\perp}} h}_2 
\leq \norm{h}_2
\leq \sqrt{|S|} \cdot \norm{h}_{\infty}
= 4\sqrt{\zeta |S|} \sigma. 
\]

The intuition is that a random Gaussian vector $y^g$ will have a large projection to a high dimensional subspace with high probability,
and we will make it precise using Fact~\ref{f:concentration}(3).

To apply Fact~\ref{f:concentration}(3), the Gaussian vector should have independent entries, and so we write $y^g$ in terms of the independent Gaussian vector $x$ in the perturbation process.
In the following, think of a $(d \times |S|)$-dimensional vector as having $|S|$ blocks of $d$ entries.
Let $y_S$ be the $(d \times |S|)$-dimensional vector with $y_j \in \R^d$ in the $j$-th block for $1 \leq j \leq |S|$, where $y_j$ is generated in the perturbation process.
Let $G$ be the $|S| \times (d \times |S|)$ matrix with the $i$-th row being the $(d \times |S|)$-dimensional vector with $g \in \R^d$ in the $i$-th block while other blocks are zero.
Then we can write 
\[y^g = G y_S.\]
Let $x_S$ be the $(d \times |S|)$-dimensional vector with $x_j \in \R^d$ in the $j$-th block for $1 \leq j \leq |S|$, where $x_j$ is generated in the perturbation process.
Recall that $y_S$ is the projection of $x_S$ into the first subspace $L_1$ such that $\inner{y_j}{u_j}=0$.
The projection matrix can be written explicitly as a $(d \times |S|) \times (d \times |S|)$ matrix $P$ where the $(j,j)$-th diagonal block of $P$ being
\[I_d - \frac{u_j u_j^T}{\norm{u_j}^2}\]
and all other blocks being zero. 
Then we can write
\[y_S = P x_S \quad \implies \quad
y^g = GPx_S,
\]
where $GP$ is the $|S| \times (d \times |S|)$ matrix with the $i$-th row being the $(d \times |S|)$-dimensional vector with the $i$-th block being
\[
g^T (I - \frac{u_iu_i^T}{\norm{u_i}^2}).
\]
So, we can finally write 
\[
\norm{P_{L^{\perp}} y^g}_2 = \norm{P_{L^{\perp}} GP x_S}_2,
\]
where $x_S$ is a vector with independent Gaussian variable with mean zero and variance $\sigma^2$.

Now, we can apply Fact~\ref{f:concentration}(3) to bound the probability that the right hand side is small.
To do so, we need to compute 
\[\tr((GP)^T P_{L^{\perp}} (GP))
= \tr(P_{L^{\perp}} (GP) (GP)^T).
\]
Note that $D:=(GP) (GP)^T$ is an $|S| \times |S|$ diagonal matrix with 
\[D_{ii} 
= g^T(I-\frac{u_iu_i^T}{\norm{u_i}^2})(I-\frac{u_iu_i^T}{\norm{u_i}^2})g 
= \norm{g}^2 - \frac{\inner{g}{u_i}^2}{\norm{u_i}^2} 
\geq \norm{g_i}^2 - \frac{9}{16} = \frac{7}{16},
\]
where the inequality is by our assumption that $|\inner{g}{u_i}| \leq 3\norm{u_i}/4$ for all $i \in S$.
Therefore,
\[
\tr((GP)^T P_{L^{\perp}} (GP)) 
= \tr(P_{L^{\perp}} D) 
\geq \frac{7}{16} \tr(P_{L^{\perp}})
= \frac{7}{16} \rank(P_{L^{\perp}})
\geq \frac{7}{16} (|S|-d)
\geq \frac{7}{32} |S|,
\]
where the last equality is by Fact~\ref{f:projection}.
Since $0 \preceq (GP)^T P_{L^{\perp}} (GP) \preceq I$, 
we can apply Fact~\ref{f:concentration}(3) and get that the probability of the event that
\[
\norm{P_{L^{\perp}} GP x_S}^2_2 = \norm{P_{L^{\perp}} y^g}^2_2 \leq 16\zeta |S| \sigma^2
\]
is at most
\[
\frac{16 \zeta |S|}{\tr\big((GP)^T P_{L^{\perp}} (GP)\big)}^{\frac14 \tr\big((GP)^T P_{L^{\perp}} (GP)\big)}
\leq \big(\frac{512}{7} \zeta\big)^{\frac{7}{128} |S|}.
\]
We conclude that 
\[
\P_y\Big[\inner{g}{u_j + y_j}^2 \leq 16\zeta \sigma^2 {\rm~for~all~} j \in S\Big] 
= \P_y\Big[ \norm{h}^2_{\infty} \leq 16\zeta \sigma^2\Big]
\leq \P_y\Big[ \norm{P_{L^{\perp}} y^g}^2_2  \leq 16\zeta |S| \sigma^2\Big]
\leq (100 \zeta)^{|S|/20}.
\]
\end{proof}

With the claim,
we apply a standard union bound on an epsilon-net in the $d$-dimensional space and all subsets of size at least $\beta n /8$ to prove the following.

\begin{claim} \label{c:epsilon-net}
Assume $\sigma^2 \leq 1/(1600n)$ and $n \geq 4(d\log d)/(\beta \zeta)$ and $\zeta \leq (\beta/(800e))^{160}$.
After generating $y_1, \ldots, y_n \in \R^d$ as described in the perturbation process,
the probability that there exists a unit vector $g \in \R^d$ and a subset $S \subseteq [n]$ with 
\[|S| \geq \beta n /8, 
\quad |\inner{g}{u_j}| \leq \norm{u_j}/2 {\rm~for~all~} j \in S,
\quad \norm{y_i}^2 \leq 100d\sigma^2 {\rm~for~all~} j \in S
\]
such that
\[
\inner{g}{u_j + y_j}^2 \leq 4 \zeta \sigma^2 {\rm~for~all~} j \in S
\]
is at most $\exp(-\beta n / 4)$.
\end{claim}
\begin{proof}
We use the same notation as defined in Claim~\ref{c:one-subset},
and use that
\[
\P_y\Big[ \inner{g}{u_j + y_j}^2 \leq 4 \zeta \sigma^2 {\rm~for~all~} j \in S\Big]
\leq \P_y\Big[ \norm{P_{L^{\perp}} y^g}^2_2  \leq 4 \zeta |S| \sigma^2\Big] 
\]
To apply an epsilon-net argument, we check how quickly the quantity $\norm{P_{L^{\perp}} y^g}^2$ changes under some small change of $g$.
Let $\tilde{g} = g + \delta$ where $\delta \in \R^d$.
Recall that the $i$-th entry of $y^g = \inner{g}{y_i}$,
and so the $i$-th entry of $y^{\tilde{g}} = \inner{g+\delta}{y_i} = \inner{g}{y_i} + \inner{\delta}{y_i}$,
and so
\begin{align*}
\norm{P_{L^{\perp}}y^{\tilde{g}}}_2 - \norm{P_{L^{\perp}} y^g}_2
& = \norm{P_{L^{\perp}}y^{g} + P_{L^{\perp}}y^{\delta}}_2 - \norm{P_{L^{\perp}} y^g}_2
\leq \norm{P_{L^{\perp}}y^{\delta}}_2
\\
& \leq \norm{y^{\delta}}_2 
= \sqrt{ \sum_{i=1}^{|S|} \inner{\delta}{y_i}^2 }
\leq \sqrt{ \sum_{i=1}^{|S|} \norm{\delta}^2 \norm{y_i}^2 }
\leq 10\sqrt{d|S|} \sigma \norm{\delta},
\end{align*}
and also
\[
\Big| |\inner{g+\delta}{u_i}| - |\inner{g}{u_i}|\Big|
\leq |\inner{\delta}{u_i}| \leq \norm{\delta} \norm{u_i}.
\]
Let $N$ is an $\frac{1}{10} \sqrt{\frac{\zeta}{d}}$-net of the unit sphere in $\R^d$.
Then, if there exists any unit vector $g \in \R^d$ such that
\[
\norm{P_{L^{\perp}} y^g}_2  \leq 2\sqrt{\zeta |S|} \sigma \quad {\rm and} \quad
|\inner{g}{u_i}| \leq \frac{1}{2} \norm{u_i},
\]
then there exists a unit vector $\tilde{g}$ in $N$ such that
\[
\norm{P_{L^{\perp}} y^{\tilde{g}}}_2  \leq 4\sqrt{\zeta |S|} \sigma \quad {\rm and} \quad
|\inner{\tilde{g}}{u_i}| \leq \frac{3}{4} \norm{u_i}.
\]
It is known that (e.g. see~\cite{V12}) that the net $N$ has size
\[
|N| \leq \Big(1+ \frac{2}{\frac{1}{10} \sqrt{\frac{\zeta}{d}}}\Big)^d
\]
By a union bound on $N$ and $S$ with Claim~\ref{c:one-subset}, 
the probability that we want to bound in this claim is at most
\begin{align*}
\sum_{|S|=\beta n/8}~\sum_{\tilde{g} \in N} \P_y\Big[ \norm{P_{L^{\perp}} y^{\tilde{g}}}^2_2  \leq 16\zeta |S| \sigma^2\Big]
& \leq \binom{n}{\beta n / 8} \Big(1+ \frac{2}{\frac{1}{10} \sqrt{\frac{\zeta}{d}}}\Big)^d
(100\zeta)^{\beta n/160}
\\
& \leq \Big( \frac{ne}{\beta n / 8} \Big)^{\beta n /2} 
\exp\Big( d \log\big(1 + 20\sqrt{\frac{d}{\zeta}}\big) + \frac{\beta n}{160} \log(100 \zeta)\Big)
\\
& = \exp\Big( \frac{\beta n}{2} \log(\frac{8e}{\beta}) + d \log\big(1 + 20\sqrt{\frac{d}{\zeta}}\big) - \frac{\beta n}{160} \log(\frac{1}{100 \zeta})\Big)
\end{align*}
If we set $\zeta \leq (\frac{\beta}{800e})^{160}$, then 
\[
\sum_{|S|=\beta n/8} \sum_{\tilde{g} \in N} \P_y\Big[ \norm{P_{L^{\perp}} y^g}^2_2  \leq 16\zeta |S| \sigma^2\Big]
\leq \exp\Big( d \log\big(1 + 20\sqrt{\frac{d}{\zeta}}\big) - \frac{\beta n}{2} \log(\frac{8e}{\beta})\Big)
\leq \exp(-\beta n/ 4),
\]
where the final inequality uses that $n \geq 4 (d \log d) / (\beta \zeta)$.
\end{proof}

Continuing the proof of Lemma~\ref{l:perturb-L1},
with probability at least $0.9999$, we have $\norm{y_j}^2 \geq 100d\sigma^2$ for at most $\beta n / 8$ vectors.
We condition on this event and assume that in the worst case that all the entries in those columns become very small.
Hence, for each unit vector $g \in \R^d$, we restrict our attention to those columns with $\norm{y_j}^2 \leq 100d\sigma^2$.
As explained before Claim~\ref{c:one-subset}, for the entries with $|\inner{g}{u_j}| \geq \norm{u_j}/2$, we always have $\inner{g}{u_j+y_j}^2 \gg 4\zeta \sigma^2$ if $\norm{y_j}^2 \leq 100 d \sigma^2$. 
So, we further restrict our attention to those entries with $|\inner{g}{u_j}| \leq \norm{u_j}/2$.
Now, we can apply Claim~\ref{c:epsilon-net} to prove that with probability at least $1-\exp(-\beta n / 4)$, every unit vector $g \in \R^d$ has at most $\beta n / 8$ such entries with $\inner{g}{u_j+y_j}^2 \leq 4 \zeta \sigma^2$.
Therefore, with probability at least $0.9999 - \exp(\beta n / 4) \geq 0.999$, every unit vector $g \in \R^d$ has at most $\beta n / 4$ entries with $\inner{g}{u_j+y_j}^2 \leq 4\zeta \sigma^2$.
\end{proof}

In Lemma~\ref{l:perturb-L1}, we have proved that with probability at least $0.999$, every unit vector $g \in \R^d$ has at most $\beta n / 4$ entries with 
$\inner{g}{u_j + y_j}^2 \leq 4\zeta \sigma^2$.
To prove Proposition~\ref{p:perturb-pseudorandom},
we need to prove that every unit vector $g \in \R^d$ has at most $\beta n / 2$ entries with $\inner{g}{u_j + z_j}^2 \leq 2\zeta \sigma^2$.
To do this, we use that $z = P_{L_1 \cap L_2} y$ and $L_2$ is a subspace with codimension $d^2$.
With the assumption $n \gg d^2$, our plan is to show that $\inner{g}{u_j+z_j}^2 \approx \inner{g}{u_j+y_j}^2$ for most entries.

\begin{lemma} \label{l:perturb-v}
Assume that $n \geq 1600 d^2 / (\beta \zeta)$.
The probability over the perturbation process 
that there exists a unit vector $g \in \R^d$ 
and a subset $S \subseteq [n]$ with $|S| \geq \beta n / 2$ such that
\[
\inner{g}{u_j+z_j}^2 \leq 2 \zeta \sigma^2 \quad {\rm for~all~} j \in S
\]
is at most $0.01$.
\end{lemma}
\begin{proof}
First, we bound the difference between $\inner{g}{u_j+y_j}^2$ and $\inner{g}{u_j+z_j}^2$.
Recall that $y \in \R^{d \times n}$ is the concatenation of $y_1, \ldots, y_n \in \R^d$, and $z \in \R^{d \times n}$ is the concatenation of $z_1, \ldots, z_n \in \R^d$.
Since $y = P_{L_1} x$ and $z = P_{L_1 \cap L_2} x$ as described in the perturbation process,
\[
\norm{y-z}^2 
= \norm{(P_{L_1} - P_{L_1 \cap L_2})x}^2
= x^T (P_{L_1} - P_{L_1 \cap L_2}) (P_{L_1} - P_{L_1 \cap L_2}) x
= x^T (P_{L_1} - P_{L_1 \cap L_2}) x.
\]
Note that $P_{L_{1}}-P_{L_{1}\cap L_{2}}=P_{L_{1}\cap L_{2}^{\perp}}$.
Since $\text{rank}(L_{1}\cap L_{2}^{\perp})\leq\text{rank}(L_{2}^{\perp})\leq d^{2}$,
using Fact~\ref{f:concentration}(2) with $A = P_{L_1} - P_{L_1 \cap L_2}$ shows that 
\[
\P_x(\norm{y-z}^{2} \geq 100 d^{2}\cdot \sigma^2)\leq e^{-10 d^{2}}.
\]
Therefore, with probability $1-e^{-10d^2}$, we have that
\begin{equation*}
\sum_{j=1}^n \inner{g}{y_j-z_j}^{2} \leq 
\sum_{j=1}^n \norm{g}^2 \norm{y_j-z_j}^2 
= \norm{y-z}^2 \leq 100d^{2} \sigma^2 
\label{eq:noise_difference}
\end{equation*}
for all unit vector $g \in \R^d$.
This implies that there are at most $400^2/\zeta$ terms with 
$\inner{g}{y_j-z_j}^2 \geq \frac{1}{4} \zeta \sigma^2$.
Lemma~\ref{l:perturb-L1} shows that for all unit vector $g \in \R^d$,
there are at most $\beta n / 4$ entries with $\inner{g}{u_j+y_j}^2 \leq 4\zeta \sigma^2$ with probability $0.999$.
For an entry with $\inner{g}{u_j + y_j}^2 > 4\zeta \sigma^2$ and $\inner{g}{y_j-z_j}^2 \leq \frac{1}{4} \zeta \sigma^2$,
we have
\[
\inner{g}{u_j+z_j}^2 \geq \inner{g}{u_j+y_j}^2 - 2|\inner{g}{u_j+y_j}| |\inner{g}{z_j-y_j}| \geq \frac{1}{2} \inner{g}{u_j+y_j}^2 \geq 2\zeta \sigma^2.
\]
Therefore, besides the $\beta n / 4$ entries with $\inner{g}{u_j+y_j}^2 \leq 4\zeta \sigma^2$, there could be at most $400d^2/\zeta$ more entries with $\inner{g}{u_j+z_j}^2 \leq 2 \zeta \sigma^2$ with probability $1 - e^{-10d^2}$.
By our assumption that $400d^2/\zeta \leq \beta n /4$,
for all unit vectors $g \in \R^d$,
there are at most $\beta n / 2$ entries with $\inner{g}{u_j + z_j}^2 \leq 2\zeta \sigma^2$ with probability at least $0.999 - e^{-10d^2} \geq 0.99$.
\end{proof}

This completes the proof of Proposition~\ref{p:perturb-pseudorandom}.

In the next subsection, with the right choice of the parameters,
we can prove point (iii) in Subsection~\ref{ss:overview} about the improved capacity lower bound.

\subsection{The Path to an Equal Norm Parseval Frame} \label{ss:together}

In this subsection, we put together the results in this section to prove Theorem~\ref{t:main}.

We describe precisely how to move from the initial $\eps$-nearly equal norm Parseval frame to the final equal norm Parseval frame.
We also state the properties that we will maintain, and will prove these properties later.

\begin{procedure}[path] \label{proc:path}
The following is a step-by-step procedure, with the properties required stated.
\begin{enumerate}

\item {\bf Input}: an $\eps$-nearly equal norm Parseval frame $U=\{u_1, \ldots, u_n\}$ where $u_i \in \R^d$ for $1 \leq i \leq n$.

\item \label{step:preprocess}
Apply Lemma~\ref{l:preprocess} to scale $U$ to $U^0$ such that 
\[U^0 \text{ is an } O(\eps)\text{-nearly equal norm Parseval frame, } s(U^0)=d \text{ and } \dist(U,U^0)=O(d \eps^2).
\]

\item \label{step:initial-Delta}
Let $\Delta := \Delta(U^0)$.  Note that $\Delta \leq O(d^2 \eps^2)$ by Lemma~\ref{l:Delta-eps}.

\item \label{step:global}
{\bf Global assumptions}:
\[
n \geq \frac{10^{15} d^{4}}{\zeta^2 \kappa^2} \quad {\rm and} \quad \Delta \leq \frac{\zeta^6 \kappa^6}{10^{55} d^{9}},
\]
where $0 < \kappa < 1$ is the small constant in Proposition~\ref{p:maintain} and $0 < \zeta < 1$ is the small constant in Proposition~\ref{p:perturb-pseudorandom}.

\item \label{step:loop}
  Let $l:=0$.  Repeat the following steps.
  \begin{enumerate}
  \item \label{step:perturb}
  Apply the perturbation process to $U^l = \{u_1^l, \ldots, u_n^l\}$ to produce 
  $W^l = \{w_1^l, \ldots, w_n^l\}$ with
  \[
  \sigma^2 = \frac{10^4 \sqrt{d\Delta(U^l)}}{\zeta \kappa n}. 
  \]
  The assumptions are that 
  \[s(U^l)=d \quad {\rm and} \quad \Delta(U^l) \leq \Delta {\rm~so~that~} \sigma^2 \ll \frac{1}{n}.\]
  The properties of $W^l$ are that
    \begin{enumerate}[(i)]
    \item \label{step:size-norm}
    (Size and equal norm) The following is guaranteed by the perturbation 
    process in Procedure~\ref{proc:perturb}:
    \[s(W^l) = d \quad {\rm and} \quad 
    \norm{w_i^l}_2^2 = \frac{d}{n}  {\rm~for~} 1 \leq i \leq n.
    \]
    \item \label{step:movement} 
    (Movement) We will prove in Proposition~\ref{p:movement} that
    \[
    \dist(U^l,W^l) = O\big(d^{3/2} \sqrt{\Delta(U^l)}\big).
    \]
    \item \label{step:Delta}
    ($\Delta$) We will prove in Proposition~\ref{p:Delta} that
    \[
    \Delta(W^l) = O(\Delta(U^l)).
    \]
    \item \label{step:capacity}
    (Capacity) We will prove in Theorem~\ref{t:capacity} that
    \[
    \capa(W^l) \geq s(W^l) - O\Big(\sqrt{\frac{\Delta(W^l)}{d}}\Big).
    \]
    \end{enumerate}
  
  \item \label{step:dynamical}
  Let $W^{(0)} := W^l$.
  Apply the dynamical system in Definition~\ref{d:dynamical} 
  with $W^{(0)}$ as the input until the first time $T$ when
  \[\Delta(W^{(T)}) \leq \frac{\Delta}{3\cdot2^l}.\]
  The properties of $W^{(T)}$ are that
    \begin{enumerate}[(i)]
    \item \label{step:movement2}
    (Movement) We will prove in Proposition~\ref{p:movement-size} that
    \[
    \dist(W^{(0)},W^{(T)}) \leq O\big(\sqrt{\frac{\Delta(U^l)}{d}}\big).
    \]
    \item \label{step:size}
    (Size) We will prove in Proposition~\ref{p:movement-size} that
    \[
    s(W^{(T)}) \geq d - O\big(\sqrt{\frac{\Delta(U^l)}{d}}\big).
    \]
    \end{enumerate}

  \item \label{step:rescale} Rescale $W^{(T)}$ to $U^{l+1}$ such that
  \[
  s(U^{l+1}) = d.
  \]
  We will prove in Lemma~\ref{l:rescale} that 
  \[\dist(W^{(T)},U^{l+1}) \leq O\big(\sqrt{\frac{\Delta(U^l)}{d}}\big) 
  \quad {\rm and} \quad 
  \Delta(U^{l+1}) \leq \frac{\Delta}{2^{l+1}}.
  \]
  Increment $l$ by one.
  \end{enumerate} 
\item {\bf Output:} an equal norm Parseval frame $V:=U^{\infty}$.
\end{enumerate}
\end{procedure}

We will first prove the main theorems assuming all the claims in Procedure~\ref{proc:perturb}.  Then we will justify all the claims in Procedure~\ref{proc:path}.

\begin{theorem} \label{t:one-iteration}
Assume the global assumptions in Step~\ref{step:global} of Procedure~\ref{proc:path} hold.
Then, in the $l$-th iteration of Step~\ref{step:loop} in Procedure~\ref{proc:path}, 
\[
\dist(U^l,U^{l+1}) \leq O(d^{3/2} \sqrt{\Delta(U^l)}) = O(d^{3/2} \sqrt{\Delta/2^l}).
\]
\end{theorem}
\begin{proof}
We analyze the total movement in the $l$-th iteration of Step~\ref{step:loop}.
We first apply the perturbation process to $U^l$ to obtain $W^l$ in Step~\ref{step:perturb} and by Proposition~\ref{p:movement}
\[
\dist(U^l,W^l) = O(d^{3/2} \sqrt{\Delta(U^l)}).
\]
Then we apply the dynamical system to $W^l$ to obtain $W^{(T)}$ in Step~\ref{step:dynamical} and by Proposition~\ref{p:movement-size}
\[
\dist(W^l,W^{(T)}) = O(\sqrt{\frac{\Delta(U^l)}{d}}).
\] 
Finally, we rescale $W^{(T)}$ to obtain $U^{l+1}$ in Step~\ref{step:rescale} and by Lemma~\ref{l:rescale}
\[
\dist(W^{(T)},U^{l+1}) = O(\sqrt{\frac{\Delta(U^l)}{d}}).
\]
By the triangle inequality, the total movement in one iteration is
\begin{eqnarray*}
\distance(U^l,U^{l+1})
& \leq & 
\distance(U^l,W^l) + \distance(W^l,W^{(T)}) + \distance(W^{(T)},U^{l+1})
\\
& \leq & O\Big(\sqrt{d^{3/2} \sqrt{\Delta(U^l)}}\Big) + O\Big(\sqrt{\sqrt{\frac{\Delta(U^l)}{d}}}\Big) + O\Big(\sqrt{\sqrt{\frac{\Delta(U^l)}{d}}}\Big) \\
& \leq & O\Big(\sqrt{d^{3/2} \sqrt{\Delta(U^l)}}\Big).
\end{eqnarray*}
Squaring both sides proves the result.
\end{proof}

The following theorem implies Theorem~\ref{t:perturb}.

\begin{theorem} \label{t:d^3 eps}
Assume the global assumptions in Step~\ref{step:global} of Procedure~\ref{proc:path} hold.
Given any $\eps$-nearly equal norm Parseval frame $U=\{u_1,\ldots,u_n\}$ where $u_i \in \R^d$ for $1 \leq i \leq n$, there exists an equal norm Parseval frame $V$ with
\[
\dist(U,V) = O(d^{3/2} \sqrt{\Delta}) = O(d^{5/2} \eps).
\]
\end{theorem}
\begin{proof}
It is clear that the $U^l$ is converging to an equal norm Parseval frame,
as $\Delta(U^l) \leq \Delta(U^0)/2^l$, and $\Delta(U^{l})=0$ if and only if $U^l$ is an equal norm Parseval frame.

Using Theorem~\ref{t:one-iteration} and the triangle inequality,
the total movement in all iterations is
\begin{eqnarray*}
\distance(U^0,U^{\infty})
& \leq &
\sum_{l=0}^{\infty} \distance(U^l,U^{l+1}) 
\leq \sum_{l=0}^{\infty} O\Big(\sqrt{d^{3/2} \sqrt{\Delta(U^l)}}\Big)
\\
& \leq & \sum_{l=0}^{\infty} O\Big(\sqrt{d^{3/2} \sqrt{\Delta/2^l}}\Big)
\leq O\Big(\sqrt{d^{3/2} \sqrt{\Delta}}\Big),
\end{eqnarray*}
where the last inequality is a decreasing geometric sum.
From the preprocessing in Step~\ref{step:preprocess},
we have $\dist(U,U^0) = O(d\eps^2)$ and $\Delta \leq O(d^2\eps^2)$.
Therefore, since we set $V = U^{\infty}$, we have
\begin{eqnarray*}
\distance(U,V)
& \leq & \distance(U,U^0) + \distance(U^0,U^{\infty})\\
& \leq & O(\sqrt{d\eps^2}) + O\Big(\sqrt{d^{3/2} \sqrt{\Delta}}\Big) 
\leq O(\sqrt{d\eps^2}) + O(\sqrt{d^{5/2} \eps}) = O(\sqrt{d^{5/2} \eps}).
\end{eqnarray*}
Squaring both sides proves the theorem.
\end{proof}

Finally, we prove an unconditional bound for the Paulsen problem, which implies Theorem~\ref{t:main}.

\begin{theorem} \label{t:d^7 eps}
Given any $\eps$-nearly equal norm Parseval frame $U=\{u_1,\ldots,u_n\}$ where $u_i \in \R^d$ for $1 \leq i \leq n$, there exists an equal norm Parseval frame $V$ with
\[
\dist(U,V) = O(d^{13/2} \eps).
\]
\end{theorem}
\begin{proof}
If the assumptions in Theorem~\ref{t:d^3 eps} are satisfied, then we can find an equal norm Parseval frame $V$ with $\dist(U,V) = O(d^{5/2} \eps)$.

If the first assumption is not satisfied, then $n = O(d^4)$,
and we can apply the result in Theorem~\ref{t:Paulsen} to obtain an equal norm Parseval frame $V$ with $\dist(U,V) = O(d^2 n \eps) = O(d^6 \eps)$.

If the second assumption is not satisfied, then $\Omega(1/d^9) \leq \Delta \leq O(d^2 \eps^2)$ by Lemma~\ref{l:Delta-eps} in Step~\ref{step:initial-Delta}.
This implies that $\eps \geq \Omega(1/d^{11/2})$.
In this case, we simply output any equal norm Parseval frame $V$ and
$\dist(U,V) = O(d) = O(d^{13/2} \eps)$, where we have shown the trivial bound $O(d)$ in Subsection~\ref{ss:alternating}.

So, in all cases, the bound is at most $O(d^{13/2} \eps)$ and is independent of $n$.
\end{proof}

\subsubsection*{Justifying the Steps}

The preprocessing in Step~\ref{step:preprocess} and the initial bound on $\Delta$ in Step~\ref{step:initial-Delta} are already justified.
So we start with the assumptions in Step~\ref{step:perturb}.

\begin{lemma} \label{l:assumptions}
The assumptions in Step~\ref{step:perturb} are always satisfied.
\end{lemma}
\begin{proof}
The assumption $s(U^l)=d$ in Step~\ref{step:perturb} is satisfied initially in Step~\ref{step:preprocess} and is always maintained by Step~\ref{step:rescale} afterwards.

The assumption $\Delta(U^l) \leq \Delta$ in Step~\ref{step:perturb} is satisfied initially and is always maintained by Step~\ref{step:rescale} afterwards.
\end{proof}

The properties in Step~\ref{step:size-norm} are guaranteed by the postprocessing step in the perturbation process in Procedure~\ref{proc:perturb}. 
We next consider Step~\ref{step:movement} in Procedure~\ref{proc:path}.

\begin{proposition} \label{p:movement}
Let $U =\{u_1,\ldots,u_n\}$ with each $u_i \in \R^d$ and $s(U)=d$.
If we apply the perturbation process in Procedure~\ref{proc:perturb} with the choice of $\sigma^2$ as defined in Procedure~\ref{proc:path} to obtain output $W$, then 
\[\dist(U,W) \leq O(\sigma^2 d n) \leq O(d^{3/2} \sqrt{\Delta(U)}),
\]
with probability at least $0.9$.
\end{proposition}
\begin{proof}
The perturbation process assumes that $s(U)=d$, which is satisfied by our assumption.
In the first step of the perturbation process, the vectors in $U$ are renormalized to $\widehat{U}$ so that they have the same norm, and by Lemma~\ref{l:perturb-preprocess}
\[
\dist(U,\widehat{U}) \leq \frac{\Delta(U)}{d}.
\]
Then we apply the projection step and the postprocessing step in the perturbation process and get $W$.
Since the vectors $\widehat{U}$ are of squared length exactly $d/n$
and $\sigma^2 \leq 1/n$,
we can apply Lemma~\ref{l:perturb-movement} to prove that
\[
\E \dist(\widehat{U},W) \leq 2\sigma^2 d n = O(d^{3/2} \sqrt{\Delta(U)}).
\]
By Markov's inequality, with probability at least $0.9$,
the random variable is at most $10$ times the expected value and we still have
\[
\dist(\widehat{U},W) \leq O(d^{3/2} \sqrt{\Delta(U)}).
\]
By the triangle inequality,
\[
\distance(U,W) \leq \distance(U,\widehat{U}) + \distance(\widehat{U},W)
\leq O(\sqrt{\frac{\Delta(U)}{d}}) + O(\sqrt{d^{3/2} \sqrt{\Delta(U)}})
= O(\sqrt{d^{3/2} \sqrt{\Delta(U)}}),
\]
and squaring both sides proves the lemma.
\end{proof}

We next justify Step~\ref{step:Delta}.
This is the step that we need to use the global assumptions, so that the next step in Theorem~\ref{t:capacity} would go through.

\begin{proposition} \label{p:Delta}
Assume the global assumptions in Step~\ref{step:global} of Procedure~\ref{proc:path}.
Let $U =\{u_1,\ldots,u_n\}$ with each $u_i \in \R^d$ and $s(U)=d$ and $\Delta(U) \leq \Delta$.
If we apply the perturbation process in Procedure~\ref{proc:perturb} with the choice of $\sigma^2$ as defined in Procedure~\ref{proc:path} to obtain output $W$, then 
\[
\Delta(W) \leq 200 \Delta(U)
\]
with probability at least $1/4$.
\end{proposition}
\begin{proof}
The perturbation process assumes that $s(U)=d$, which is satisfied by our assumption.
In the first step of the perturbation process, the vectors in $U$ are renormalized to $\widehat{U}$ so that they have the same norm, and by Lemma~\ref{l:perturb-preprocess}
\[
\Delta(\widehat{U}) \leq 20\Delta(U).
\]
Then we apply the projection step and the postprocessing step in the perturbation process and get $W$.
Since the vectors $\widehat{U}$ are of squared length exactly $d/n$
and $\sigma^2 \leq 1/n$,
we can apply Proposition~\ref{p:perturb-Delta} to prove that
\begin{eqnarray*}
\E \Delta(W) & \leq &
6\Delta(\widehat{U}) + 40\sigma^4 n^2 \sqrt{\Delta(\widehat{U})}
+ 10^7 \sigma^4 d^{3} n + 10^{14} \sigma^6 d^3 n^3.
\\
& \leq & 
120 \Delta(U) + 200 \sigma^4 n^2 \sqrt{\Delta(U)}
+ 10^7 \sigma^4 d^{3} n + 10^{14} \sigma^6 d^3 n^3.
\end{eqnarray*}
We choose $\sigma^2 = 10^4 \sqrt{d\Delta(U)}/(\zeta \kappa n)$ in Procedure~\ref{proc:path}, and so the right hand side is
\begin{eqnarray*}
& = &
120\Delta(U) 
+ 200 (\frac{10^8 d \Delta(U)}{\zeta^2 \kappa^2 n^2}) n^2 \sqrt{\Delta(U)}
+ 10^7 (\frac{10^8 d \Delta(U)}{\zeta^2 \kappa^2 n^2}) d^{3} n
+ 10^{14} (\frac{10^{12} d^{3/2} \Delta(U)^{3/2}}{\zeta^3 \kappa^3 n^3}) d^3 n^3
\\
& \leq &
120\Delta(U) + \frac{10^{11} d \Delta(U)^{3/2}}{\zeta^2 \kappa^2}
+ \frac{10^{15} d^{4} \Delta(U)}{\zeta^2 \kappa^2 n} 
+ \frac{10^{26} d^{9/2} \Delta(U)^{3/2}}{\zeta^3 \kappa^3}
\\
& \leq &
120\Delta(U) 
+ \frac{10^{15} d^{4} \Delta(U)}{\zeta^2 \kappa^2 n} 
+ \frac{10^{27} d^{9/2} \Delta(U)^{3/2}}{\zeta^3 \kappa^3}
\\
& \leq &
120\Delta(U) + 
\Delta(U) + \Delta(U)
\\
& = &
122 \Delta(U),
\end{eqnarray*}
where the last inequality uses the assumption that $n \geq 10^{15} d^{4} / (\zeta^2 \kappa^2)$ 
and the assumption that $\Delta(U) \leq 10^{-55} \zeta^6 \kappa^6 /d^{9}$ so that $\sqrt{\Delta(U)} \leq 10^{-27} \zeta^3 \kappa^3 /d^{9/2}$.
By Markov's inequality, with probability at least $1/4$, we have
$\Delta(W) \leq 200\Delta(U)$.
\end{proof}

Step~\ref{step:capacity} is the heart of the smoothed analysis, where we have removed the dependency on $n$ from Theorem~\ref{t:cap_lower_bound} for a perturbed instance.

\begin{theorem} \label{t:capacity}
Let $U =\{u_1,\ldots,u_n\}$ with each $u_i \in \R^d$ and $s(U)=d$ and $\Delta(U)\leq \Delta$.
If we apply the perturbation process in Procedure~\ref{proc:perturb} with the choice of $\sigma^2$ as defined in Procedure~\ref{proc:path} to obtain output $W$, then 
\[
\capa(W) \geq s(W) - O(\sqrt{\frac{\Delta(W)}{d}}),
\]
with probability at least $0.9$ assuming $\Delta(W) \leq 200\Delta(U)$.
\end{theorem}
\begin{proof}
The perturbation process assumes that $s(U)=d$, which is satisfied by our assumption.
In the first step of the perturbation process, the vectors in $U$ are renormalized to $\widehat{U}$ so that they have the same norm.
Then we apply the projection step and the postprocessing step in the perturbation process and get $W = \{w_1, \ldots, w_n\}$ where $w_i \in \R^d$ for $1 \leq i \leq n$.
Note that $s(W)=d$ and $\norm{w_i}_2^2 = d/n$ for $1 \leq i \leq n$ by the postprocessing step in the perturbation process.
We would like to prove a lower bound on the capacity of $W$.

To prove a capacity lower bound for the frame $W$,
we apply the reduction in Proposition~\ref{p:reduction-capacity} to construct a 
$d \times n$ non-negative matrix $B$ with
\[
\capa(B) \leq \capa(W), \quad \Delta(B) \leq \Delta(W), \quad {\rm and~} s(B)=s(W)=d.
\]
Our plan is to prove that $B$ is pesudorandom and use Theorem~\ref{t:pseudorandom} to establish a lower bound on $\capa(B)$.
Since all vectors in $\widehat{U}$ have the same squared norm $d/n$ and the assumptions $\sigma^2 \leq 1/n$ and $n \gg d^2$ are satisfied,
Proposition~\ref{p:perturb-pseudorandom} implies that the matrix $B$ is $(\zeta \sigma^2, 10^{-9})$-pseudorandom with probability at least $0.9$.
As we choose $\sigma^2 = 10^4 \sqrt{d\Delta(U)}/(\zeta \kappa n)$ in Procedure~\ref{proc:path},
let 
\[\alpha=\frac{10^4 \sqrt{d\Delta(U)}}{\kappa n} 
\quad {\rm and} \quad \beta = 10^{-9} \quad
\implies \quad B \gtrsim_{\beta} \alpha.
\]
By Proposition~\ref{p:Delta}, we have $\Delta(U) \geq \Delta(W)/200$.
Therefore, assuming this event happens, $\alpha$ satisfies the assumption in Theorem~\ref{t:pseudorandom} as 
\[\alpha = \frac{10^4 \sqrt{d\Delta(U)}}{\kappa n} 
\geq \frac{80\sqrt{d \Delta(W)}}{\kappa n}
\geq \frac{80\sqrt{d \Delta(B)}}{\kappa n}.
\]
Also, $\beta$ satisfies the assumption in Theorem~\ref{t:pseudorandom}, and $\Delta(B)$ satisfies the assumption in Theorem~\ref{t:pseudorandom} as $\Delta(B) \leq \Delta(W) \leq 200\Delta(U) \leq 200\Delta \ll 1/10$, and $s(B)$ satisfies the assumption in Theorem~\ref{t:pseudorandom} as $s(B)=s(W)=d$.
It remains to verify the assumption that $c_j(B) = d/n$ for all $1 \leq j \leq n$, for which we will use the property that $\norm{w_i}_2^2 = d/n$ for $1 \leq i \leq n$ guaranteed by the perturbation process as stated in Step~\ref{step:size-norm}.
Using the reduction from frame to operator in Definition~\ref{d:reduction} where $W_l$ is the $d \times n$ matrix with the $l$-th column being $w_l$ and all other columns being zero,
we have $\sum_{l=1}^n W_l^T W_l = \frac{d}{n} I_n$.
Therefore, by the furthermore part of Proposition~\ref{p:reduction-capacity}, we have that $c_j(B) = d/n$ for $1 \leq j \leq n$.

So, all the assumptions of Theorem~\ref{t:pseudorandom} are satisfied, by substituting $(A^{(0)})^{\circ 2} := B$. 
This substitution is justified as $B$ is a non-negative matrix, and thus is the Hadamard product square of some matrix (e.g. where each entry is the square root of the corresponding entry in $B$).
Therefore, Theorem~\ref{t:pseudorandom} concludes that 
\[
\capa(B) \geq s(B) - \frac{\Delta(B)}{5 \kappa \alpha n}.
\]
The properties from the reduction implies that
\[
\capa(W) \geq \capa(B) \geq s(B) - \frac{\Delta(B)}{5 \kappa \alpha n}
\geq s(W) - \frac{\Delta(W)}{5 \kappa \alpha n}
= s(W) - \frac{\Delta(W)}{50000 \sqrt{d\Delta(U)}},
\]
where the last equality is by our choice of $\alpha$.
By Proposition~\ref{p:Delta} that $\Delta(U) \geq \Delta(W)/200$,
we conclude that
\[
\capa(W) \geq s(W) - \frac{\sqrt{\Delta(W)}}{5000\sqrt{d}}.
\]
\end{proof}

\begin{remark}[Probability]
The conclusions of Proposition~\ref{p:movement}, Proposition~\ref{p:Delta} and Theorem~\ref{t:capacity} are all probabilistic.
By the union bound, however, there is a positive probability that all the conclusions hold.
This implies that there exists a perturbation that simultaneously satisfies all three lemmas, and we will proceed with such a perturbation.
\end{remark}

We next justify Step~\ref{step:movement2} in Procedure~\ref{proc:path}.
The argument is similar to the proof of Proposition~\ref{p:half}.
The current statement highlights more clearly the relation between the distance moved and the capacity lower bound.

\begin{proposition} \label{p:movement-size}
Let $W^{(0)} = \{w_1^{(0)}, \ldots, w_n^{(0)}\}$ be the input to the dynamical system in Definition~\ref{d:dynamical} with $w_j^{(0)} \in \R^d$ for $1 \leq j \leq n$ and
\[\capa(W^{(0)}) \geq s(W^{(0)}) - p(d,n,\Delta(W^{(0)})) {\rm~where~} p(d,n,\Delta(W^{(0)})) \text{ is a function of } d, n, \Delta(W^{(0)}).
\]
Let $T$ be the first time in the dynamical system when $\Delta(W^{(T)}) \leq \ol{\Delta}$ for some given $\ol{\Delta} < \Delta(W^{(0)})$.
Then 
\[
\dist(W^{(0)},W^{(T)}) \leq \frac{2 \Delta(W^{(0)}) \cdot p(d,n,\Delta(W^{(0)}))}{\ol{\Delta}}
\quad {\rm and} \quad 
s(W^{(T)}) \geq s(W^{(0)}) - p(d,n,\Delta(W^{(0)})).
\]
In particular, given the specific setting in Step~\ref{step:movement2} and Step~\ref{step:size} in Procedure~\ref{proc:path}, we have
\[
\dist(W^{(0)},W^{(T)}) \leq O\Big(\sqrt{\frac{\Delta(U^l)}{d}}\Big) 
\quad {\rm and } \quad
s(W^{(T)}) \geq d - O\Big(\sqrt{\frac{\Delta(U^l)}{d}}\Big).
\]
\end{proposition}
\begin{proof}
Using the capacity upper bound in Lemma~\ref{l:capUB}, the capacity lower bound in the assumption, and Lemma~\ref{l:cap-unchanged} that capacity is unchanged over time, we have the size lower bound
\[
s(W^{(T)}) \geq \capa(W^{(T)}) = \capa(W^{(0)}) \geq s(W^{(0)}) - p(d,n,\Delta(W^{(0)})).
\]
The assumption about $T$ is equivalent to $\Delta^{(t)} > \ol{\Delta}$ for $0 \leq t < T$.
It follows from Lemma~\ref{l:s'} that
\[
\d s(W^{(t)}) = -2\Delta(W^{(t)}) < -2\ol{\Delta} {\rm~for~all~} 0 \leq t < T.
\]
The size lower bound thus allows us to conclude that $T \leq p(d,n,\Delta(W^{(0)})/(2\ol{\Delta})$, as otherwise
\[
s(W^{(T)}) = s(W^{(0)}) + \int_0^T \d s(W^{(t)}) dt 
< s(W^{(0)}) - 2\ol{\Delta} T < s(W^{(0)}) - p(d,n,\Delta(W^{(0)})).
\]
Similar to the calculation in Proposition~\ref{p:half} and using the reduction from frame to operator in Definition~\ref{d:reduction},
\begin{eqnarray*}
\distance(W^{(0)},W^{(T)})
& \leq & \int_{0}^T \sqrt{\sum_{i=1}^n \norm{\d W_i^{(t)}}_F^2} dt
~=~ 2\int_0^T \sqrt{-\d \Delta(W^{(t)})} dt\\
& \leq & 2 \sqrt{ \int_0^T (-\d \Delta(W^{(t)})) dt \cdot \int_0^T 1 dt}
~\leq~ 2 \sqrt{\Delta(W^{(0)}) \cdot T}\\ 
& \leq & \sqrt{\frac{2 \Delta(W^{(0)}) \cdot p(d,n,\Delta(W^{(0)}))}{\ol{\Delta}}},
\end{eqnarray*}
where the first inequality is by Lemma~\ref{l:triangle},
the first equality is by Lemma~\ref{l:Delta'},
the second inequality is by Cauchy-Schwarz,
and the last inequality is by the bound on $T$ above.

In the specific setting in Step~\ref{step:movement2} and Step~\ref{step:size},
we have from Step~\ref{step:capacity} that
\[
\capa(W^{(0)}) \geq s(W^{(0)}) - O\Big(\sqrt{\frac{\Delta(W^{(0)})}{d}}\Big) {\rm~~so~that~~} p(d,n,\Delta(W^{(0)})) \leq O\Big(\sqrt{\frac{\Delta(W^{(0)})}{d}}\Big).
\]
As $\Delta(W^{(0)}) = \Delta(W^l) \leq O(\Delta(U^l))$ by Proposition~\ref{p:Delta}
and $s(W^{(0)}) = s(W^l) = d$ by Step~\ref{step:size-norm} in Procedure~\ref{proc:path},
we have the size lower bound
\[
s(W^{(T)}) 
\geq s(W^{(0)}) - p(d,n,\Delta(W^{(0)}))
\geq d - O\Big(\sqrt{\frac{\Delta(W^{(0)})}{d}}\Big)
\geq d - O\big(\sqrt{\frac{\Delta(U^l)}{d}}\big).
\]
The target $\ol{\Delta}$ in Step~\ref{step:dynamical} is 
\[\ol{\Delta}= \frac{\Delta}{3 \cdot 2^l} 
\quad {\rm and} \quad
\Delta(W^{(0)}) = \Delta(W^l) \leq O(\Delta(U^l)) \leq O(\frac{\Delta}{2^{l+1}}),
\] 
where the first inequality is by Proposition~\ref{p:Delta} and the second inequality is by Step~\ref{step:rescale} in the previous iteration.
Therefore, $\Delta(W^{(0)})/\ol{\Delta} = O(1)$ and this implies that
\[
\dist(W^{(0)},W^{(T)}) \leq 
O\Big(\frac{\Delta(W^{(0)}) \cdot p(d,n,\Delta(W^{(0)}))}{\ol{\Delta}}\Big)
\leq O\Big(\sqrt{\frac{\Delta(W^{(0)})}{d}}\Big) 
\leq O\Big(\sqrt{\frac{\Delta(U^l)}{d}}\Big),
\]
where the last inequality is by Proposition~\ref{p:Delta}.
\end{proof}

Finally, we justify Step~\ref{step:rescale}, which is similar to the proof in Lemma~\ref{l:postprocess}.

\begin{lemma} \label{l:rescale}
Let $W=\{w_1,\ldots,w_n\}$ where $w_i \in \R^d$ for $1 \leq i \leq n$ and $s:=s(W)$.
Let $U=\{u_1,\ldots,u_n\}$ where $u_i = \sqrt{\frac{d}{s}} \cdot w_i$ for $1 \leq i \leq n$.
Then 
\[s(U)=d \quad {\rm and} \quad 
\dist(U,W) = (\sqrt{d}-\sqrt{s})^2
\quad {\rm and} \quad
\Delta(U) =  \frac{d^2}{s^2} \Delta(W).
\]
In particular, given the specific setting in Step~\ref{step:rescale}, we have
\[
\dist(W^{(T)},U^{l+1}) \leq O\big(\sqrt{\frac{\Delta(U^l)}{d}}\big)
\quad {\rm and} \quad
\Delta(U^{l+1}) \leq \frac{\Delta}{2^{l+1}}.
\]
\end{lemma}
\begin{proof}
It is clear from Definition~\ref{d:frame-parameters} that 
\[s(U) = \sum_{i=1}^n \norm{u_i}_2^2 = \frac{d}{s} \sum_{i=1}^n \norm{w_i}_2^2 = \frac{d}{s} \cdot s(W) = d.
\]
The squared distance between $U$ and $W$ is
\[
\dist(U, W) 
= \sum_{i=1}^n \norm{u_i-w_i}_2^2
= (\sqrt{\frac{d}{s}}-1)^2 \sum_{i=1}^n \norm{w_i}_2^2
= (\sqrt{\frac{d}{s}}-1)^2 \cdot s
= (\sqrt{d}-\sqrt{s})^2.
\]
It follows from Definition~\ref{d:frame-parameters} that 
\[\Delta(c W) = c^4 \Delta(W) \quad \implies \quad
\Delta(U) = \Delta( \sqrt{\frac{d}{s}} W) = \frac{d^2}{s^2} \Delta(W).
\]

In Step~\ref{step:rescale}, we have the size lower bound $s:=s(W^{(T)}) \geq d - O(\sqrt{\frac{\Delta(U^l)}{d}})$, and we rescale $W^{(T)}$ to $U^{l+1}$ such that $s(U^{l+1})=d$.
Therefore,
\[
\dist(U^{l+1},W^{(T)}) \leq \Big(\sqrt{d}-\sqrt{d-O\big(\sqrt{\frac{\Delta(U^l)}{d}}\big)}~\Big)^2
= d\Big(1-\sqrt{1-O\big(\sqrt{\frac{\Delta(U^l)}{d^3}}\big)}~\Big)^2
\leq O\big(\sqrt{\frac{\Delta(U^l)}{d}}\big),
\]
where the last inequality follows from the inequality $1-\sqrt{1-x}\leq x$ for $0 \leq x \leq 1$ and the assumption that $\Delta(U^l) \leq \Delta \ll 1$.
Finally, 
\[\frac{d^2}{s^2} 
\leq \frac{d^2}{\big(d-O\big(\sqrt{\frac{\Delta(U^l)}{d}}\big)\big)^2} 
= \Big(1-O\big(\sqrt{\frac{\Delta(U^l)}{d^3}}\big)\Big)^{-2}
\leq \Big(1 + O\big(\sqrt{\frac{\Delta(U^l)}{d^3}}\big)\Big)^2
\leq \frac{3}{2},
\]
where the second inequality is by $(1-x)^{-2} \leq 1+2x$ for $0 \leq x \leq 1/2$
and $O\big(\sqrt{\Delta(U^l)/d^3}\big) \ll 1/2$ by our assumption that $\Delta$ is small enough in Step~\ref{step:global} of Procedure~\ref{proc:path},
and the last inequality is also by the same assumption.
(To be precise, we should trace out the constants in the big-$O$ notation in Proposition~\ref{p:movement-size}.
In the last line of Theorem~\ref{t:capacity}, we see that the hidden constant in the capacity lower bound is reasonable. 
Hence, the hidden constants in Proposition~\ref{p:movement-size} are also reasonable, as they only depend on the capacity lower bound.
On the other hand, the constant in our assumption on $\Delta$ in Step~\ref{step:global} of Procedure~\ref{proc:path} is much smaller.)
Therefore, 
\[\Delta(U^{l+1}) \leq \frac{3}{2} \Delta(W^{(T)}) \leq \frac{3}{2} \cdot \frac{\Delta}{3 \cdot 2^l} \leq \frac{\Delta}{2^{l+1}},
\]
where the second inequality is by Step~\ref{step:dynamical} in Procedure~\ref{proc:path}.
\end{proof}

We have justified all steps in Procedure~\ref{proc:path},
and so we have the conclusions from Theorem~\ref{t:one-iteration}, 
Theorem~\ref{t:d^3 eps} and Theorem~\ref{t:d^7 eps}.

\section{Conclusions and Discussions}

We have proved that the bound in the Paulsen problem is independent of the number of vectors, and through the reduction in~\cite{projection} we have also proved the projection conjecture with the same bound.
We hope that our results and techniques will find applications in other problems,
in particular the dynamical system and the new method in proving capacity lower bound.

The following are some discussions about improving the results.

\begin{enumerate}
\item The $O(m^2 n \eps)$ bound in Theorem~\ref{t:Paulsen} holds in the more general operator setting.  
The current smoothed analysis only works in the frame setting.
It would be very nice if the smoothed analysis can also be extended to the operator setting (e.g. this is related to the Brascamp-Lieb constants).

\item There is a gap between the $O(d^{13/2} \eps)$ upper bound and the $\Omega(d \eps)$ lower bound.
We believe that the correct answer is $\Theta(d \eps)$.
There are several bottlenecks in improving the current proof.
One interesting intermediate step is to prove the $O(d \eps)$ bound in the case when $n$ is large enough and $\eps$ is small enough, for which we proved the bound $O(d^{5/2} \eps)$, where the bottleneck is in the perturbation process. 

\item Our proof shows the existence of an equal norm Parseval frame which is close to the input frame, but it does not provide an efficient algorithm to output such a frame.
It would be very nice to find a polynomial time algorithm to output such a frame, with the running time depends only on $\log(1/\eps)$.

\item The current approach relies heavily on the operator capacity lower bound, which is used to argue indirectly that $\Delta$ converges to zero in order to prove the squared distance bound.
To prove the operator capacity lower bound, we reduce to proving matrix capacity lower bound, which we prove by establishing a lower bound on the convergence rate $-\d \Delta$. 
This approach is rather indirect.

Consider the simpler matrix Paulsen problem.
Since we can directly lower bound $-\d \Delta$, it follows that $\Delta$ converges to zero quickly, and this implies the squared distance bound without using the concept of matrix capacity.

This naturally leads to the question whether we can directly analyze the perturbation process and prove a lower bound on $-\d \Delta$ in the operator/frame case.  If so, this will likely significantly simplifies and improves the current analysis.
The technical challenge is to identify the correct pseudorandom property in the operator/frame setting.
\end{enumerate}

\subsection*{Acknowledgement}

We thank Nikhil Srivastava for suggesting the problem and Vern Paulsen for telling us about the history and motivations of the problem.
We also thank Nick Harvey, Mohit Singh and Avi Wigderson for useful comments.

\bibliographystyle{plain}

\begin{appendix}

\section{List of Definitions}

\vspace*{-10mm}

\listoftheorems[ignoreall,show={definition}]

\section{Tight Example for Matrix Capacity Lower Bound} \label{a:tight}

We complete the details in Lemma~\ref{l:tight}.
Let $E = x k (k + 1)$ be the sum of the upper right submatrix and $F = y (k - 1) k$ be the sum of the lower left submatrix.
Then $E + F = s(A) = 1$.
Now we have
\[
r_i(A)
=
\left\{
\begin{array}{ll}
E / k & \text{for $i \le k$,} \\
F / (k - 1) & \text{for $i > k$.} \\
\end{array}
\right.
\]
And similarly,
\[
c_j(A)
=
\left\{
\begin{array}{ll}
F / k & \text{for $i \le k$,} \\
E / (k + 1) & \text{for $i > k$.} \\
\end{array}
\right.
\]
So
\begin{align*}
\Delta(A)
& = \frac1{2k - 1} \sum_{i = 1}^{2k - 1} (s - (2k - 1) r_i)^2 + \frac1{2k + 1} \sum_{j = 1}^{2k + 1} (s - (2k + 1) c_j)^2 \\
& = \frac1{2k - 1} \left( k (1 - \frac{(2k - 1) E}{k})^2 + (k - 1) (1 - \frac{(2k - 1)F}{k - 1})^2 \right) \\
& \hspace{1in}  + \frac1{2k + 1} \left( k (1 - \frac{(2k + 1) F}{k})^2 + (k + 1) (1 - \frac{(2k + 1) E}{k + 1})^2 \right) \\
& = \frac1{2k - 1} \left( k - 2(2k - 1) E + \frac{(2k - 1)^2 E^2}{k} + (k - 1) - 2(2k - 1) F + \frac{(2k - 1)^2 F^2}{k - 1} \right) \\
& \hspace{1in}  + \frac1{2k + 1} \left( k - 2(2k + 1) F + \frac{(2k + 1)^2 F^2}{k} + (k + 1) - 2(2k + 1) E + \frac{(2k + 1)^2 E^2}{k + 1} \right) \\
& = 1 - 2(E + F) + (2k - 1) (\frac{E^2}{k} + \frac{F^2}{k - 1}) \\
& \hspace{1in}  + 1 - 2(F + E) + (2k + 1) (\frac{F^2}k + \frac{E^2}{k + 1}) \\
& = -2 + (\frac{2k - 1}{k} + \frac{2k + 1}{k + 1}) E^2 + (\frac{2k - 1}{k - 1} + \frac{2k + 1}{k}) F^2 \\
& = -2 + \frac{(2k^2 + k - 1) + (2k^2 + k)}{k (k + 1)} E^2 + \frac{(2k^2 - k) + (2k^2 - k - 1)}{(k - 1) k} F^2 \\
& = -2 + \frac{4k^2 + 2k - 1}{k (k + 1)} E^2 + \frac{4k^2 - 2k - 1}{(k - 1) k} F^2.
\end{align*}
Now given $E + F = 1$, for positive $a$ and $b$, $a E^2 + b F^2$ attains its minimum $\opt = (a^{-1} + b^{-1})^{-1}$ when $E = \opt / a$ and $F = \opt / b$.
Taking corresponding value for $x$ and $y$, we have
\begin{align*}
\Delta(A)
& = -2 + (\frac{k (k + 1)}{4k^2 + 2k - 1} + \frac{(k - 1) k}{4k^2 - 2k - 1})^{-1} \\
& = -2 + (k \frac{(4k^3 + 2k^2 - 3k - 1) + (4k^3 - 2k^2 - 3k + 1)}{16k^4 - 12k^2 + 1})^{-1} \\
& = -2 + \frac{16k^4 - 12k^2 + 1}{8k^4 - 6k^2} \\
& = \frac{1}{8k^4 - 6k^2}.
\end{align*}
Since $m^2 n^2 = 16k^4 - 8k^2 + 1$, we have
\[
\Delta
= (2 + o(1)) \frac{1}{m^2 n^2}.
\]

\end{appendix}

\end{document}